\documentclass[11pt]{article}

\usepackage[utf8]{inputenc}
\usepackage{lipsum}
\usepackage[table,xcdraw]{xcolor}
\usepackage{amsmath, amssymb,amsfonts,amsthm,mathtools,nicefrac}
\usepackage{xspace,graphicx,relsize,bm}
\usepackage{soul} %

\usepackage{parskip}  %
\usepackage{comment}
\usepackage{libertinus}  %
\usepackage[zerostyle=a,scaled=.97]{newtxtt}  %
\usepackage[T1]{fontenc}
\usepackage[libertine]{newtxmath}  %
\usepackage{dsfont}  %
\usepackage[cal=stix, calscaled=.97]{mathalpha}  %
\usepackage{anyfontsize}

\usepackage{url}
\usepackage[
    backend=biber,
    style=alphabetic,
    sorting=nyt,
    backref=true,
    maxbibnames=99
    ]{biblatex}
\DefineBibliographyStrings{english}{%
  backrefpage = {page},%
  backrefpages = {pages},%
  mathesis = {MSc thesis},%
}

\usepackage[linktocpage=true]{hyperref}
\usepackage[margin=1in]{geometry}
\usepackage[singlespacing]{setspace}
\definecolor{linkcol}{rgb}{0.0,0.55,0.7}
\definecolor{citecol}{rgb}{0.0, 0.6, 0.45}
\definecolor{urlcol}{rgb}{0.7, 0.0, 0.55}
\hypersetup{
	colorlinks,
	linkcolor={linkcol},
	citecolor={citecol},
	urlcolor={urlcol}
}

\usepackage{subcaption}
\usepackage{mleftright}
\usepackage{tipa}
\usepackage{multirow}
\usepackage{physics}  %
\usepackage{authblk}  %

\usepackage{subfiles} %

\usepackage{algorithm}%
\usepackage{algpseudocode}%

\renewcommand{\algorithmicrequire}{\textbf{Input:}}
\renewcommand{\algorithmicensure}{\textbf{Output:}}

\usepackage{multirow}

\usepackage{pifont}

\usepackage{float}

\usepackage{chngpage}

\usepackage[capitalize,noabbrev]{cleveref}  %

\usepackage{longtable}
\usepackage{lscape}

\usepackage{algorithm}
\usepackage{algpseudocode}
\usepackage{complexity}

\def\01{\{0,1\}}
\newcommand{\mc}[1]{\mathcal{#1}}

\DeclareMathOperator*{\Ex}{\mathbf{E}}
\let\Pr\relax
\DeclareMathOperator*{\Pr}{\mathbf{Pr}}

\DeclareMathOperator{\GL}{GL}
\DeclareMathOperator{\Z}{\mathbb{Z}}
\DeclareMathOperator*{\argmax}{arg\,max}

\renewcommand{\O}{\ensuremath{\mathcal{O}}}

\renewcommand{\eqref}[1]{Eq.~(\ref{#1})}

\newcommand{\U}{\ensuremath{\mathrm{U}}}
\newcommand{\Cl}{\ensuremath{\mathrm{Cl}}}

\newcommand{\tv}{\mathrm{d}_\mathrm{TV}}

\newcommand{\td}{\mathrm{d}_\mathrm{TD}}

\newcommand{\lr}[1]{\left(#1\right)}
\newcommand{\lrq}[1]{\left[#1\right]}
\newcommand{\lrb}[1]{\left\{#1\right\}}
\newcommand{\lra}[1]{\left |#1\right |}
\newcommand{\lraa}[1]{\left |\left|#1\right |\right|}

\DeclarePairedDelimiter\floor{\lfloor}{\rfloor}

\renewcommand{\norm}[1]{\lVert{#1}\rVert}

\DeclarePairedDelimiter\lrangle{\langle}{\rangle}
\makeatletter
\let\oldlrangle\lrangle
\def\lrangle{\@ifstar{\oldlrangle}{\oldlrangle*}}
\makeatother

\allowdisplaybreaks
\newtheoremstyle{mydefinitionsty}%
{10pt}%
{10pt}%
{}%
{}%
{}%
{}%
{.5em}%
{\textbf{\thmname{#1}~\thmnumber{#2}:  }\thmnote{(#3)}}%
\theoremstyle{mydefinitionsty}
\newtheorem{definition}{Definition}
\newtheorem{remark}{Remark}

\newtheorem{fact}{Fact}

\newtheoremstyle{mythmsty}%
{10pt}%
{10pt}%
{\itshape}%
{}%
{}%
{}%
{.5em}%
{\textbf{\thmname{#1}~\thmnumber{#2}:  }\thmnote{(#3)}}%
\theoremstyle{mythmsty}

\newtheorem{theorem}{Theorem}
\newtheorem{lemma}{Lemma}

\newtheorem{corollary}{Corollary}
\newtheorem{claim}{Claim}

\newtheorem{proposition}{Proposition}

\definecolor{nbcolor}{rgb}{0.2, 0.5, 0.5}

 \title{Cloning is as Hard as Learning for Stabilizer States}

\author[1]{Nikhil Bansal\footnote{\href{mailto:nikhil.bansal@warwick.ac.uk}
{nikhil.bansal@warwick.ac.uk}}}
\author[1]{Matthias C. Caro\footnote{\href{mailto:matthias.caro@warwick.ac.uk}
{matthias.caro@warwick.ac.uk}}}
\author[2,3]{Gaurav Mahajan\footnote{\href{mailto:gaurav.mahajan@mistral.ai}
{gaurav.mahajan@mistral.ai}}}

\affil[1]{Department of Computer Science, University of Warwick, Coventry, UK}
\affil[2]{Mistral AI, Palo Alto, California, USA}
\affil[3]{Toulouse School of Economics, Toulouse, France}

\date{}

\begin{document}
\maketitle

\begin{abstract}
    
    The impossibility of simultaneously cloning non-orthogonal states lies at the foundations of quantum theory.
    Even when allowing for approximation errors, cloning an arbitrary unknown pure state requires as many initial copies as needed to fully learn the state.
    
    Rather than arbitrary unknown states, modern quantum learning theory often considers structured classes of states and exploits such structure to develop learning algorithms that outperform general-state tomography.
    This raises the question: How do the sample complexities of learning and cloning relate for such structured classes?

    We answer this question for an important class of states. Namely, for $n$-qubit stabilizer states, we show that the optimal sample complexity of cloning is $\Theta(n)$. 
    Thus, also for this structured class of states, cloning is as hard as learning.
    To prove this result, we relate stabilizer state cloning to a variant of the sample amplification problem for probability distributions that was recently introduced in classical learning theory.
    This allows us to obtain our cloning lower bounds in analogy to new sample amplification lower bounds for classes of distributions with an underlying linear structure. 
    Our results provide a more fine-grained perspective on No-Cloning theorems, opening up connections from foundations to quantum learning theory and quantum cryptography.
\end{abstract}

\newpage
\tableofcontents
\newpage

\section{Introduction}
The \textit{No-Cloning} theorem~\cite{wootters1982single, DIEKS1982271} is a foundational result in quantum theory. %
It states: There is no \textit{one-fits-all} cloning unitary $U$ such that $ U (\ket{\psi}\ket{0}) = \ket{\psi}\ket{\psi}$ holds simultaneously for non-orthogonal pure states $\ket{\psi}$. 
We also have what one may call an \textit{approximate No-Cloning} theorem~\cite{Werner_1998, Keyl_1999}:
Approximately cloning arbitrary pure states requires as many copies as pure state tomography~\cite{gross2010, baldwin2016, haah2016sample-optimal, odonnell2016efficient, pelecanos2025debiasedkeylsalgorithmnew}. %
While this tells us cloning general pure states is exponentially costly, imposing sufficiently stringent structure can make cloning easy. 
As a trivial example, any class of orthogonal states can be cloned (and learned) perfectly from a single copy. 
Thus, there are two extremes: Approximate cloning and learning are exponentially costly without structure, but even exact cloning and learning become trivial with orthogonality structure. 

In many areas of quantum computation and information, and in particular in quantum learning theory~\cite{arunachalam2017surveyquantumlearningtheory, anshu2023surveycomplexitylearningquantum}, we often consider neither of these extremes. 
Instead, we focus on classes of states that are structured yet sufficiently rich to exhibit non-trivial quantum behavior. 
Prominent examples include stabilizer states~\cite{montanaro2017learningstabilizerstatesbell}, tensor network states~\cite{landoncardinal2010efficientdirecttomographymatrix, cramer2010efficient, landau2013polynomialtimealgorithmgroundstate, Arad_2017, Bakshi_2025}, quantum example states~\cite{bshouty1995learning-dnf, arunachalam2017surveyquantumlearningtheory, caro2024classical-verification, caro2025testingclassicalpropertiesquantum}, and phase states~\cite{arunachalam2023optimalalgorithmslearningquantum}. 
However, whereas questions of learning such structured classes of states are well-studied, questions of (approximate) cloning under structural assumptions remain unexplored. This leads us to our first central question: 

\begin{center}
    \textit{Can approximate cloning be easier than learning for structured classes of states?}
\end{center}

In the asymptotic regime, as the number of cloned copies required goes to infinity, \cite{chiribella2014optimal} gave strong evidence that the answer to the above question should be negative.
In contrast, a computational version of the question was raised in~\cite{fefferman2025hardnesslearningquantumcircuits}, and there is evidence~\cite{nehoran2024, bostanci2025duality, cakan2026publickeyquantumkeyfireclassical} towards a positive answer. %
We, however, interpret the question in terms of sample complexity: Are there structured classes of states for which the sample complexity of approximate cloning is strictly smaller than that of learning?

As classical information can be copied, ``classical cloning'' does not provide any insights useful for the quantum case. However, classical learning theorists have recently introduced the task of \textit{sample amplification}~\cite{axelrod2019sampleamplificationincreasingdataset, axelrod2024statisticalcomplexitysampleamplification}: Given a dataset of random samples drawn from an unknown distribution from some known class, produce a larger dataset that looks like it consists of true samples. 
Notably, \cite{axelrod2019sampleamplificationincreasingdataset} proved that sample amplification can be strictly easier than learning.
When thinking of states loosely as a quantum counterpart of classical probability distributions, we see that the task of sample amplification is conceptually similar to approximate quantum cloning. %
We posit that, via this analogy, new results in sample amplification can lead to novel insights into quantum cloning. 

To this end, we connect sample amplification with a prominent area in computational learning theory, the learnability of Boolean function classes.
Here, given i.i.d.~samples of the form  $(x, f(x))$, with the input $x$ drawn from some distribution $\mc V$, and with the unknown function $f$ promised to lie in some known class $\mc F$, we aim to learn $f$. 
If $\mc V$ is known, this can be viewed as learning the distribution $(\mc V, f)$ from samples. 
Hence, from the perspective of computational learning theory, the following question arises naturally: 
\begin{center}
    \textit{Can sample amplification for structured classes of distributions $\{(\mc V, f)\}_{f\in\mc F}$ be easier than learning?}
\end{center}

In this work, we study the above two questions for specific kinds of structure. 
First, we exhibit function classes for which sample amplification is as hard as learning in terms of sample complexity. Concretely, we show that amplifying parities is as hard as learning them, and we provide a coding-theoretic interpretation of this result. 
Second, we draw on our classical no-go result for amplification to show a no-go for structured quantum cloning: Stabilizer states are no easier to approximately clone than to learn in terms of sample complexity. This serves as an approximate No-Cloning theorem for stabilizer states and is, to our knowledge, the first No-Cloning theorem for a practically relevant class of structured quantum states.

\subsection{Framework}

Our first contribution is to propose structured variants of sample amplification %
and quantum cloning. %

\paragraph{Structured sample amplification.} 
Informally, sample amplification as introduced by \cite{axelrod2019sampleamplificationincreasingdataset, axelrod2024statisticalcomplexitysampleamplification} is the following task: Let $\mc D$ be a known class of distributions. Given $t$ samples drawn i.i.d.~from an unknown distribution $D \in \mc D$, ``amplify'' them to $t+m$ samples that are indistinguishable from $t+m$ true i.i.d.~samples drawn from $D$. 
Here, it is natural to formalise indistinguishability in terms of small total variation (TV) distance. Thus, a distribution class $\mc D$ over domain $\mc X$ admits a $(t, t+m, \epsilon)$ sample amplification scheme if there exists a (randomised) map $T_{\mc{D}, t, m, \epsilon}: \mc{X}^{t} \rightarrow \mc{X}^{t+m}$ such that,
\begin{equation}
        \sup_{D \in \mc{D}} \tv\lr{D^{\otimes t} \circ T^{-1}_{\mathcal{D}, t, m, \epsilon} , D^{\otimes t+m}} \leq \epsilon\, ,
\end{equation} 
where $\tv$ denotes the total variation distance, and where we use $D^{\otimes t} \circ T^{-1}_{\mathcal{D}, t, m, \epsilon}$ to denote the pushforward measure of $D^{\otimes t}$ under $T_{\mathcal{D}, t, m, \epsilon}$.
As already highlighted in \cite{axelrod2019sampleamplificationincreasingdataset}, sample amplification can be viewed as a game between an amplifier and a distinguisher. The amplifier tries to amplify $t$ samples from an unknown distribution $ D$ to $t+m$ samples. The distinguisher, who knows the distribution $ D$ but not the specific samples available to the amplifier, receives a set of $t+m$ samples and aims to distinguish whether it is true i.i.d.~data or an output produced by the amplifier. 
The notion of distance originates as the advantage over random guessing that the distinguisher can achieve. That is, equivalently to the above, a distribution class $\mc{D}$ admits a $(t, t+m, \epsilon)$ sample amplification procedure if there exists a (randomised) map $T_{\mc{D}, t, m, \epsilon}: \mc{X}^{t} \rightarrow \mc{X}^{t+m}$ such that,
\begin{equation}\label{eq:intro-sample-amplification-distinguisher}
    \sup_{D \in \mc D}\sup_{\mc A} \lra{\Pr_{X^{t+m} \leftarrow D^{\otimes t+m}}\lrq{\mc{A}(X^{t+m}) = 1} - \Pr_{X^t \leftarrow D^{\otimes t}}\lrq{\mc{A}(T_{\mc{D}, t, m, \epsilon}(X^t)) = 1}} \leq \epsilon,
\end{equation}
where the second supremum is over all possible distinguishers $\mc A$. 

\cite{axelrod2019sampleamplificationincreasingdataset, axelrod2024statisticalcomplexitysampleamplification} studied sample amplification for different classes of discrete and continuous distributions. For discrete distributions over a bounded support of size $k$, %
they gave a $(t, t+\Theta(t\epsilon/\sqrt{k}), \epsilon)$ sample amplification scheme whenever $t=\Omega(\sqrt{k}/\epsilon)$. 
This in particular implies that we can amplify by one sample to constant accuracy already from $\mc O (\sqrt{k})$ samples, which constitutes a quadratic improvement over the sample complexity of learning this class of distributions to constant accuracy. %
So, for (unstructured) discrete distributions of bounded support, sample amplification is strictly easier than learning. 

While \cite{axelrod2019sampleamplificationincreasingdataset, axelrod2024statisticalcomplexitysampleamplification} were motivated by distribution learning, we propose a variant of sample amplification that is inspired by the standard setting of computational probably approximately correct (PAC) learning theory \cite{valiant1984theory}: 
Given i.i.d.~samples $\{(x_i, f(x_i))\}_i$ of some unknown function $f:\{0, 1\}^n \rightarrow \{0, 1\}$ from some known class $\mc F$, output a hypothesis $\hat{f}$ s.t.~$\Pr_x [f(x)\neq \hat{f}(x)]$ is small with high success probability. 
Here, we focus on the special case in which the inputs
$\{x_i\}_i$ are i.i.d.~uniformly random $n$-bit strings,
so that the learning task can equivalently be formulated as that of learning an unknown distribution from a known class $\mc D_{\mc F} = \{(\mc U_n, f)\}_{f \in \mc F}$ from samples to good approximation in total variation distance. %
We thus propose the following structured version of sample amplification for distributions described by Boolean functions: 

\begin{definition}[Sample amplification of functions -- Informal]
    A class $\mc F$ of Boolean functions is said to admit a \emph{$(t, t+m, \epsilon)$-sample amplification scheme} if the class $\mc D_{\mc F} = \{(\mc U_n, f)\}_{f \in \mc F}$ of distributions does.
\end{definition}
We will study sample amplification for different Boolean function classes. This will then serve as a tool for understanding approximate quantum cloning of structured classes of states.

\paragraph{Structured quantum cloning.} 

Approximate quantum cloning\footnote{From here on, we often omit the ``approximate.'' We always consider \emph{approximate} cloning unless specified otherwise.}~\cite{Werner_1998} is the task of mapping $t$ many i.i.d.~copies of an $n$-qubit state to a $((t+m)\cdot n)$-qubit state that is close to $t+m$ i.i.d.~copies of the original state. Natural notions of closeness to consider are fidelity and trace distance.
\cite{Werner_1998} studied pure-state cloning with respect to global fidelity, so that the figure of merit for a linear completely positive and trace-preserving (CPTP) cloning map $\Lambda:\mc{B}((\mathbb{C}^{2^n})^{\otimes t}) \rightarrow \mc{B}((\mathbb{C}^{2^n})^{\otimes t+m})$ is given as
\begin{equation}
    \mathfrak{F} (\Lambda) = \inf_{|\psi\rangle} \tr\lr{| \psi\rangle \langle \psi|^{\otimes t+m}\,  \Lambda\lr{| \psi\rangle \langle \psi|^{\otimes t}}}.
\end{equation}
\cite{Werner_1998} showed that the optimal cloning fidelity for pure state cloning from $t$ copies to $t+m$ copies equals $\frac{d[t]}{d[t+m]}$, where $d[t] = {{2^n + t-1} \choose t}$. This is achieved by projecting $|\psi\rangle^{\otimes t}\otimes I_{2^{nm}}/{2^{mn}}$ onto the symmetric subspace of $(\mathbb{C}^{2^n})^{\otimes t+m}$. 
In particular, the optimal achievable fidelity for cloning one extra copy ($m=1$) is $\frac{t+1}{t+2^n}$. Consequently, the sample complexity necessary and sufficient to achieve $1-\epsilon$ fidelity when cloning one extra copy is $t = \Theta\lr{2^n/\epsilon}$. This matches that of optimal pure state tomography with respect to infidelity~\cite{Bru__1999, gross2010, baldwin2016, haah2016sample-optimal, odonnell2016efficient, pelecanos2025debiasedkeylsalgorithmnew}. 
While exact cloning is known to become impossible as soon as we allow for any two non-orthogonal states, these bounds for approximate cloning %
are derived for the case of an arbitrary unknown state; the analysis of approximate cloning in \cite{Werner_1998, Keyl_1999} relies on the representation theory of the \textit{full} unitary group and in particular on Schur-Weyl duality. Thus, it does not carry over to structured classes of quantum states.

Quantum learning theory has recently focused on learning structured classes of states such as stabilizer states~\cite{montanaro2017learningstabilizerstatesbell}, states prepared using few non-Clifford gates~\cite{Leone2024learningtdoped, Grewal_2025}, matrix-product states~\cite{landoncardinal2010efficientdirecttomographymatrix, cramer2010efficient, landau2013polynomialtimealgorithmgroundstate, Arad_2017, Bakshi_2025}, states with bounded gate complexity~\cite{zhao2024}, output states of shallow circuits~\cite{Huang_2024, vasconcelos2025learning, Landau_2025}, Gibbs states~\cite{anshu2021sample, Rouze2024learningquantummany, haah2024learning, bakshi2024learning, chen2025learningquantumgibbsstates, bluhm2025certifyinglearningquantumising, arunachalam2025testing}, and more. This motivates us to connect the modern perspective of quantum learning theory with the foundational question of cloning by proposing a structured version of approximate cloning:

\begin{definition}[Cloning of structured states -- Informal]
    A class $\mc S$ of quantum states admits a $(t, t+m, \epsilon)$-quantum cloning scheme if there exists a linear CPTP map $\Lambda_{\mc{S}, t, m, \epsilon}:\mc{B}((\mathbb{C}^{2^n})^{\otimes t}) \rightarrow \mc{B}((\mathbb{C}^{2^n})^{\otimes t+m})$ s.t.
    \begin{equation}
        \sup_{\rho \in \mc{S}} \td  \lr{\Lambda_{\mc{S}, t, m, \epsilon}(\rho^{\otimes t}), \rho^{\otimes t+m}} \leq \epsilon \, ,
    \end{equation}
    where $\td$ denotes the trace distance.
    We call $\epsilon$ the error incurred by the cloning scheme.
\end{definition}

We will study structured cloning of stabilizer states. 
While adapting the reasoning of \cite{Werner_1998} to stabilizer state cloning may be achievable via the recently developed representation theory of the Clifford group and the Clifford commutant~\cite{Gross_2021, bittel2025completetheorycliffordcommutant}, we take an alternative path based on a learning-theoretic distinguisher
and elementary linear-algebraic arguments.

\subsection{Overview of the Main Results} 

With the framework established, we can now study the two main questions highlighted above: First, is structured sample amplification easier than learning? And second, is structured quantum cloning easier than learning? In this work, we show that the two questions are related for specific structured classes of functions and states, and for those classes we give negative answers to both questions. Our results on sample amplification and cloning sample complexity lower bounds for different structured classes are highlighted and compared with learning sample complexities in Table~\ref{table:overview}.

\begin{table}[]
\centering
\makebox[\linewidth][c]{
\begin{tabular}{|c|c|c|}
\hline
\textbf{Problem Instance} & \textbf{Learning} & \textbf{Sample Amplification/Cloning}  \\
\textbf{(Classical/Quantum)} & \textbf{(with small constant error)} & \textbf{(with small constant error)} \\
\hline
$n$-bit parities   & $\Theta(n)$ & $\Omega(n)$ (Corollary~\ref{corollary:parities-sa-lb}) \\
\hline
$k_{\mc F}$-dimensional & $O(k_{\mc F})$ & $\Omega(k_{\mc F})^{*}$ (Theorem~\ref{thm:coding-theory-general})  \\
Boolean function subspace &  & \\
\hline
$n$-qubit stabilizer states  & $\Theta(n)$~\cite{montanaro2017learningstabilizerstatesbell}  & $\Omega(n)$ (Theorem~\ref{thm:stab-states-clone}) \\
\hline
\end{tabular}}
\caption{\textbf{Overview of our main results:} A comparison between known sample complexity bounds for learning (with small constant error) and our sample complexity lower bounds for amplification/cloning. Here, the $^\ast$ in the second row indicates that this bound depends on additional properties of the linear code $C_\mathcal{F}$ associated with the function space $\mathcal{F}$ and its dual $C_\mathcal{F}^\perp$.}
\label{table:overview}
\end{table}

\paragraph{Structured sample amplification.} On the classical side, we develop a general lower bound on the optimal achievable error in structured sample amplification for any structured distribution class of the form $\mc{D} = \{(\mc U_n, f)\}_{f\in \mc F}$, with $\mc F$ a known Boolean function class.
We then instantiate this bound for parity functions, and more generally for function classes that form a linear space.

\begin{theorem}[Structured Sample Amplification Lower Bounds -- Informal]\label{thm:classical-results-informal}
    \hspace{1cm}
    
    \begin{enumerate}
        \item \textbf{General lower bound:} Let $\mc F$ be a ``symmetric\footnote{Here, we call a function class symmetric if all its elements have the same minimum teaching sequence size and if obtaining a teaching sequence from random samples is equally likely for each function in the class. See \Cref{cor:semi-general-lower-bound-symmetric} for details.}'' class of Boolean functions on $n$ bits. Let $b_{\mc F}$ be the teaching dimension \cite{GOLDMAN199520} of $\mc F$ (see Section~\ref{section:ssa} for a definition). 
        Sample amplification for $\mc F$ from $<b_{\mc F}$ samples incurs an error of at least half the probability that $b_{\mc F}$ uniformly random labeled examples uniquely identify the target function.
        \item \textbf{Lower bound for parities:} Sample amplification for the class $\mathcal{F}_{\mathrm{par}}=\{x\mapsto s\cdot x\mod 2\}_{s\in\mathbb{Z}_2^n}$ of $n$-bit parity functions from $<n$ samples incurs at least a constant error of $0.14$.
        \item \textbf{Lower bounds from coding theory:} Let $\mathcal{F}$ be some $k_{\mc F}$-dimensional linear subspace of all Boolean functions on $n$ bits.
        Sample amplification for $\mathcal{F}$ from $<k_{\mc F}$ samples incurs an error of at least half the probability that a certain corresponding linear code (see Section~\ref{section:ssa} for a definition) corrects $k_{\mc F}$ random erasure errors minus an $O(k_\mathcal{F}^2/2^n)$ term.
    \end{enumerate}
\end{theorem}

Whereas sample amplification is strictly easier than learning for the class of all Boolean functions as a consequence of results from \cite{axelrod2019sampleamplificationincreasingdataset, axelrod2024statisticalcomplexitysampleamplification} (see \Cref{sec:sa-all-functions}),
Theorem~\ref{thm:classical-results-informal} shows that structured sample amplification to small constant error is as hard as learning for particular structured classes of functions. 
Theorem~\ref{thm:classical-results-informal} is proved in Section~\ref{section:ssa}. Concretely, bullet point 1 is proved in \Cref{subsec:amplification-general-lower-bound}, bullet point 2 is proved in \Cref{subsec:amplification-parity-lower-bound}, and bullet point 3 is proved in \Cref{subsec:amplification-coding-theory}.

\paragraph{Structured quantum cloning.}
On the quantum side, we adapt the same learning-based distinguisher
principle underlying our structured sample amplification lower bounds
to derive cloning lower bounds for different families of quantum states. We first prove a general lower bound on the optimal achievable error for cloning a set of (mixed) states that ``hide'' unknown symmetry subgroups of a large (known) Abelian group.
Here, we say that a state $\rho$ ``hides'' a subgroup $H\leq G$ under some fixed unitary representation\footnote{For point 1 in \Cref{thm:quantum-results-informal} below, we assume the representation $\mu$ to contain every irrep of $G$ exactly once.} $\mu$ if the elements of $H$ leave the state invariant, i.e., $\tr\lr{\rho \mu(h)} = 1, \forall h \in H$, and the elements not in $H$ change the state significantly, i.e., $|\tr\lr{\rho \mu(g)}|$ is bounded away from $1$ for all $g \in G\setminus H$.
We then give our cloning lower bound for stabilizer states.

\begin{theorem}[Structured Quantum Cloning Lower Bounds -- Informal]\label{thm:quantum-results-informal}
\hspace{1cm}
    \begin{enumerate}
        \item \textbf{General lower bound}: 
        Let $G$ be a known Abelian group.
        Let $\mc{S}_{\alpha}$ be a class of (mixed) $n$-qubit states that ``hide'' unknown symmetry subgroups $H\leq G$ of order $\alpha$. 
        Let $p_H^t$ denote the probability of uniquely identifying $H$ from $t$ copies of a corresponding canonical $H$-hiding state $\sigma_H$ (see \Cref{subsec:techniques} below for a definition).
        Then, cloning $\mathcal{S}_{\alpha}$ from $<t$ copies incurs an error of at least $\frac{1}{2}\underset{\substack{H\leq G: |H|=\alpha}}{\Ex}\left[p^{t}_{H}-p^{t-1}_{H}\right]$, where the expectation is uniform over subgroups $H\leq G$ of order $\alpha$. 
        \item \textbf{Lower bound for stabilizer states}: Cloning for the class of pure $n$-qubit stabilizer states from $< \floor{n/3}$ copies incurs at least an error of $0.28$ minus a term exponentially small in $n$.
    \end{enumerate}
\end{theorem}

Approximate cloning is already known to be as hard as learning for the unstructured class of all pure states \cite{Werner_1998, Keyl_1999}.
Bullet point 2 in Theorem~\ref{thm:quantum-results-informal}, which is proved in Section~\ref{section:sc}, shows that the same is true for stabilizer states. Thus, imposing such structure does not separate cloning from learning in terms of sample complexity.

\subsection{Techniques}\label{subsec:techniques}

\paragraph{Structured sample amplification lower bounds.} 
Sample amplification schemes for discrete distributions~\cite{axelrod2019sampleamplificationincreasingdataset, axelrod2024statisticalcomplexitysampleamplification} are based on repeating some observed samples since, by birthday paradox, for a discrete distribution with support size $k$ and for a sample size of $\Omega(\sqrt{k})$, repetitions are not overly suspicious.
We show that for specific structured distributions, this approach fails. 
To understand why, consider a simple example: Take the class of distributions over $\Z_2^{2n}$ that are uniform on some $n$-dimensional subspace. 
Learning such a distribution comes down to seeing $n$ linearly independent samples, which is reasonably likely from $n$ samples but impossible from $n-1$ samples.
In this sense, $n$ uniformly random samples hold significantly more information about the unknown distribution than $n-1$ samples do.
As information cannot be created ``for free,'' amplifying from $n-1$ to $n$ samples with small error becomes impossible. To prove our general structured sample amplification lower bound (bullet point 1 in \Cref{thm:classical-results-informal}), we formalise this intuition as follows.

Let the random variable $F$ be uniform over $\mathcal{F}$, and let $Z_t$ denote a dataset of size $t$ labelled according to $F$. Furthermore, given a realization $z_t$ of $Z_t$, let $V_{z_t}$ denote the set of hypotheses in $\mathcal{F}$ that are consistent with the sample $z_t$. Under the uniform prior over $\mathcal{F}$, all hypotheses in
$V_{z_t}$ have the same prior probability and assign the same
likelihood to the observed sample $z_t$. Hence, the posterior after
observing $z_t$ is uniform over $V_{z_t}$. In particular, the Bayes-optimal average probability of identifying $F$ from $Z_t$ is given by $p_t = \Ex[|V_{Z_t}|^{-1}]$. Consider now the following specific learner: first, sample amplify $t$ samples to $t+1$ samples, then output a hypothesis chosen uniformly at random from all hypotheses consistent with the amplified sample. As this is a learner from $t$ samples, by Bayes-optimality, its average success probability of identifying the unknown function correctly is $\leq p_t$. In contrast, when given $t+1$ samples and choosing a consistent hypothesis uniformly at random, we succeed with probability $p_{t+1}$. Therefore, we can distinguish between $t+1$ true samples and an amplified size-$t$ dataset with advantage at least $p_{t+1}-p_t$. In particular, by \Cref{eq:intro-sample-amplification-distinguisher}---in fact, even with a uniform average instead of a supremum over the distribution class $\mathcal{D}_{\mathcal{F}}$---any sample amplifier from $t$ to $t+1$ samples incurs a sample amplification error of at least $p_{t+1}-p_t$. And if we further denote the probability that $t$ samples uniquely identify the function $f$ in the class $\mathcal{F}$ by $p_{\mathcal{F}}^t(f)$, we show that $2(p_{t+1}-p_t)\geq \Ex[p_{\mathcal{F}}^{t+1}(F) - p_{\mathcal{F}}^{t}(F)]$. Bullet point 1 in \Cref{thm:classical-results-informal} is obtained from this lower bound on the sample amplification error by making the symmetry assumption that every function in $\mathcal{F}$ has a minimum teaching sequence length equal to the teaching dimension $b_\mathcal{F}$ and that the probability that a dataset of length $b_\mathcal{F}$ produces a teaching sequence is the same for all functions in $\mathcal{F}$.

For the class of $n$-bit parity functions, a labeled dataset uniquely identifies the unknown parity if and only if the input vectors in the sample span all of $\mathbb{F}_2^n$. Clearly, this cannot be achieved with $n-1$ samples, but it is easy to see that $n$ uniformly random $n$-bit strings are linearly independent with probability 
\begin{equation}
    \prod_{i=1}^n (1-2^{-i})\geq 0.28 \, .
\end{equation}
Therefore, the class of parity functions has the required symmetry condition, has teaching dimension $n$, and the probability of obtaining a teaching sequence from $n$ i.i.d.~samples is $\geq 0.28$. Applying our general structured sample amplification lower bound thus implies: Producing one additional parity example from fewer than $n$ initial labelled examples incurs a sample amplification error of at least $0.14$. This shows that non-trivial sample amplification of parity functions requires $\geq n$ samples (bullet point 2 of \Cref{thm:classical-results-informal}), which matches the sample complexity of parity learning up to an additive constant. We mention that a similar linear independence idea was considered in~\cite{axelrod2024statisticalcomplexitysampleamplification} while presenting a very different example of a distribution which is as hard to sample amplify as to learn.

The same argument about consistency as a unique identifier exactly in the case of full-dimensional linear independence extends beyond parities to any Boolean function class that forms a linear space. In this case, after choosing a basis for the class, evaluating the corresponding basis functions on an $n$-bit input produces a column of the generator matrix of an associated linear code. Viewed this way, a labelled sample uniquely identifies the target function if and only if the sampled columns of the generator matrix span the full image space. Thus, the relevant probability in our general lower bound becomes one of randomly chosen columns in a code's generator matrix yielding a full-rank submatrix. This probability is equivalent to the probability that the dual code corrects a uniformly random erasure error of fixed size \cite{abbe2014reedmullercodesrandomerasures}, and allows us to establish bullet point 3 in \Cref{thm:classical-results-informal}.

\paragraph{Stabilizer state cloning lower bound.}
Our cloning lower bound for stabilizer states follows the same distinguisher-based proof strategy as the classical sample-amplification argument above. There, the distinguisher exploits a gap between the probability of identifying the unknown function from true and amplified data; here, our distinguisher runs an average-case optimal learner for the relevant state ensemble (which we choose to be a specific subclass of stabilizer states, see below) on the purported $t+1$ quantum copies. That learner outputs a hypothesis state, and the distinguisher, using their knowledge of the true state, accepts precisely when the hypothesis is correct.
Composing any $t$-to-$(t+1)$ quantum cloner with a $(t+1)$-copy learner gives a valid $t$-copy learner, so the cloning error must be lower bounded by the gap between the optimal average probabilities of learning from $t+1$ and from $t$ genuine copies. 

We analyse these probabilities for a specific subclass of stabilizer states. Concretely, within the class of $(3n)$-qubit stabilizer states, we consider the class $\mathrm{BiGraph}_{3n}$ of $(n,2n)$-bipartite graph states of the form
\begin{equation}
    |A\rangle = \frac{1}{\sqrt{2^{3n}}}\sum_{\substack{{x \in \Z_2^{n}}\\ {y \in \Z_2^{2n}}}}(-1)^{x^T A y} |x\rangle |y\rangle\, ,
\end{equation}
which are indexed by their biadjacency matrices $A\in\mathbb{Z}_2^{n\times 2n}$. 
First, for $t=n-1$ copies, we upper bound the optimal average probability of learning an unknown state from $\mathrm{BiGraph}_{3n}$ via a dimension counting argument.
To do so, we recall that, for any (finite) uniform ensemble of pure states $\mathcal{E}=\{(|\mathcal{E}|^{-1}, \ket{\psi_i})\}_{i=1}^{|\mathcal{E}|}$, we can upper bound the optimal average success probability of identifying a uniformly random $\ket{\psi}\in\mathcal{E}$ from $t$ copies by $\mathrm{dim}(\mathrm{span}\{\ket{\psi}^{\otimes t}\})/|\mathcal{E}|$.
We then note that every $t$-fold tensor power $|A\rangle^{\otimes t}$ for $\ket{A}\in \mathrm{BiGraph}_{3n}$ lies in the subspace spanned by the vectors
\begin{equation}
    |S\rangle = \sum_{\substack{{x_1, \ldots, x_t \in \Z_2^n} \atop {y_1, \ldots, y_t \in \Z_2^{2n}}} \atop {\sum_i x_i y_i^T} = S}  \bigotimes_{i=1}^t|x_i y_i\rangle\, ,\quad S\in\mathbb{Z}_2^{n\times 2n}: ~\mathrm{rank}(S)\leq t \, .
\end{equation}
Thus, for $t=n-1$, we conclude that the span of the states $\ket{A}^{\otimes {n-1}}$ has dimension upper bounded by the number of rank-deficient matrices in $\mathbb{Z}_2^{n\times 2n}$. These form an exponentially small fraction of all possible biadjacency matrices, leading to an exponentially small optimal average success probability of identifying an unknown state in $\mathrm{BiGraph}_{3n}$ from $n-1$ copies.
For $t+1=n$ copies, by contrast, applying Hadamard gates to the $2n$ qubits on one side of the bipartition and measuring in the computational basis produces independent samples $(x,x^{\mathsf T}A)$. These observed measurement outcomes uniquely determine $A$ as soon as we see $n$ many linearly independent sampled vectors $x$, which among $n$ samples happens with probability $\prod_{i=1}^n(1-2^{-i})\geq 0.28$. This constant jump in learning success yields a constant cloning-error lower bound. Monotonicity of the cloning error together with a simple padding (to cover the cases of $(3n+1)$- and $(3n+2)$-qubit states) then leads us to the claimed $\Omega(n)$ lower bound for general $n$-qubit stabilizer state cloning with sufficiently small constant error.

\paragraph{Mixed Abelian StateHSP cloning lower bound.}
We extend the same learning-based distinguisher strategy to classes of mixed states\footnote{\label{footnote:mistake}In a previous version of this paper, we had proposed a proof strategy that extended these mixed-state cloning lower bounds to pure states via certain versions of random purifications. Our proofs of the corresponding structured random purification results were faulty. In \Cref{appendix:purification-no-go}, we now give no-go theorems that broadly exclude the existence of such structured random purification channels as we had falsely claimed in Theorem 9 in the first arXiv version of our work, and we exclude the existence of a potentially desirable variant of such a random purification result tailored to stabilizer states.} arising from Abelian State Hidden Subgroup Problems, so-called Abelian StateHSPs \cite{bouland2024statehiddensubgroupproblem, hinsche2025abelianstatehiddensubgroup}. Fix a finite Abelian group $G$ and a unitary representation $\mu$ that contains every irreducible representation of $G$ exactly once, and let $\mathcal{S}_\alpha$ be the class of states, including mixed states, that hide an unknown subgroup $H\leq G$ of fixed order $\alpha$. We first identify the optimal strategy for learning the hidden subgroup in the worst case. For any state $\rho$ hiding $H$, its group-twirled version
\begin{equation}
    \overline{\rho}=\frac{1}{|G|}\sum_{g\in G}\mu(g)\rho\mu(g)^\dagger
\end{equation}
hides the same subgroup and commutes with the representation. It is therefore block diagonal with respect to the character POVM. As the irreducible subrepresentations are by assumption non-isomorphic, all character POVM elements have rank one. Hence, the tensor powers of the twirled state are diagonal with respect to the rank-one product POVM formed by independently applying the character POVM to every copy. Consequently, over these twirled worst-case instances, no collective measurement on several copies can extract more information about $H$ than by measuring the character POVM independently on every copy and classically post-processing the outcomes.

For the cloning lower bound, we restrict to Abelian StateHSPs with underlying unitary representation in which each irrep appears exactly once (see \Cref{eq:new_rep_assumption}), and we then consider the (generally mixed) instances
\begin{equation}
    \sigma_H=\frac{1}{|H^\perp|}\sum_{\lambda\in H^\perp}|\lambda\rangle\!\langle\lambda|,
\end{equation}
diagonal in the character POVM basis, which form a subclass $\mathcal{S}^{\sigma}_\alpha\subseteq\mathcal{S}_\alpha$. Measuring the character POVM on $\sigma_H$ produces a uniformly random element of the dual subgroup $H^\perp$. Thus, learning the hidden subgroup from copies of $\sigma_H$ becomes the classical task of identifying $H^\perp$ from independent samples drawn uniformly at random precisely from $H^\perp$. This suggests the following distinguisher: Measure the character POVM on the purported copies, choose uniformly at random among the order-$\alpha$ subgroups consistent with the observed measurement outcomes, and accept exactly when the chosen subgroup is the target $H$, which is known to the distinguisher in advance. As we restrict ourselves to subgroups of a fixed order $\alpha$, every subgroup consistent with the observations has the same likelihood. So, as in the classical sample-amplification argument, outputting a consistent subgroup uniformly at random is the Bayes-optimal learner under the uniform prior over subgroups.

Once again, composing a $(t-1)$-to-$t$ cloner with a $t$-copy learner gives a valid $(t-1)$-copy learner. The optimality of character measurements bounds the average probability of this learner correctly identifying $H$ by the corresponding optimal success probability from $t-1$ true copies. In contrast, when run on $t$ genuine copies, the distinguisher achieves the $t$-copy learning probability.  Writing $p^k_{\sigma_H}$ for the probability that $k$ character-POVM
measurement outcomes obtained from independent copies of $\sigma_H$
uniquely determine $H$, fix any integer threshold $d\geq 1$.
The same consistency argument as in the classical setting, together
with monotonicity of the cloning error, gives
\begin{equation}
    \epsilon^\star_{\mathrm{Cl}}(\mathcal{S}_\alpha,t-1,t)
    \geq \epsilon^\star_{\mathrm{Cl}}(\mathcal{S}^{\sigma}_\alpha,t-1,t)
    \geq \frac{1}{2}\underset{\substack{H\leq G\\ |H|=\alpha}}{\Ex}
    \left[p^{d}_{\sigma_H}-p^{d-1}_{\sigma_H}\right],\quad \text{for any }t\leq d.
\end{equation}
Hence, whenever an additional character POVM measurement sample at some threshold $d-1\to d$ produces a non-negligible jump in the probability of identifying the hidden subgroup, cloning below that number of copies necessarily incurs a corresponding non-negligible error.

\subsection{Directions for Future Work}

We have studied quantum cloning for particular classes of structured quantum states, namely stabilizer states and mixed-state instances of certain Abelian StateHSPs.
For stabilizer states, we have shown that the sample complexity of cloning essentially coincides with that of learning. 
We have achieved this by relating quantum cloning to classical sample amplification.
Here, we have established that sample amplification is as hard as learning for parity functions, based on a simple linear independence argument. Finally, we have highlighted some connections between sample amplification and coding theory.
Thus, our work raises several natural follow-up questions.

\paragraph{Sample amplification of Reed-Muller codes.}
As we have given a general recipe for obtaining sample amplification lower bounds for linear spaces of Boolean functions from coding theory, a clear challenge is to instantiate this recipe beyond our example of parities.
A natural next example to explore here is the space of low-degree polynomials over $\Z_2$ and the corresponding Reed-Muller codes. 
Using our arguments, the structured sample amplification error for degree-$d$ polynomials on $n$ bits can be lower bounded by the probability that a Reed-Muller (with suitable parameters) can correct $n \choose \leq d$ random erasure errors. While the question of robustness to random erasures has been studied extensively for Reed-Muller codes and they are known to achieve capacity for erasure noise channels in a variety of regimes~\cite{abbe2014reedmullercodesrandomerasures, bhandari2022}, characterising the above probability, to the best of our knowledge, remains open. 

\paragraph{Combinatorial dimension for sample amplification.}
It has been a fruitful endeavour in computational learning theory to characterize learnability of concept classes in different models via combinatorial parameters such as the VC dimension~\cite{vapnik1971}, teaching dimension~\cite{GOLDMAN199520}, or Eluder dimension~\cite{li2022understandingeluderdimension}, among many others.
\Cref{thm:classical-results-informal} already shows that the teaching dimension plays a key role in lower bounds for sample amplification.
However, our results fall short of a full combinatorial characterization of the complexity of sample amplification.
Achieving such a characterization may require a new combinatorial dimension and would constitute a significant advance over the current case-by-case understanding of sample amplification.

\paragraph{Cloning hypergraph states.}
We have shown that stabilizer states are as hard to (approximately) clone as to learn, an approximate No-Cloning theorem for stabilizer states. It is now natural to consider the same question for other structured classes of states. An interesting example to explore here is the class of hypergraph states of order $d$~\cite{Rossi_2013}, which also admit a description in terms of stabilizer operators that include Toffoli gates in addition to Paulis. 
These hypergraph states can alternatively be thought of as phase states
\begin{equation}
    |\psi_f\rangle = \frac{1}{\sqrt{2^n}} \sum_{x} (-1)^{f(x)}|x\rangle\, ,
\end{equation}
with $f$ a degree-$d$ polynomial over $\Z_2$~\cite{arunachalam2023optimalalgorithmslearningquantum}. 
It is an intriguing question whether cloning hypergraph states can be related to sample amplification of low-degree polynomials, and whether such a connection can give rise to lower bounds.
Beyond hypergraph states, the study of cloning for classes such as matrix product states, output states of $\mathsf{QNC}^0$ and $\mathsf{QAC}^0$ circuits may be of interest.

\paragraph{Cloning against restricted distinguishers.}
In the distinguisher-based perspective on cloning, we have not restricted the computational power of the distinguisher. 
As \cite{axelrod2019sampleamplificationincreasingdataset} has highlighted for sample amplification, we may obtain interesting variants of the amplification and cloning tasks by considering restricted distinguishers. 
Concretely, it is natural to consider the task of cloning structured classes of states against polynomial-time adversaries. 
This is naturally motivated by quantum cryptography, with potential applications in quantum money~\cite{wiesner1983, jls_pseudo_2018, molina2012}, and related questions have been considered in~\cite{fefferman2025hardnesslearningquantumcircuits, nehoran2024, bostanci2025duality}. 
While our lower bounds for stabilizer state cloning immediately carry over to cloning against polynomial-time adversaries, since our proofs of these bounds use computationally bounded distinguishers, it remains to explore the impact of imposing computational restrictions on the adversary for cloning other classes of quantum states.
Specifically, it would be interesting to see whether cloning and learning can be separated for relevant classes of states in this computational model of cloning.

\section{Notation and Preliminaries}~\label{section:preliminaries}

In this section, we introduce the notation and underlying notions used throughout the paper. A familiarity with standard notions as well as Dirac bra-ket notation in quantum information theory is assumed, for this the reader is referred to excellent textbooks such as~\cite{nielsen2010quantum}.

\subsection{Representation Theory}

In this subsection, we recall basic notions from representation theory. We refer the reader to the excellent references~\cite{mele2025rep, fulton2013representation, serre1977linear} for a detailed introduction.
Throughout, we will use $V$ to denote a finite-dimensional vector space. We use $\GL(V)$ to denote the group of all bijective linear maps from $V$ to itself.

\begin{definition}[Representation of finite groups]
    Let $G$ be a finite group. Then a tuple $(\mu, V)$ with finite-dimensional vector space $V$ and map $\mu: G \rightarrow \GL(V)$ is called a \emph{representation} of the group $G$ over the vector space $V$ if $\mu$ is a homomorphism. That is, 
    \begin{equation}
        \mu(gh) = \mu(g)\mu(h)   \quad\forall g,h\in G.
      \end{equation}
\end{definition}

At times, we will be abusing the terminology: We will not explicitly mention the vector space $V$ and simply call $\mu$ a representation, wherever the vector space is clear from the context.

\begin{definition}[Sub-representation of finite Groups]
    Let $G$ be a finite group and let $(\mu, V)$ be a representation of the group $G$. Furthermore, let $W\subseteq V$ be a subspace such that $W$ is $G$-invariant, i.e.,
    \begin{equation}
       \mu(g) |w\rangle \in W \hspace{3ex} \forall |w\rangle \in W, \forall g \in G\, . 
    \end{equation}
    Then $(\mu|_W, W)$ is called a \emph{sub-representation} of the representation $(\mu, V)$.
\end{definition}

\begin{definition}[Irreducible Representation]
    Let $G$ be a group with a representation $(\mu, V)$. Then a sub-representation $(\mu|_W, W)$ is called \emph{irreducible (or, an irrep)} if there are no non-trivial sub-representations of $(\mu|_W, W)$.
\end{definition}

Here, we say that a sub-representation is non-trivial if it is onto a strict but non-trivial (i.e., not equal to $\{0\}$) subspace.
It is trivial to see that any sub-representation onto a one-dimensional vector space is irreducible.

\begin{definition}[Unitary Representation]
    Let $G$ be a finite group and let $(V, \langle \cdot | \cdot \rangle)$ be a complex inner product space. Then, a representation $(\mu, (V, \langle \cdot | \cdot \rangle))$ is called a \emph{unitary representation} with respect to the inner product if,
    \begin{equation}
        \langle \mu(g) v | \mu(g) w\rangle = \langle v | w\rangle , \hspace{3ex}\forall v, w \in V \text{ and } \forall g \in G.
    \end{equation}
\end{definition}
In other words, a unitary representation is a representation $\mu$ such that $\mu(g)$ is unitary w.r.t.~the inner product space $(V, \langle \cdot | \cdot \rangle)$ for all $g\in G$. For finite groups, any representation can be treated as a unitary representation with respect to a suitably defined inner product.
Thus, it is natural to restrict our focus to unitary representations for finite groups. 

The following is the natural notion of equivalence between representations: 

\begin{definition}[Isomorphic Representations] Two representations $(\mu_1, V_1)$ and $(\mu_2, V_2)$ are called \emph{isomorphic} if there exists an invertible linear map $\phi: V_1\rightarrow V_2$ such that
\begin{equation}
    \mu_2(g) = \phi\mu_1(g)\phi^{-1}, \hspace{3ex} \forall g \in G\, .
\end{equation}
If $\mu_1$ and $\mu_2$ are isomorphic, we write, $\mu_1 \cong \mu_2$.
\end{definition}

The following complex-valued map induced by a representation is useful.

\begin{definition}[Character]
    For a representation $\mu$ of a finite group $G$, its \emph{character} is the function
    \begin{equation}
        \chi_{\mu}:G\to \mathbb{C}\, ,\chi_{\mu}(g) = \tr(\mu(g))\, .
    \end{equation}
\end{definition}

Two representations are isomorphic if and only if their characters are same, i.e.,
\begin{equation}
    \mu_1 \cong \mu_2 \, \iff\,  \chi_{\mu_1}(g) = \chi_{\mu_2}(g), \hspace{3ex} \forall g \in G\, . 
\end{equation}

\begin{lemma}[Schur's orthonormality condition]\label{lemma:schur-on}
    Let $\mu$ and $\nu$ be two irreducible representations of a finite group $G$. Then,  
    \begin{equation}
        \Ex_{g \in G} [\overline{\chi_{\mu}(g)} \chi_{\nu}(g)] = \delta_{\mu \cong \nu}\, .
    \end{equation}
\end{lemma}

In this work, we restrict our focus to Abelian groups, i.e., groups $G$ in which the group operation satisfies $gh=hg$ for all $g,h\in G$. For such Abelian groups, it is easy to fully characterize the irreducible representations and characters. To this end, first notice that, if $G$ is a Abelian and if $\mu$ is a representation of $G$, then $\mu(g)$ and $\mu(h)$ commute for all $g,h\in G$, i.e., $[\mu(g), \mu(h)]=0$. If $\mu$ is additionally unitary, then in particular this implies that $\mu(g)$ is unitarily diagonalizable for every $g\in G$, and together with the commutativity this implies that all $\mu(g)$ can be simultaneously diagonalized. It thus makes sense to speak of the eigenvectors of the representation $\mu$. 

\begin{lemma}
    For a unitary representation $\mu$ of an Abelian group $G$, the irreducible sub-representations of $\mu$ are exactly the one-dimensional spaces spanned by the eigenvectors of $\mu$, and the corresponding eigenvalues are the corresponding characters.
\end{lemma}

In the remainder of the paper, all representations will be unitary representations of an Abelian group over some Hilbert space, unless stated otherwise. 
We can then write the eigenvectors as $|\lambda, v_{\lambda}\rangle$, where $\lambda$ indicates the eigenvalue and where $v_{\lambda}$ enumerates the eigenvectors with the same eigenvalue $\chi_{\lambda}(g), \forall g \in G$. Equivalently, the $v_\lambda$ enumerates the isomorphic irreducible representations with character $\chi_\lambda$. The space spanned by isomorphic irreps, $\mathrm{span}\{|\lambda, v_{\lambda}\rangle\}_{v_\lambda}$, is called \emph{$\lambda$-isotypic component} of the Hilbert space 
Now, we can define the projective measurement $\{\Pi_{\lambda}\}_{\lambda}$ via
\begin{equation}
    \Pi_{\lambda} = \sum_{v_{\lambda}} |\lambda, v_{\lambda}\rangle \langle \lambda, v_{\lambda}| \, ,
\end{equation}
which are projectors on the $\lambda$-isotypic subspaces of the Hilbert space. This measurement is called the character POVM of the representation. 
The following alternative expression for the character POVM elements will be useful:

\begin{fact}[\cite{hinsche2025abelianstatehiddensubgroup}]\label{fact:character-povm}
    The character POVM elements can be written as
    \begin{equation}
        \Pi_{\lambda} = \frac{1}{|G|} \sum_{g \in G} \overline{\chi_{\lambda}(g)} \mu(g).
    \end{equation}
\end{fact}

\subsection{Abelian State Hidden Subgroup Problem}
Symmetries are at the core of physics. It is thus a natural question to ask whether we can learn the symmetries of an unknown quantum state given some form of access to the state. This serves as a generalization of the \textit{Hidden Subgroup Problem (HSP)} \cite{shor1994, Jozsa2001}, which entails finding the hidden symmetries of a Boolean function from (either classical or quantum) query access. A restricted version of this problem for quantum states, the so-called \emph{Abelian State Hidden Subgroup Problem (Abelian StateHSP)} has recently been formalised~\cite{bouland2024statehiddensubgroupproblem, hinsche2025abelianstatehiddensubgroup}. Formally, the problem is defined as follows:

\begin{definition}[Abelian StateHSP $(G, \mu, \Psi_H^{\epsilon})$~\cite{hinsche2025abelianstatehiddensubgroup}]\label{defn:abelian-statehsp}
    Let $G$ be an Abelian group and let $H\leq G$ be a subgroup. Further, let $\mu:G\rightarrow U(\mathcal{H})$ be a unitary representation of $G$ over some Hilbert space $\mathcal{H}$. Then, we say that a state $\rho \in \cal{B}(\mathcal{H})$ \emph{$\epsilon$-hides the subgroup $H\leq G$} (for which we write $\rho \in \Psi_H^{\epsilon}$) if, 
    \begin{enumerate}
        \item $\forall h \in H$, $ \tr(\mu(h)\rho)  = 1$, and
        \item $\forall g \in G\setminus H$, $|\tr(\mu(g) \rho)| \leq 1-\epsilon$.
    \end{enumerate}
    The Abelian StateHSP for $G$ is the following problem: Given i.i.d.~copies of a state $\rho$ that is promised to $\epsilon$-hide $H$, identify $H$. If we are additionally promised that $\rho$ is pure, then we refer to this variant of the problem as \textit{pure Abelian StateHSP}.
\end{definition}

In this formulation, we can also think of the Abelian StateHSP as a problem of discriminating between sets $ \Psi_H^{\epsilon}$ of states, each of which $\epsilon$-hides a symmetry subgroup $H \leq G$ of fixed order $\alpha$. 
The Abelian StateHSP can be solved using $O(\log |G|/\epsilon)$ copies of the state~\cite{hinsche2025abelianstatehiddensubgroup}. The main algorithmic step is to perform the character POVM $\{\Pi_{\lambda}\}_{\lambda}$ and to then classically post-process the observed measurement outcomes. Next, following \cite{hinsche2025abelianstatehiddensubgroup}, we sketch how the Abelian StateHSP algorithm works in a bit more detail.

Note that the outcome probabilities when performing the character POVM are
\begin{equation}
    q_{\rho}(\lambda) = \tr(\rho \Pi_{\lambda}).
\end{equation}
For a finite Abelian group $G$, we can define the dual group $\widehat{G}$ as the group of all character functions over $G$ corresponding to irreducible representations. Then, for a subgroup $H\leq G$, we can define the dual subgroup $H^{\perp} \leq \widehat{G}$ \begin{equation}
    H^{\perp} = \{\lambda \in \widehat{G} ~|~ \chi_{\lambda} (h) = 1, \forall h \in H\}\, .
\end{equation}
We also define
\begin{equation}
    q_{\rho}(H^{\perp}) = \sum_{\lambda \in {H}^{\perp}} q_{\rho}(\lambda) = \tr\lr{\rho \sum_{\lambda \in H^{\perp}} \Pi_{\lambda}}\, .
\end{equation}
As in \cite{hinsche2025abelianstatehiddensubgroup}, the probability can be shown to be equal to
\begin{equation}
    q_{\rho}(H^{\perp}) = \frac{1}{|H|}\sum_{h \in H} \tr(\mu(h) \rho)\, .
\end{equation}
For an Abelian StateHSP instance, $\rho \in \Psi_H^{\epsilon}$, we have
\begin{equation}
    q_{\rho}(H^{\perp}) = 1.
\end{equation}
Thus, if we use character POVMs $\{\Pi_{\lambda}\}_{\lambda}$, we will always sample a $\lambda \in H^{\perp}$, and given enough independent samples, we can determine $H^{\perp}$ and hence its dual $H$.

\subsection{Stabilizer States and Phaseless Stabilizer StateHSP}~\label{section:intro-to-stab-states-and-stabhsp}
Let $\mc{P}_n$ be the $n$-qubit Pauli group. The single-qubit Pauli group is generated by $X$ and $Z$ operators (up to phases), which act as
\begin{equation}
    \begin{split}
        &X|a\rangle = |a\oplus 1\rangle,  \\
        &Z|a\rangle =  (-1)^a|a\rangle,
    \end{split}
\end{equation}
where $a \in \{0, 1\}.$ 
The $n$-qubit Pauli group consists of tensor products of elements of the single-qubit Pauli group.
An $n$-qubit stabilizer state $|\psi\rangle \in \mathrm{Stab}_n$ is defined to have a stabilizer group $S \subset \mc{P}_n$ of order $2^n$ such that
\begin{equation}
    P|\psi\rangle = |\psi\rangle, \forall P \in S.
\end{equation}
We define the following $n$-qubit operators
\begin{equation}\label{eq:phaseless-weyl-ops}
    V_x =  \bigotimes_{i = 1}^n Z^{a_i} X^{b_i}, \hspace{3ex} \forall x \in \Z_2^{2n}\, ,
\end{equation}
which are Weyl operators up to an imaginary phase. Moreover, these operators satisfy,
\begin{equation}
    V_x V_y = (-1)^{[x, y]} V_y V_x\, .
\end{equation}
and that,
\begin{equation}
    V_x V_y = (-1)^{b \cdot c} V_{x \oplus y}\, .
\end{equation}
where $x = (a, b)$, $y = (c, d)$ and $[x, y] = a\cdot d + b \cdot c \mod 2$. Now, define the $(2n)$-qubit Bell basis states as,
\begin{equation}
    |\Phi_y\rangle= (V_y \otimes I) |\Phi_0\rangle, \hspace{3ex} \text{ with } \hspace{3ex} |\Phi_0\rangle= \frac{1}{\sqrt{2^n}} \sum_{x \in \Z_2^{n}} |x\rangle \otimes |x\rangle\, ,
\end{equation}
for any $y \in \Z_2^{2n}$. It is easy to see that these are orthonormal and forms a basis for the space of $(2n)$-qubit states. Moreover,
\begin{equation}
    V_x^{\otimes 2} |\Phi_y\rangle = (-1)^{[x, y]} |\Phi_y\rangle, \hspace{3ex} \forall x, y \in \Z_2^{2n}.
\end{equation}

\subsection{Mixed Phaseless Stabilizer StateHSP}
We consider the following unitary representation,
\begin{equation}\label{eq:unitary-rep-stab-hsp}
    \begin{split}
        \mu:\mathbb{Z}_2^{2n} \rightarrow \mathcal{U}(((\mathbb{C}^{2})^{\otimes n})^{\otimes 2})\, ,x \mapsto V_{x}^{\otimes 2}.
    \end{split}
\end{equation}
Then, we define the states $\sigma_L$ as,
\begin{equation}
    \sigma_L = \frac{1}{2^n} \sum_{y \in L^{\perp}} |\Phi_y\rangle \langle \Phi_y|\, ,
\end{equation}
where $L$ is an $n$-dimensional subspace of $\Z_2^{2n}$ and $L^{\perp}$ is its dual subspace defined as,
\begin{equation}
    L^{\perp} = \{y \in \Z_2^{2n} ~|~ (-1)^{[x, y]} = 1, \forall x \in L\}.
\end{equation}

\begin{lemma}\label{lemma:sigma-stab-states}
    The states $\sigma_L$ as defined above have the properties,
    \begin{enumerate}
        \item $\forall x \in L$, $\tr \lr{V_x^{\otimes 2} \sigma_L} = 1$, and
        \item $\forall x \notin L$, $\tr \lr{V_x^{\otimes 2} \sigma_L} = 0$.
    \end{enumerate}
\end{lemma}

Thus, using ~\Cref{lemma:sigma-stab-states}, we can define an Abelian StateHSP $(\Z_2^{2n}, \mu, \Psi_{L}^{1})$ where the unitary representation $\mu$ is given by~\eqref{eq:unitary-rep-stab-hsp} and,
\begin{equation}
    \Psi_{L}^{1} = \{\sigma_L\}.
\end{equation}

By~\Cref{lemma:sigma-stab-states}, the state $\sigma_L$ hides the subgroup $L$ of $\Z_2^{2n}$. Now, the irreducible representations (eigenstates) of the unitary representation $V_x^{\otimes 2}$ are just the Bell basis defined in the previous section $\{|\Phi_y\rangle\}_{y \in \Z_{2}^{2n}}$. Moreover, the character (eigenvalues) are given as,
\begin{equation}
    \chi_{y}(x) = (-1)^{[x, y]}\, .
\end{equation}
From this, we in particular see that the characters for any two irreps $y \neq \tilde{y}$ satisfy $\chi_y\neq \chi_{\tilde{y}}$. Thus, from the previous discussion, all irreps are non-isomorphic. And, the character POVM is given by,
\begin{equation}
    \Pi_{y} = |\Phi_y\rangle \langle \Phi_y|.
\end{equation}
We call the Abelian StateHSP $(\Z_2^{2n}, \mu, \Psi_L^1)$ as mixed-phaseless-stabilizer StateHSP because of the similarity of the hidden group to the (phaseless) stabilizer group of the stabilizer states.

\subsection{Random Purification}
The random purification channel by~\cite{tang2025conjugatequerieshelp} allows one to transform $t$ i.i.d.~copies of an arbitrary unknown mixed state to $t$ i.i.d.~copies of a randomly chosen purification of that state. 

\begin{lemma}[\cite{tang2025conjugatequerieshelp}, Lemma 2.11]
    Let $m,r\in\mathbb{N}$.
    There is a unitary circuit $C^{(m)}$ such that for all mixed states $\rho \in \mathbb{C}^{d \times d}$ of rank at most $r$,
    \begin{equation}
        C^{(m)}(\rho^{\otimes m}) = \Ex_{|\rho\rangle} |\rho\rangle \langle \rho|^{\otimes m} \, ,
    \end{equation}
    where the expectation is over random purifications $|\rho\rangle\in\mathbb{C}^{d}\otimes \mathbb{C}^r$ of $\rho$. Additionally, the circuit $C^{(m)}$ can be implemented to accuracy $\epsilon$ with $\poly(m, \log d, \log(1/\epsilon))$ gate complexity.
\end{lemma}
On a very high level, the approach in \cite{tang2025conjugatequerieshelp} is to first perform weak Schur sampling on the $m$ copies of the mixed input state, and to then prepare a specific state conditioned on the outcome observed. 
The random purification channel has attracted considerable attention in the community recently, evidenced by an alternate construction and proof \cite{girardi2025randompurificationchannelsimple}, an extensions to bosonic or fermionic states \cite{mele2025randompurificationchannelpassive, walter2025randompurificationchannelarbitrary}, extensions to quantum channels \cite{girardi2025randomstinespringsuperchannelconverting, yoshida2025randomdilationsuperchannel}, and applications in quantum learning theory \cite{pelecanos2025mixedstatetomographyreduces, mele2025optimallearningquantumchannels}. Ideas similar to random purification have been studied before~\cite{mehdi22testing, chen2025learningquantumgibbsstates}. Specifically, ~\cite{chen2025localtestunitarilyinvariant} studies a notion of local testing, which can be thought of as dual to random purifications, and which has been employed to study property testing of states~\cite{chen2025localtestunitarilyinvariant} and quantum channel tomography~\cite{chen2026quantumchanneltomographyestimation, chen2026quantumchanneltomographyoptimal}.

We also highlight the more general random purification channel introduced recently for general symmetries~\cite{walter2025randompurificationchannelarbitrary}. We denote the Hilbert space on which our mixed state lives in by $\mc H$ and the Hilbert space for the purifying register by $\mc H^{\prime}$.
\begin{theorem}[{\cite[from Theorem 3.1 and its proof]{walter2025randompurificationchannelarbitrary}}]\label{thm:general_random_purification}
For any $*$-algebra $\mc{A}$, there is a quantum channel $\mc{P}_{\mc{A}}: \mc{B}(\mc H) \rightarrow \mc{B}(\mc{H}\otimes \mc{H}^{\prime})$ with the following properties:
\begin{enumerate}
    \item Random Purification of symmetric states: For any quantum state $\rho \in \mc{S}(\mc H)$ that commutes with $\mc{A}$ i.e. $\rho \in \mc{A}^{\prime}$, the commutant of algebra $\mc{A}$, we have that
    \begin{equation}
        \mc{P}_{\mc A}(\rho) = \int_G (I \otimes g^T) \psi_{\rho}^{\mathrm{std}}(I \otimes \bar{g}) dg
    \end{equation}
    where $G \subseteq U(\mc H)$ is any closed subgroup with $\mathbb{C}G = \mc{A}^{\prime}$, and $|\psi_{\rho}^{\mathrm{std}}\rangle = (\sqrt{\rho}\otimes I)|\tilde{\Phi}_0\rangle$, with $|\tilde{\Phi}_0\rangle = \sum_{x}|x\rangle |x^{\prime}\rangle \in \mc{H} \otimes \mc{H}^{\prime}$ the un-normalized maximally entangled state on $\mc{H} \otimes \mc{H}^{\prime}$. 
    \item Assume that the Hilbert space decomposes as $\mc{H} \cong \bigoplus_{\lambda \in \Lambda} L_{\lambda} \otimes R_{\lambda}$ (and $\mc{H}^{\prime} \cong \bigoplus_{\lambda \in \Lambda} L^{\prime}_{\lambda} \otimes R^{\prime}_{\lambda}$) with finite index set $\Lambda$, such that
    \begin{equation}
        \mc{A} \cong \bigoplus_{\lambda \in \Lambda} \mc{B}(L_{\lambda}) \otimes I_{R_{\lambda}}, \hspace{3ex} \text{and} \hspace{3ex} \mc{A}^{\prime} \cong \bigoplus_{\lambda \in \Lambda} I_{L_{\lambda}} \otimes \mc{B}(R_{\lambda}), 
    \end{equation}
    (which is always possible for $*$-algebras). Then, the channel $\mc{P}_{\mc A}$ is given by
    \begin{equation}\label{eqn:general_random_purification_channel}
        \mc{P}_{\mc A}(\rho) = \bigoplus_{\lambda \in \Lambda} \frac{|\tilde{\Phi}_0\rangle_{L_{\lambda} L^{\prime}_{\lambda}} \langle \tilde{\Phi}_0|_{L_{\lambda} L_{\lambda}^{\prime}}}{\textrm{dim } L_{\lambda}} \otimes \tr_{L_{\lambda}}\lrq{P_{\lambda} \rho P_{\lambda}} \otimes \frac{I_{R^{\prime}_{\lambda}}}{\textrm{dim } R^{\prime}_{\lambda}},
    \end{equation}
    where $|\tilde{\Phi}_0\rangle_{L_{\lambda} L^{\prime}_{\lambda}}$ is the un-normalised maximally entangled state on the $L_{\lambda} \otimes L_{\lambda}^{\prime}$ component, $P_{\lambda}$ is the projector on $\lambda$-component $L_{\lambda} \otimes R_{\lambda}$ of $\mc{H}$.
\end{enumerate}
\end{theorem}
From \eqref{eqn:general_random_purification_channel}, we see that to implement the random purification channel $\mc{P}_{\mc A}$, we first measure the given state $\rho$ using projectors $\{P_{\lambda}\}_{\lambda \in \Lambda}$, obtaining outcome $\lambda$ and the (un-normalised) post-measurement state $P_{\lambda} \rho P_{\lambda}$ on $L_{\lambda} \otimes R_{\lambda}$. Then, we discard the state on register $L_{\lambda}$, prepare the maximally entangled state on $L_{\lambda} \otimes L_{\lambda}^{\prime}$ and the maximally mixed state on register $R_{\lambda}^{\prime}$.

\section{Structured Sample Amplification}\label{section:ssa}
We now formally define the notion of sample amplification for distributions with a special structure induced by a function class. Namely, we consider the distributions of the form $(\mc V_n, f)$ over $\Z_2^n \times \Z_2$, where $\mc V_n$ is some fixed distribution over the $n$-bit input marginal, and where $f$ is some Boolean function. More formally, for $f \in \mc{F}$, where $\mathcal{F}$ is some Boolean function class, let $D_f = (\mc{V}_n, f)$ denote the distribution over 
$\mc{X} = \Z_2^n \times \Z_2$ obtained by drawing 
$X \sim \mc{V}_n$ and setting $Y = f(X)$. Equivalently,
\begin{equation}
    D_f(x,y) = \mc{V}_n(x)\,\mathbf{1}_\{y = f(x)\},
    \qquad (x,y) \in \Z_2^n \times \Z_2.
\end{equation}

\begin{definition}[Structured Sample Amplification]\label{def:structured-sample-amplification}
    Let $\mathcal{F}$ be a class of Boolean functions on $n$ bits and let $\mathcal{V}_n$ be some distribution over the domain $\mathbb{Z}^n_2$. We say that the class of distributions $\mathcal{D}= \{D_f\}_{f \in \mathcal{F}} = \{(\mathcal{V}_n, f)\}_{f \in \mathcal{F}}$ over the domain $\mathcal{X} = \mathbb{Z}^n_2 \times \mathbb{Z}_2$ \emph{admits a $(t, t+m, \epsilon)$-sample amplification scheme} if there exists a stochastic map $T_{\mathcal{D}, t, m, \epsilon}: \mathcal{X}^{t} \rightarrow \mathcal{X}^{t+m}$ such that, 
    \begin{equation}
        \sup_{f \in \mathcal{F}} \tv\lr{D_f^{\otimes t} \circ T^{-1}_{\mathcal{D}, t, m, \epsilon}, D_f^{\otimes t+m}} \leq \epsilon \, ,
    \end{equation}
    where $\tv$ denotes the total variation distance. 
   The \emph{minimax sample amplification error} for the distribution class $\mathcal{D}$ is defined by minimising the sample amplification error over all stochastic maps $T_{SA}: \mathcal{X}^{t} \rightarrow \mathcal{X}^{t+m}$,
    \begin{equation}
        \epsilon^{\star}_{SA}(\mathcal{D}, t, t+m) := \min_{T_{SA}} \sup_{f \in \mathcal{F}}\tv\lr{D_f^{\otimes t} \circ T_{SA}^{-1}, D_f^{\otimes t+m}}.
    \end{equation}
\end{definition}

We recall that for two probability distributions $P$ and $Q$ over a finite domain $\Omega$, 
their total variation distance is defined as
\begin{equation}
    \tv(P,Q)
    := \frac{1}{2}\sum_{x\in\Omega}|P(x)-Q(x)|
    = \max_{A\subseteq\Omega}|P(A)-Q(A)|.
\end{equation}
Equivalently,
\begin{equation}
    \tv(P,Q)
    = \max_{\mc A:\Omega\rightarrow\{0,1\}}
    \left|
        \Pr_{X\sim P}[\mc A(X)=1]
        -
        \Pr_{X\sim Q}[\mc A(X)=1]
    \right|.
\end{equation}
Here, a \emph{distinguisher} $\mc A:\Omega\rightarrow\{0,1\}$ is a 
binary-valued statistical test, and its \emph{distinguishing advantage} 
is the absolute difference between its acceptance probabilities under 
$P$ and $Q$. Notice that we can also take the supremum over randomised distinguishers without changing its value.

The above definition of structured sample amplification is simply that of general sample amplification as in \cite{axelrod2019sampleamplificationincreasingdataset, axelrod2024statisticalcomplexitysampleamplification} with the additional assumption that the unknown distribution is of the form $(\mc V_n, f)$, where $\mc V_n$ is known and where $f\in \mc F$ for some known function class $\mc F$.
Structured sample amplification can alternatively be thought of as a game between two parties, an amplifier and a distinguisher, where the distinguisher knows the distribution. Then, the minimax sample amplification error is the maximum advantage achievable by an adversarial distinguisher $\mc A$ in telling the true samples and amplified samples apart: 
\begin{equation}\label{eq:sa-error-distinguisher-formulation}
    \epsilon^{\star}_{SA}(\mc{D}, t, t+1) = \min_{T_{SA}} \sup_{f \in \mc F}\max_{\mc{A}}\mathrm{Adv}(\mc{A}, \mc{V}_n,f, T_{SA}, t) \, ,
\end{equation} 
where
\begin{equation}
    \mathrm{Adv}(\mc{A}, \mathcal{V}_n, f, T_{SA}, t) = \lra{\Pr_{Z_{t+1}}[\mc{A}^f(Z_{t+1}) = 1] - \Pr_{Z_{t}}[\mc{A}^f(T_{SA}(Z_{t})) = 1]}.
\end{equation}
Note that we will drop the dependence on the input marginal distribution $\mc{V}_n$ later when it is fixed.
We sometimes omit the superscript $f$ from the distinguisher, but it is to be understood that the distinguisher has knowledge of underlying Boolean function (and of the sample amplification procedure).

\subsection{General Lower Bound}\label{subsec:amplification-general-lower-bound}
For a labeled sample
\begin{equation}
    z_t = \big((x^{(1)},y^{(1)}),\ldots,(x^{(t)},y^{(t)})\big),
\end{equation}
we say that a function $f'\in\mc F$ is \emph{consistent with $z_t$} if
\begin{equation}
    f'(x^{(i)}) = y^{(i)}, \qquad \forall i\in[t].
\end{equation}
We denote the number of functions in $\mc F$ consistent with $z_t$ by
\begin{equation}
    n_{\mc F}(z_t)
    :=
    \left|
        \left\{
            f'\in\mc F
            ~\middle|~
            f'(x^{(i)})=y^{(i)},\ \forall i\in[t]
        \right\}
    \right|.
\end{equation}

For $f\in\mc F$, let $p_{\mc F}^t(f)$ denote the fraction of input
sequences of length $t$ for which the corresponding labeled sample
uniquely determines $f$ in $\mc F$, i.e.,
\begin{equation}
    p_{\mc F}^t(f)
    :=
    \frac{
        \left|
            \left\{
                x_t\in(\Z_2^n)^t
                ~\middle|~
                n_{\mc F}\big((x_t,f(x_t))\big)=1
            \right\}
        \right|
    }{2^{nt}}.
\end{equation}
Equivalently, $p_{\mc F}^t(f)$ is the probability that $t$
i.i.d.~uniformly sampled inputs, labeled according to $f$, uniquely
determine $f$ in $\mc F$.

For simplicity, we restrict ourselves to the uniform marginal over the $n$-bit inputs, $\mc{U}_n$. With this notation established, we can now state and prove our first main result.

\begin{theorem}[General Lower Bound]\label{thm:general-lower-bound}
    Let $d\in\mathbb{N}$.
    The minimax sample amplification error for the distribution class, $\mc{D} = \{(\mc{U}_n, f)\}_{f \in \mc{F}}$ over the domain $\Z_2^n \times \Z_2$ in amplifying $t$ samples to $t+1$ samples for any $t \leq d-1$ satisfies
    \begin{equation}
         \epsilon_{SA}^{\star}(\mathcal{D}, t, t+1) \geq \frac{1}{2}\Ex_{f\sim\mc{F}}\lrq{p^{d}_{\mc F}(f) - p^{d-1}_{\mc F}(f)}, \hspace{3ex}\forall t \leq d-1.
    \end{equation}
\end{theorem}

Here, we use $f\sim \mathcal{F}$ to denote that $f$ is drawn uniformly at random from $\mathcal{F}$.

\begin{proof}
    We prove this theorem using a distinguisher-based approach as mentioned in the Definition~\ref{def:structured-sample-amplification}.   
    We use the simple distinguisher in \Cref{alg:distinguisher-general-lower-bound}. 

    \begin{algorithm}
        \caption{Distinguisher $\mc A^f$}\label{alg:distinguisher-general-lower-bound}
        \algorithmicrequire{ A labeled sample $z$, description of $f$.} \\
        \algorithmicensure{ Accept or Reject} \\
        1. Let $V_z=\{f'\in\mc F : f' \text{ is consistent with } z\}$.\\
        2. If $V_z=\emptyset$, set $\hat{f}=f_0$, where $f_0\in\mathcal{F}$ is some arbitrary function fixed a priori. Otherwise, sample
       $\hat f$ uniformly at random from $V_z$.\\
    3. Accept if $\hat f=f$; otherwise reject.
    \end{algorithm}

    We start by noticing that
    \begin{equation}
        \begin{split}
            \epsilon_{SA}^{\star}(\mathcal{D}, t, t+1) &= \min_{T_{SA}} \sup_{f \in \mc F}\max_{\mc{A}'}\mathrm{Adv}(\mc{A}', \mc{D}, T_{SA}, t) \, \\
            &\stackrel{(a)}{\geq} \min_{T_{SA}} \sup_{f \in \mc F}\mathrm{Adv}(\mc{A}^f, \mc{D}, T_{SA}, t)\\
            &\stackrel{(b)}{\geq} \min_{T_{SA}}\Ex_{f\sim\mc{F}}\mathrm{Adv}(\mc{A}^f,\mc{D}, T_{SA}, t)\, ,
        \end{split}
    \end{equation}
    where $(a)$ consists in lower bounding the supremum over all distinguishers by choosing the specific distinguisher from \Cref{alg:distinguisher-general-lower-bound}, and $(b)$ consists in lower bounding a supremum by an average.
    Here, we use $\mc{A}^f$ to highlight that $\mc A$ has knowledge of the function $f$.
    We use the shorthand $\mathrm{Adv}(\mc{A}^f)$ for $\mathrm{Adv}(\mc{A}^f,\mc{D}, T_{SA}, t)$. 
    
    Let
    \begin{equation}
        Z_t=\big((X^{(1)},f(X^{(1)})),\ldots,
        (X^{(t)},f(X^{(t)}))\big),
    \end{equation}
    where $X^{(1)},\ldots,X^{(t)}\overset{\mathrm{i.i.d.}}{\sim}\mc U_n$.
    We use $z_t$ to denote a realization of $Z_t$.
    Then, we can write
    \begin{equation}\label{eq:advantage-lower-bound-jensen}
        \begin{split}
            \Ex_{f} \mathrm{Adv}(\mc{A}^f) &= \Ex_f \lra{\Pr_{Z_{t+1}}[\mc{A}^f(Z_{t+1}) = 1] - \Pr_{Z_{t}}[\mc{A}^f(T_{SA}(Z_{t})) = 1]} \\
            &\stackrel{(a)}{\geq} \lra{\Ex_f\left[\Pr_{Z_{t+1}}[\mc{A}^f(Z_{t+1}) = 1] - \Pr_{Z_{t}}[\mc{A}^f(T_{SA}(Z_{t})) = 1]\right]}
        \end{split}
    \end{equation}
    where $(a)$ follows by Jensen's inequality and convexity of the absolute value function. 
    By \Cref{prop:sa-con}, which we prove separately in the appendix, we have
    \begin{equation}
        \Ex_f\Pr_{Z_{t}}[\mc{A}^f(T_{SA}(Z_{t})) = 1] \leq \Ex_f\Pr_{Z_{t}}[\mc{A}^f(Z_{t}) = 1] \leq \Ex_f\Pr_{Z_{t+1}}[\mc{A}^f(Z_{t+1}) = 1]
    \end{equation}
    where the last inequality holds since the probability of a consistent learner that outputs a uniformly random consistent hypothesis (as in \Cref{alg:distinguisher-general-lower-bound}) producing the correct hypothesis does not decrease as the number of samples increases.
    Therefore, we can further lower bound \Cref{eq:advantage-lower-bound-jensen} as
    \begin{equation}
        \Ex_{f} \mathrm{Adv}(\mc{A}^f) \geq \Ex_f\left[\Pr_{Z_{t+1}}[\mc{A}^f(Z_{t+1}) = 1] - \Pr_{Z_{t}}[\mc{A}^f(Z_t) = 1]\right]\, .
    \end{equation}

    We now control the terms inside the expectation. By definition of the
    distinguisher in \Cref{alg:distinguisher-general-lower-bound}, for any
    $f\in\mc F$ and any $k\in\mathbb{N}$,
    \begin{equation}
        \Pr\lrq{\mc A^f(Z_k)=1}
        =
        \sum_{z_k}
        \Pr\lrq{Z_k=z_k}
        \frac{1}{n_{\mc F}(z_k)}.
    \end{equation}
    Therefore, taking $t=d-1$ in the previous inequality and coupling
    $Z_{d-1}$ and $Z_d$ by taking $Z_{d-1}$ to be the first $d-1$ samples
    of $Z_d$, we obtain
    \begin{align}
        \Ex_f \mathrm{Adv}(\mc A^f)
        &\geq
        \Ex_f\left[
            \Pr\lrq{\mc A^f(Z_d)=1}
            -
            \Pr\lrq{\mc A^f(Z_{d-1})=1}
        \right] \\
        &=
        \Ex_f\sum_{z_d}
        \Pr\lrq{Z_d=z_d}
        \left(
            \frac{1}{n_{\mc F}(z_d)}
            -
            \frac{1}{n_{\mc F}(z_{d|d-1})}
        \right),
    \end{align}
    where $z_{d|d-1}$ denotes the restriction of $z_d$ to its first
    $d-1$ samples.
    
    Since adding an additional labeled example can only decrease the
    number of hypotheses consistent with the sample,
    \begin{equation}
        1\leq n_{\mc F}(z_d)
        \leq n_{\mc F}(z_{d|d-1}).
    \end{equation}
    Moreover,
    \begin{equation}
        \frac{1}{n_{\mc F}(z_d)}
        -
        \frac{1}{n_{\mc F}(z_{d|d-1})}
        \geq
        \frac{1}{2}
        \mathbf{1}_{\{
            n_{\mc F}(z_{d|d-1})\geq 2,\,
            n_{\mc F}(z_d)=1
        \}}.
    \end{equation}
    Hence,
    \begin{align}
        \Ex_f \mathrm{Adv}(\mc A^f)
        &\geq
        \frac{1}{2}
        \Ex_f
        \Pr\left[
            n_{\mc F}(Z_{d|d-1})\geq 2,\,
            n_{\mc F}(Z_d)=1
        \right] \\
        &=
        \frac{1}{2}
        \Ex_f
        \left[
            p_{\mc F}^{d}(f)-p_{\mc F}^{d-1}(f)
        \right].
    \end{align}
    Overall, we have shown that 
    \begin{equation}
        \epsilon_{SA}^{\star}(\mathcal{D}, d-1, d)
        \geq \Ex_f \mathrm{Adv}(\mc A^f)
        \geq \frac{1}{2}\Ex_f\left[p_{\mc F}^{d}(f) - p_{\mc F}^{d-1}(f)\right]\, .
    \end{equation}
    It remains to observe that, due to monotonicity of sample amplification,
    \begin{equation}
        \epsilon_{SA}^{\star}(\mathcal{D}, t_1, t_1+1) \geq \epsilon_{SA}^{\star}(\mathcal{D}, t_2, t_2+1), \hspace{3ex} \forall t_1 \leq t_2.
    \end{equation}
     This can be understood by noting that if there is a sample amplification procedure by one sample from $t_1$ samples, than this implies sample amplification from any $t_2 \geq t_1$ because we can use sample amplification on first $t_1$ samples out of $t_2$ to produce an extra sample. Hence, the proven lower bound for sample amplification from $d-1$ to $d$ samples carries over to sample amplification from $t$ to $t+1$ for any $t\leq d-1$.
\end{proof}

We can now obtain point 1 from Theorem~\ref{thm:classical-results-informal} as a corollary of our general result above by adding a suitable symmetry assumption on the function class. To state this result, for any $f\in\mc F$, define
\begin{equation}
    b_{\mc F}(f)
    :=
    \min\left\{
        t\in\mathbb{N}
        ~\middle|~
        \exists x_t\in(\Z_2^n)^t
        \text{ such that }
        n_{\mc F}\big((x_t,f(x_t))\big)=1
    \right\}.
\end{equation}
Thus, $b_{\mc F}(f)$ is the minimum size of a teaching sequence for
$f$ in $\mc F$. Finally, define
\begin{equation}
    b_{\mc F}
    :=
    \max_{f\in\mc F} b_{\mc F}(f),
    \qquad
    p_{\mc F}
    :=
    \min_{f\in\mc F} p_{\mc F}^{b_{\mc F}}(f).
\end{equation}
The quantity $b_{\mc F}$ is the teaching dimension of the function
class $\mc F$~\cite{GOLDMAN199520}. Since the input distribution is
uniform, $p_{\mc F}^t(f)$ is equivalently the probability that $t$
i.i.d.~uniformly sampled inputs yield a teaching sequence for $f$.

\begin{corollary}\label{cor:semi-general-lower-bound-symmetric}
    Suppose the function class $\mathcal{F}$ satisfies the following two symmetry conditions: 
    \begin{itemize}
        \item Uniform minimum teaching sequence size: $b_\mathcal{F}(f)=b_\mathcal{F}(f')=b_\mathcal{F}$ holds for all $f,f'\in\mathcal{F}$.
        \item Uniform fraction of teaching sequences: $p_\mathcal{F}^{b_\mathcal{F}(f)}(f) = p_\mathcal{F}^{b_\mathcal{F}(f')}(f')=p_\mathcal{F}$ holds for all $f,f'\in\mathcal{F}$.
    \end{itemize}
    Then, the minimax sample amplification error for the distribution class $\mc{D} = \{(\mc{U}_n, f)\}_{f \in \mc{F}}$ over the domain $\Z_2^n \times \Z_2$ in amplifying $t$ samples to $t+1$ samples for any $t \leq b_{\mc{F}} - 1$ satisfies
    \begin{equation}
         \epsilon_{SA}^{\star}(\mathcal{D}, t, t+1) \geq p_{\mc F}/2, \hspace{3ex} \forall t \leq b_{\mathcal{F}} - 1.
    \end{equation}
\end{corollary}

\begin{proof}
    Since $b_{\mc F}(f)=b_{\mc F}$ for every $f\in\mc F$, no sample
    of size $b_{\mc F}-1$ uniquely determines any $f$, and hence
    $p_{\mc F}^{b_{\mc F}-1}(f)=0$. By assumption,
    $p_{\mc F}^{b_{\mc F}}(f)=p_{\mc F}$ for every $f$.
    The result now follows from
    \Cref{thm:general-lower-bound} with $d=b_{\mc F}$.
\end{proof}

We now apply the above lower bounds for structured sample amplification to the function class of parities and more generally to linear spaces of functions. We will omit the distribution class from the notation $\epsilon_{SA}^{\star}(\mc D, t, t+1)$ wherever it is clear from the context.

\subsection{Parity Lower Bound}\label{subsec:amplification-parity-lower-bound}
We study structured sample amplification for the $n$-bit parity distributions $\mc{D}_{\mathrm{par}} = \{(\mc{U}_n, f_a)\}_{a \in \Z_2^n}$, where
\begin{equation}
   f_a:\mathbb{Z}_2^n\to\mathbb{Z}_2,\, f_a(x) = a\cdot x \mod 2\, ,
\end{equation}
and where $\mc{U}_n$ is the uniform distribution on $\Z_2^n$. We look at a useful lemma first.

\begin{lemma}\label{lemma:lin_ind_parities}
    Let $X_1,\ldots,X_k$ be drawn independently and uniformly from
    $\Z_2^n$. Then, for $0\leq\alpha\leq\min\{n,k\}$,
    \begin{equation}
        \Pr\left[
            \mathrm{rank}\{X_1,\ldots,X_k\}=\alpha
        \right]
        =
        \frac{
            \prod_{i=0}^{\alpha-1}(2^n-2^i)(2^k-2^i)
        }{
            2^{nk}\prod_{i=0}^{\alpha-1}(2^\alpha-2^i)
        }.
    \end{equation}
\end{lemma}

Thus, the probability of obtaining $n$ linearly independent $n$-bit strings from $n$ uniformly random draws is
\begin{equation}
    \Pr(n, n) = \prod_{i=1}^n (1 - 2^{-i}) \, .
\end{equation}

The probability $\Pr(n, n)$ is lower bounded by the constant $0.28$ for all $n$~\cite{oeis, finch2003mathematical}. Using this lemma and the techniques developed above, we provide the following constant lower bound on the minimax sample amplification error for $\mathcal{D}_{\mathrm{par}}$.

\begin{corollary}[Formal Statement of Theorem~\ref{thm:classical-results-informal}, Point 2: Error Lower Bound for Sample Amplification of Parities]\label{corollary:parities-sa-lb}
    Let $\mc{D}_{\mathrm{par}} $ be the class of $n$-bit parity distribution as defined above. Then the minimax sample amplification error for this class of distributions can be lower bounded as
    \begin{equation}
        \epsilon^{\star}_{SA}(\mc{D}_{\mathrm{par}} ,t-1, t)  \geq 0.14, \hspace{3ex} \forall 1\leq t \leq n\, .
    \end{equation}
\end{corollary}

The above corollary rules out structured sample amplification of $n$-bit parity distributions to arbitrary small error from $t \leq n-1$ samples. This can be formalised as follows:

\begin{corollary}
    The class $\mc{D}_{\mathrm{par}}$ of $n$-bit parity distributions does not admit a $(t-1, t, \epsilon)$ sample amplification procedure for any $t \leq n$ and $\epsilon < 0.14$.
\end{corollary}

Notice that the TV distance between any two distinct parity
distributions is $1/2$, independent of $n$. Consequently, learning an
$n$-bit parity distribution to sufficiently small constant TV error
has sample complexity $\Theta(n)$, matching the sample complexity of
exact parity learning from uniformly random inputs up to constant
factors. Thus, the above corollaries show that sample amplification of
parity distributions with sufficiently small constant error is no
easier than learning.

\begin{proof}[Proof of \Cref{corollary:parities-sa-lb}]
    The teaching dimension for the function class $\mathcal{F}_{\mathrm{par}}=\{f_a\}_a$ equals $n$. More precisely, for any parity function, a sequence of data points is a teaching sequence if and only if it contains $n$ linearly independent input strings. Thus, $b_{\mathcal{F}_{\mathrm{par}}}(f_a) = b_{\mathcal{F}_{\mathrm{par}}}=n$ and $p_{\mathcal{F}_{\mathrm{par}}}^{b_{\mathcal{F}_{\mathrm{par}}}(f_a)}(f_a) = p_{\mathcal{F}_{\mathrm{par}}}^n(f_a)=p_{\mathcal{F}_{\mathrm{par}}}$ holds for all $f_a\in \mathcal{F}_{\mathrm{par}}$. In particular, $\mathcal{F}_{\mathrm{par}}$ satisfies the two symmetry assumptions from \Cref{cor:semi-general-lower-bound-symmetric}, so we immediately conclude that
    \begin{equation}
        \epsilon_{SA}^\star (\mathcal{D}_{\mathrm{par}}, t-1,t)\geq p_{\mathcal{F}_{\mathrm{par}}}/2,\hspace{3ex} \forall 1\leq t \leq n\, ,
    \end{equation}
    where $p_{\mathcal{F}_{\mathrm{par}}}$ is the probability that a uniformly random sample of size $n$ contains $n$ linearly independent inputs. Using \Cref{lemma:lin_ind_parities} and the discussion thereafter, we get
    \begin{equation}
        p_{\mathcal{F}_{\mathrm{par}}}
        = \Pr(n,n)\geq 0.28 .
    \end{equation}
    Plugging this back into the above sample amplification error lower bound finishes the proof.
\end{proof}

The reasoning used above can be adapted to show the same constant sample amplification error lower bound of $0.14$ for the class of distributions that are uniform over some $n$-dimensional subspace in $\mathbb{Z}_2^m$. We can prove this also by using a distinguisher based approach. The distinguisher takes $t\leq n$ samples, chooses an $n$-dimensional subspace containing the $t$-samples uniformly at random, and accepts only if it agrees with the support of true distribution. This distinguisher has similarities to the one used in the function case (\Cref{alg:distinguisher-general-lower-bound}), which, when given $t$ samples, obtains a uniformly random consistent hypothesis and accepts only if it is the true hypothesis.

\begin{corollary}
    Let $\mc D_{m, n}$ with $m > n$ denote the class of uniform distributions supported on an $n$-dimensional subspace of $\Z^m_2$. Then, the minimax sample amplification error for this class of distributions can be lower bounded as
    \begin{equation}
        \epsilon^{\star}_{SA}(\mc{D}_{m, n} ,t-1, t) \geq 0.14, \hspace{3ex} \forall 1\leq t \leq n\, .
    \end{equation}
\end{corollary}
\begin{proof}
    We again prove this by a distinguisher-based approach. We prove the statement for $t = n$, which can then directly be extended to any $t \leq n$ because of monotonicity of the minimax sample amplification error, i.e.,
    \begin{equation}
        \epsilon_{\mathrm{SA}}^{\star}\lr{\mc D_{m, n}, t-1, t} \geq \epsilon_{\mathrm{SA}}^{\star}\lr{\mc D_{m, n}, n-1, n}\, ,\hspace{3ex} \forall t \leq n\, ,
    \end{equation}
    which follows because any sample amplification scheme for $t\leq n$ is also a valid sample amplification scheme for $t =n$.
    
    As mentioned previously, we allow distinguishers that have complete knowledge of the true distribution. We consider the following distinguisher (analogous to \Cref{alg:distinguisher-general-lower-bound}): The distinguisher takes the input sample, chooses an $n$-dimensional subspace that contains all of the given samples (i.e., a consistent subspace) uniformly at random, and accepts if the obtained subspace is same as the support of true distribution. Then, we can use similar arguments as in \Cref{thm:general-lower-bound} and \Cref{prop:sa-con} to lower bound the sample amplification error via the advantage of this distinguisher as 
    \begin{equation}
         \begin{split}
             \epsilon_{\mathrm{SA}}^{\star}\lr{\mc{D}_{m, n}, n-1, n} &\geq \lra{\Ex_{D \sim \mc{D}_{m, n}} \lrq{\Pr_{Z_{n}}[\mc{A}^D(Z_{n}) = 1]} -  \Ex_{D \sim \mc{D}_{m, n}} \lrq{\Pr_{Z_{n-1}}[\mc{A}^D(Z_{n-1}) = 1]}}\, \\
             &\geq \frac{1}{2}\Ex_{D \sim \mc{D}_{m, n}}\lrq{p_{\mc{D}_{m, n}}^n(D) - p_{\mc{D}_{m, n}}^{n-1}(D)}\, ,
         \end{split}
    \end{equation}
    where $p_{\mc{D}_{m, n}}^t(D)$ is the probability of uniquely identifying the support of $D$, or, equivalently, the probability of obtaining a basis for the support of $D$ in $t$ samples. It is easy to see that the class $\mc{D}_{m, n}$ is symmetric in the sense that the probabilities $p_{\mc{D}_{m, n}}^t(D)$ do not depend on the particular choice of $D$. Thus, we get
    \begin{equation}
         \epsilon_{\mathrm{SA}}^{\star}\lr{\mc{D}_{m, n}, n-1, n} \geq \frac{p_{\mc{D}_{m, n}}^n(D)}{2} \geq 0.14\, ,
    \end{equation}
    where first step follows by noting $p_{\mc{D}_{m, n}}^{n-1}(D) = 0$ and the second step uses \Cref{lemma:lin_ind_parities}.
\end{proof}

\subsection{Coding Theory Lower Bounds}\label{subsec:amplification-coding-theory}
In this section, we highlight connections between structured sample amplification and coding theory. We first recall that, in principle, for any Boolean function class, we can define an associated linear code (Definition~\ref{defn:linear_function_code}). Then, we show that the minimax sample amplification error for structured distributions $\{(\mc U_n, f)\}$ with $f\in\mc F$ can be related to the ability of the corresponding dual code to correct random erasures. 

We start with the following relation between the rank of a submatrix obtained by choosing random columns of parity check matrix of a code and the ability of the primal code to correct random erasures:

\begin{lemma}[\cite{abbe2014reedmullercodesrandomerasures}, Lemma 2.8]\label{lemma:prop_2_codes}
    For a parity check matrix $H$ with $N$ columns, take $S \subseteq [N]$, and denote by ${[N] \choose s}$ the collection of all $s$-element subsets of $[N]$. Let $H[S]$ be the restriction of $H$ to columns indexed by elements in $S$, and let $x[S]$ be the restriction of string $x$ to indices in $S$. Then, the set of bad $s$-erasure patterns is given as
    \begin{equation}
        \left\{S \in  {[N] \choose s} : \exists x, y \in \ker(H), x \neq y, x[S^c] = y[S^c]\right\} = \left\{S \in  {[N] \choose s}: \mathrm{rank}(H[S]) < s \right\} \, .
    \end{equation}
\end{lemma}

Here, the set of bad $s$-erasures is the set of erasures at $s$ many locations that cannot be corrected. In particular, the above lemma says that the fraction of bad $s$-erasures equals the probability that the submatrix formed by randomly sampling $s$-columns of the parity check matrix is rank-deficient. Thus, if a code can correct $d$ random erasures with high probability, then a uniformly random set of $d$ columns from its parity check matrix are linearly independent with high probability. 
As in the parity case above, a high (or at least constant) probability of linear independence will be crucial in proving a sample amplification lower bound.

We first define the special class of codes that we will study in this subsection.

\begin{definition}[Linear Function Codes]\label{defn:linear_function_code}
    Let $\mathcal{F}$ be a linear space of Boolean functions with dimension $k_{\mathcal{F}}$ and with a chosen basis $B_{\mathcal{F}}=\{e_{\mc{F}, i}\}_{i=1}^{k_{\mc F}}$. 
    We define the \emph{linear function code} $C_{\mathcal{F}}$ with parameters $[2^n, k_{\mathcal{F}}]$ via the encoding map
    \begin{equation}
    	\text{Enc}:\mathbb{Z}_2^{k_{\mc F}}\to\mathbb{Z}_2^{2^n} \, , 
        \text{Enc} (a) = (f_a(x))_{x\in \mathbb{Z}_2^n}\, , 
    \end{equation}
    where we define $f_a = \sum_{i=1}^{k_{\mc F}} a_i e_{\mc{F}, i}$.
    That is, a $k_{\mc F}$-bit string is interpreted as a vector of basis coefficients and is encoded into the vector of evaluations of the corresponding function in $\mc F$.
\end{definition}

The generator matrix $G(C_{\mathcal{F}})$ of such a code is of size $k_{\mathcal{F}} \times 2^n$. Its columns are indexed by elements in $\mathbb{Z}_2^n$, and its rows are indexed by basis functions in $e_{\mc{F}, i}$. Then, $G(i, j) = e_{\mc{F}, i}(x_j)$:
    \[
    \begin{array}{c|cccc}
    & x_0 & \cdots & \cdots & x_{2^{n}-1} \\ \hline
    e_{\mathcal{F}, 1} &     &     &        &     \\
    e_{\mathcal{F}, 2} &     &     &        &     \\
    \vdots & & &    &        \\
    e_{\mc{F}, i}&     & {G(i,j)=e_{\mc{F}, i}(x_j)} & &     \\
    \vdots & & &    &        \\
    e_{\mathcal{F}, {k_{\mathcal{F}}}} &     &     &        &     
    \end{array}
    \]

We will denote the parameters of the code as $[2^n, k_{\mathcal{F}}, d]$ with their usual meanings. Moreover, we will denote the dual of the above code as $C_{\mathcal{F}}^{\perp}$, which is given by the kernel of the generator matrix $G(C_{\mathcal{F}})$. A familiar example of linear function codes is the $(n, d)$ Reed-Muller code, where $\mathcal{F}$ is the class of degree-$d$ polynomials over $\mathbb{Z}_2^n$, the basis $B_{\mathcal{F}}$ consists of all monomials of degree at most $d$, and the dimension is $k_{\mathcal{F}} = {n \choose \leq d}$. We now show that a certain resistance to random erasures in the code $C_{\mathcal{F}}^{\perp}$ implies lower bounds for sample amplification of the distribution class $\mc{D} = \{(\mc{U}, f)\}_{f \in \mc{F}}$.

\begin{theorem}[Formal Statement of Theorem~\ref{thm:classical-results-informal}, Point 3: Sample Amplification Lower Bounds From Coding Theory]\label{thm:coding-theory-general}
    Let $C_{\mathcal{F}}$ be a $[2^n, k_{\mathcal{F}}, d]$ linear function code as defined above, and let $C_{\mathcal{F}}^{\perp}$ be its dual code. Then, the following hold:
    \begin{enumerate}
        \item If $C^{\perp}_\mathcal{F}$ can correct $k_{\mathcal{F}}$ random erasures with probability $p_1$ over the erasures, then the minimax sample amplification error for the distribution class $\{(\mc U_n, f)\}_{f \in \mc F}$ satisfies
        \begin{equation}
            \epsilon^{\star}_{SA}(\{(\mc U_n, f)\}_{f \in \mc F}, t-1, t) \geq (p_1/2)\cdot\left(1 - \frac{(k_{\mc F}-1)k_{\mc F}}{2^{n+1}}\right), \hspace{3ex} \forall 1\leq t \leq k_{\mathcal{F}}\, .
        \end{equation}
        \item If $C_{\mathcal{F}}^{\perp}$ can correct $k_{\mathcal{F}}-1$ random erasures with probability $p_2$ over the erasures, then the minimax sample amplification error for the distribution class $\{(\mc U_n, f)\}_{f \in \mc F}$ satisfies
        \begin{equation}
            \epsilon^{\star}_{SA}(\{(\mc U_n, f)\}_{f \in \mc F}, t-1, t) \geq p_2 \cdot \left(1 - \frac{(k_{\mc F}-2)(k_{\mc F}-1)}{2^{n+1}}\right)\cdot \frac{d}{2^{n+1}}, \hspace{3ex} \forall 1\leq t \leq k_{\mathcal{F}}.
        \end{equation}
    \end{enumerate}
\end{theorem}

 To prove Theorem~\ref{thm:coding-theory-general}, we first prove two simple technical lemmas.

\begin{lemma}\label{lemma:prop_1_codes}
    Let $C_{\mc F}$ be a $[2^n,k_{\mc F},d]$ linear function code.
    Then, for every nonzero $f\in\mc F$,
    \begin{equation}
        \Pr_{x\sim\mc U_n}[f(x)=0]
        \leq 1-\frac{d}{2^n}.
    \end{equation}
\end{lemma}

\begin{proof}
    For any nonzero $f\in\mc F$, the corresponding codeword
    \begin{equation}
        \text{Enc}(f)=(f(x))_{x\in\Z_2^n}
    \end{equation}
    is a nonzero codeword of $C_{\mc F}$. Since $C_{\mc F}$ has
    minimum distance $d$, every nonzero codeword has Hamming weight at
    least $d$. Therefore,
    \begin{equation}
        |\{x\in\Z_2^n:f(x)=1\}|\geq d,
    \end{equation}
    and hence
    \begin{equation}
        \Pr_{x\sim\mc U_n}[f(x)=0]
        =
        1-\Pr_{x\sim\mc U_n}[f(x)=1]
        \leq
        1-\frac{d}{2^n}.
    \end{equation}
\end{proof}
\begin{lemma}\label{lemma:prop_3_codes}
    Let $C_{\mathcal{F}}$ be a $[2^n, k_{\mathcal{F}}, d]$ linear function code with generator matrix denoted as $G$ such that,
    \begin{equation}
        \Pr_{\stackrel{S=\{x_i\}_{i=1}^{k_{\mc F}-1}}{x_i\overset{iid}{\sim} \mathcal{U}_n}} \left[ \mathrm{rank}(G[S]) = k_{\mathcal{F}}-1\right] = p \, .
    \end{equation}
    Then 
    \begin{equation}
         \Pr_{\stackrel{S=\{x_i\}_{i=1}^{k_{\mc F}}}{x_i\overset{iid}{\sim} \mathcal{U}_n}} \left[ \mathrm{rank}(G[S]) = k_{\mathcal{F}}\right] \geq p \cdot \frac{d}{2^n}\, .
    \end{equation}
\end{lemma}
\begin{proof}
     We define an indicator random variable $X_m$ which takes the value $1$ if the first $m$ sampled columns are linearly independent. Then, the probability of interest is $\Pr (X_{k_{\mathcal{F}}} = 1)$ which can be written as,
    \begin{equation}
        \Pr\lrq{X_{k_{\mathcal{F}}} = 1} =  \Pr\lrq{X_{k_{\mathcal{F}} - 1} = 1} \times \Pr\lrq{\text{new linearly independent sample} ~|~ X_{k_{\mathcal{F}} -1} = 1}.
    \end{equation}
    Conditioned on the first $k_{\mathcal{F}}-1$ samples $\{x_1, \ldots x_{k_{\mc F} - 1}\}$ indexing $k_{\mathcal{F}}-1$ many linearly independent columns $\{g(x_1), \ldots g(x_{k_{\mc F} - 1)}\}$ of the generator matrix $G$, we can identify $a$ as the unique orthogonal direction to these columns, satisfying $a \cdot g(x_i) = 0, \forall i$.
    Now, the next sample indexes a linearly dependent column iff it is orthogonal to $a$. We can consider function $f_a(x) = a \cdot g(x)$, where $g(x)$ again denotes the column indexed by $x$ in the generator matrix $G$. In this notation, the next sample $x_{k_\mathcal{F}}$ indexes linearly dependent column iff $f_a(x_{k_\mathcal{F}})= 0$. Using \Cref{lemma:prop_1_codes}, we have that
    \begin{equation}
        \Pr_{x}[f_a(x) = 0] \leq 1 - \frac{d}{2^n}.
    \end{equation}
    Hence, the probability that the new sample will index a linearly dependent column is at most $1 - d/2^n$. 
    Thus 
    \begin{equation}
         \Pr_{\stackrel{S=\{x_i\}_{i=1}^{k_{\mc F}}}{x_i\overset{iid}{\sim} \mathcal{U}_n}} \lrq{\mathrm{rank}(G[S]) 
         = k_{\mathcal{F}}} 
         = \Pr [X_{k_\mathcal{F}-1}=1]\cdot \Pr [X_{k_\mathcal{F}}=1|X_{k_\mathcal{F}-1}=1]
         \geq p \cdot \frac{d}{2^n}.
    \end{equation}
\end{proof}

We are now well-equipped to prove Theorem~\ref{thm:coding-theory-general}.
\begin{proof}[Proof of Theorem~\ref{thm:coding-theory-general}]
    We first prove bullet point 1. Note that if $C_{\mathcal{F}}^{\perp}$ can correct $k_{\mathcal{F}}$ random erasures with probability $p_1$, then from Lemma~\ref{lemma:prop_2_codes} we have
    \begin{align}
        \Pr_{\stackrel{S=\{x_i\}_{i=1}^{k_{\mc F}}}{x_i\overset{iid}{\sim} \mathcal{U}_n}} \left[ \mathrm{rank}(G[S]) = k_{\mathcal{F}}\right]
        &= \Pr_{S\sim {[2^n] \choose k_{\mathcal{F}}}}\left[\mathrm{rank}(G[S]) = k_{\mathcal{F}}\right]\cdot \Pr_{x_1,\ldots, x_{k_{\mc F}}\overset{iid}{\sim} \mathcal{U}_n}\left[x_i\neq x_j~\forall i\neq j\right]\\
        &\geq p_1 \cdot \frac{2^n(2^n-1)\ldots(2^n-k_{\mc F}+1)}{2^{nk_{\mc F}}}\\
        &\geq p_1 \cdot \prod_{i=1}^{k_{\mc F}-1} \left(1 - \frac{i}{2^n}\right)\\
        &\overset{(a)}{\geq} p_1 \cdot\left(1 - \sum_{i=1}^{k_{\mc F}-1}\frac{i}{2^n}\right)\\
        &= p_1 \cdot\left(1 - \frac{(k_{\mc F}-1)k_{\mc F}}{2^{n+1}}\right) \, ,
    \end{align} 
    where $G$ is the parity check matrix of $C_{\mathcal{F}}^{\perp}$, i.e., the generator matrix of $C_{\mathcal{F}}$, and where $(a)$ used the auxiliary \Cref{lemma:math_inequal} proved below.
    So, if we sample columns of $G$ uniformly at random, then, up to a collision correction to account for sampling with vs without replacement, with probability $p_1$, they are linearly independent. 
    Interpreting functions in $\mc F$ as computing parities of vectors of basis function evaluations, we see that to uniquely identify any function in $\mathcal{F}$, it is necessary and sufficient to see its values on inputs that index $k_{\mc F}$ many linearly independent columns of $G$.
    Hence, we can reinterpret the above as $b_{\mc F}(f) = b_{\mc F} = k_{\mc F}$ and $p_{\mc F}^{b_\mathcal{F}}(f)= p_{\mc F}$ for all $f\in\mathcal{F}$, as well as $p_{\mc F}\geq p_1 \cdot\left(1 - \frac{(k_{\mc F}-1)k_{\mc F}}{2^{n+1}}\right)$.
    We can now appeal to \Cref{cor:semi-general-lower-bound-symmetric} and obtain
    \begin{equation}
        \epsilon^{\star}_{SA}(\{(\mc U_n, f)\}_{f\in \mc F},t-1, t) \geq (p_1/2)\cdot\left(1 - \frac{(k_{\mc F}-1)k_{\mc F}}{2^{n+1}}\right),\hspace{3ex} \forall 1\leq t \leq k_{\mathcal{F}}.
    \end{equation}

    For the second part, again using Lemma~\ref{lemma:prop_2_codes}, we first note that if $C_{\mathcal{F}}^{\perp}$ can correct $k_{\mathcal{F}}-1$ random erasures with probability $p_2$, then, by the same line of reasoning as above
    \begin{equation}
        \Pr_{\stackrel{S=\{x_i\}_{i=1}^{k_{\mc F}}-1}{x_i\overset{iid}{\sim} \mathcal{U}_n}} \left[ \mathrm{rank}(G[S]) = k_{\mathcal{F}}-1\right]
        \geq p_2\cdot \left(1 - \frac{(k_{\mc F}-2)(k_{\mc F}-1)}{2^{n+1}}\right).
    \end{equation}
    Then, using Lemma~\ref{lemma:prop_3_codes}, we get
    \begin{equation}
        \Pr_{\stackrel{S=\{x_i\}_{i=1}^{k_{\mc F}}}{x_i\overset{iid}{\sim} \mathcal{U}_n}} \left[ \mathrm{rank}(G[S]) = k_{\mathcal{F}}\right] \geq p_2 \cdot \left(1 - \frac{(k_{\mc F}-2)(k_{\mc F}-1)}{2^{n+1}}\right)\cdot \frac{d}{2^n}.
    \end{equation}
    So, using \Cref{cor:semi-general-lower-bound-symmetric}, we have
    \begin{equation}
        \epsilon^{\star}_{SA}(\{(\mc U_n, f)\}_{f \in \mc F},t-1, t) \geq p_2 \cdot \left(1 - \frac{(k_{\mc F}-2)(k_{\mc F}-1)}{2^{n+1}}\right)\cdot \frac{d}{2^{n+1}}.
    \end{equation}
\end{proof}

The above theorem has two consequences, one for coding theory and one for the problem of sample amplification. Assume that $k_{\mathcal{F}}=o(2^{n/2})$. On the one hand, if $C_{\mathcal{F}}$ is a linear function code with parameters $[2^n, k_{\mathcal{F}}, d]$ that has constant fractional distance $d/2^n = \Theta(1)$ and $C^{\perp}_{\mathcal{F}}$ can correct $k_{\mathcal{F}} - 1$ erasures with high, constant probability, the corresponding distribution class $\{(\mc U_n, f)\}_{f \in \mc F}$ cannot be sample amplified with error less than a constant threshold, starting with less than $k_{\mathcal{F}}$ samples. On the other hand, if there exists a function class for which the distribution class $\{(\mc U_n, f)\}_{f \in \mc F}$ is easy to sample amplify even from $t \leq k_{\mathcal{F}}$ samples, with error bounded from above by $\delta$ which decays with $n$, then the corresponding dual code can correct $k_{\mathcal{F}}$ erasures only with probability at most $2\delta$, up to collision corrections.

\begin{corollary}
    Take $k_\mathcal{F}$ and $n$ such that $1 - \frac{(k_{\mc F}-1)k_{\mc F}}{2^{n+1}}>0$.
    Let $\{(\mc U_n, f)\}_{f \in \mc F}$ be a distribution class (where $\mathcal{F}$ is a linear function space) such that the minimax sample amplification error satisfies
    \begin{equation}
        \epsilon_{SA}^{\star}(\{(\mc U_n, f)\}_{f \in \mc F}, k_{\mathcal{F}}-1, k_{\mathcal{F}}) \leq \delta(n) ,
    \end{equation}
    where $\delta$ is some function of $n$. Then, the dual linear function code $C_{\mathcal{F}}^\perp$ can correct $k_{\mathcal{F}}$ random erasures with probability at most $2\delta(n)\left(1 - \frac{(k_{\mc F}-1)k_{\mc F}}{2^{n+1}}\right)^{-1}$, and it can correct $k_{\mathcal{F}} - 1$ random erasures with probability at most $(2^{n+1}/d) \delta(n) \left(1 - \frac{(k_{\mc F}-2)(k_{\mc F}-1)}{2^{n+1}}\right)^{-1}$.
\end{corollary}

\subsection{Sample Amplification for the Class of all Boolean Functions}\label{sec:sa-all-functions}

In the previous sections, we saw examples of function classes for which structured sample amplification (with small error) is hard as learning. This provides evidence that when enough structure is present, the best strategy is to learn. 
In this section we provide a simple example of class of distributions based on Boolean functions for which sample amplification is strictly easier than learning in the sample complexity sense. 

We consider the distribution class $\mc{D} = \{(\mc U_n, f)\}_{f \in \mc F_{\mathrm{all}}}$ where $\mc F_{\mathrm{all}}$ is the class of all Boolean functions over $\Z_2^n$. Note that our lower bound technique fails to provide a meaningful lower bound in this case because the teaching dimension for $\mc F_{\mathrm{all}}$ is $2^n$  and because $p_{\mc F}$ is exponentially small. In fact, we can use the result for sample amplification of any discrete distribution~\cite{axelrod2019sampleamplificationincreasingdataset} to sample amplify $\mc F_{\mathrm{all}}$.

\begin{corollary}
        Let $\mathcal{F}$ be the class of all Boolean functions $f:\mathbb{Z}_2^n \rightarrow \mathbb{Z}_2$. Then, the class of distributions $\{(\mc U_n, f)\}_{f \in \mathcal{F}}$ admits a $(t, t+1, \epsilon)$ sample amplification scheme with $t = \Theta\left(\frac{\sqrt{2^n}}{\epsilon}\right)$.
\end{corollary}
\begin{proof}
    As the class of distributions $\mc{D} = \{(\mc U_n, f)\}_{f \in \mc F_{\mathrm{all}}}$ consists of distributions with (unknown) support of size $2^{n}$, this immediately follows from \cite[Theorem 1]{axelrod2019sampleamplificationincreasingdataset} and \cite[Example A.9]{axelrod2024statisticalcomplexitysampleamplification}.
\end{proof}

In particular, the complexity of sample amplification for $\mc F_{\mathrm{all}}$ is quadratically better than the complexity of exactly learning the same class, which by coupon collector is $\widetilde{\Theta}(2^n)$, and of learning to constant total variation error, which is $\Theta(2^n)$.

\section{Structured Cloning}\label{section:sc}
We can define structured quantum cloning analogously as structured sample amplification.

\begin{definition}[Structured Quantum Cloning]\label{defn:structured-quantum-cloning}
    Let $\mc S$ be a class of $n$-qubit states. Then, $\mc S$ is said to admit $(t, t+m, \epsilon)$-quantum cloning scheme if there exists a CPTP map $\Lambda_{S, t, m, \epsilon}:\mc{B}((\mathbb{C}^{2^n})^{\otimes t}) \rightarrow \mc{B}((\mathbb{C}^{2^n})^{\otimes t+m})$ such that,
    \begin{equation}
        \sup_{\rho \in \mc{S}} \td  \lr{\Lambda_{\mc{S}, t, m, \epsilon}(\rho^{\otimes t}), \rho^{\otimes t+m}} \leq \epsilon,
    \end{equation}
    and the optimal cloning error for the class of states $\mc S$ is defined by minimising the cloning error over all CPTP maps,
    \begin{equation}
        \epsilon^{\star}_{\Cl}(\mc S, t, t+m) := \min_{\Lambda}  \sup_{\rho \in \mc S} \td  \lr{\Lambda(\rho^{\otimes t}), \rho^{\otimes t+m}}.
    \end{equation}
\end{definition}

Here, $\td$ denotes the trace distance, defined as $\td(\rho,\sigma)=\frac{1}{2}\norm{\rho-\sigma}_1$, with $\norm{\cdot}_1$ the Schatten $1$-norm.
Structured cloning can also be thought of as a game between two parties, a cloner and a distinguisher where the distinguisher knows the quantum state. Then the optimal cloning error is the maximum advantage achievable by an adversarial distinguisher $\mc A$ in telling the true copies and cloned copies apart:
\begin{equation}
    \epsilon^{\star}_{\Cl}(\mc S, t, t+1)= \min_{\Lambda}  \sup_{\rho \in \mc S} \max_{\mc{A}} \mathrm{Adv}\lr{\mc{A}, \rho, \Lambda, t}\, ,
\end{equation}
where
\begin{equation}
    \mathrm{Adv}\lr{\mc{A}, \rho, \Lambda, t} = \lra{\Pr\lrq{\mc{A}^{\rho}(\rho^{\otimes t+1}) = 1} - \Pr\lrq{\mc{A}^{\rho}\lr{\Lambda\lr{\rho^{\otimes t}}} = 1}}\, ,
\end{equation}
where the superscript $\rho$ represents the fact that the distinguisher has complete knowledge of the state $\rho$. 

\subsection{Cloning Lower Bound for Stabilizer States}

We first prove our stabilizer state cloning lower bound (bullet point 2 in \Cref{thm:quantum-results-informal}).

\begin{theorem}[Cloning Lower Bound for Stabilizer States]
\label{thm:stab-states-clone}
    Let $\mathrm{Stab}_n$ be the set of pure $n$-qubit stabilizer
    states. Then,
    \begin{equation}
        \epsilon^\star_{\Cl}
        (\mathrm{Stab}_n,t,t+1)
        \geq
        0.28-O(2^{-n/3}),
        \qquad
        \forall\, t<\lfloor n/3\rfloor.
    \end{equation}
\end{theorem}

\begin{proof}
    For simplicity, we work with $(3n)$-qubit pure stabiliser states, and show that cloning these states from $t <n$ copies incurs a constant error. We prove this result via a similar distinguisher-based argument that we used to prove the classical sample amplification result. We first define a class of states, $\mathrm{BiGraph}_{3n}$ as follows:
    \begin{equation}
        \mathrm{BiGraph}_{3n} = \lrb{|A\rangle ~|~ A \in \Z_{2}^{n \times 2n}}\, ,
    \end{equation}
    where the matrix $A$ corresponds to the biadjacency matrix of a bipartite graph on $n+2n$ vertices, and the corresponding state is the $(n, 2n)$-bipartite graph state given by
    \begin{equation}
        |A\rangle = \frac{1}{\sqrt{2^{3n}}}\sum_{\substack{{x \in \Z_2^{n}}\\ {y \in \Z_2^{2n}}}}(-1)^{x^T A y} |x\rangle |y\rangle \, .
    \end{equation}
    Note that these states are stabilizer states with the stabilizer group specified by the biadjacency matrix $A$. 
    A set of independent generators for the stabilizer group of $|A\rangle$ is (see \Cref{lemma:prop_bigraph})
\begin{equation}
    \left\{
        X_i \prod_{j=1}^{2n} Z_{n+j}^{A_{i,j}}
    \right\}_{i=1}^{n}
    \cup
    \left\{
        \left(\prod_{j=1}^{n} Z_j^{A_{j,k}}\right)X_{n+k}
    \right\}_{k=1}^{2n},
\end{equation}
where the subscripts indicate the qubits on which the Pauli operators act. Now, we define the particular distinguisher that we use as follows:
    \begin{algorithm}[H]
        \caption{Distinguisher $\mc A$\label{alg:distinguisher-stab-lower-bound}} 
        \algorithmicrequire{ A state $\sigma$ on $(t+1)$-copy Hilbert space, biadjacency matrix $A$ describing the stabilizer group of a bipartite graph state.} \\
        \algorithmicensure{ Accept or Reject} \\
        1. Run an average-case optimal $(t+1)$-copy learner for $\mathrm{BiGraph}_{3n}$ on $\sigma$ to find the stabilizer group $\hat{A}$.\\
        2. Accept (i.e., output $1$) if $\hat{A} = A$, else reject (i.e., output $0$).
\end{algorithm}

The optimal cloning error is lower bounded by the distinguishing advantage of a distinguisher, i.e.,

\begin{equation}
    \begin{split}
        \epsilon_{\mathrm{Cl}}^{\star}\lr{\mathrm{Stab}_{3n}, t, t+1} 
        &=  \min_{\Lambda} \sup_{|\psi\rangle \in \mathrm{Stab}_{3n}} \max_{\mc{A}'} \mathrm{Adv}(\mc{A}', |\psi\rangle\langle \psi|, \Lambda, t) \\
        &\stackrel{(a)}{\geq}\min_{\Lambda} \sup_{|\psi\rangle \in \mathrm{Stab}_{3n}} \mathrm{Adv}(\mc{A} , |\psi\rangle\langle \psi|, \Lambda, t)\\
        &\stackrel{(b)}{\geq}\min_{\Lambda} \sup_{|\psi\rangle \in \mathrm{BiGraph}_{3n}} \mathrm{Adv}(\mc{A}, |\psi\rangle\langle \psi|, \Lambda, t) \\
        &\stackrel{(c)}{\geq}\min_{\Lambda} \Ex_{|\psi\rangle \sim \mathrm{BiGraph}_{3n}} \mathrm{Adv}(\mc{A},|\psi\rangle\langle \psi|,  \Lambda, t)\, ,
    \end{split}
\end{equation}
where $(a)$ follows by fixing the distinguisher to be the one in \Cref{alg:distinguisher-stab-lower-bound}, $(b)$ follows by replacing the supremum over $(3n)$-qubit stabilizer states by a supremum over a subset of these states, $\mathrm{BiGraph_{3n}}\subset \mathrm{Stab}_{3n}$, as defined above. Finally, $(c)$ follows by replacing the supremum by an average. Now, we have
\small
\begin{equation}
    \begin{split}
    \Ex_{|\psi\rangle \sim \mathrm{BiGraph}_{3n}} \mathrm{Adv}(\mc{A}, \Lambda, t, |\psi\rangle) 
    &=  \Ex_{|\psi\rangle \sim \mathrm{BiGraph}_{3n}} \lra{\Pr[\mc{A}^{|\psi\rangle}(|\psi\rangle^{\otimes t+1}) = 1] - \Pr[\mc{A}^{|\psi\rangle}(\Lambda(|\psi\rangle^{\otimes t})) = 1]} \\
    &\stackrel{(a)}{\geq} \lra{\Ex_{|\psi\rangle \sim \mathrm{BiGraph}_{3n}}\Pr[\mc{A}^{|\psi\rangle}(|\psi\rangle^{\otimes t+1}) = 1] - \Ex_{|\psi\rangle \sim \mathrm{BiGraph}_{3n}}\Pr[\mc{A}^{|\psi\rangle}(\Lambda(|\psi\rangle^{\otimes t})) = 1]}\, ,
    \end{split}        
    \end{equation}
    \normalsize
    where $(a)$ follows by Jensen's inequality. In the above we used the notation $\mc{A}^{|\psi\rangle}$ to highlight the fact that the distinguisher has complete knowledge of the input quantum state.

    Our distinguisher first runs an optimal ($t+1$)-copy learning algorithm on the given copies of state and accepts only if the learnt hypothesis is same as the true stabilizer group. Thus, the probability of acceptance of distinguisher is simply the probability that an optimal learner correctly identifies the stabilizer group. Mathematically, suppose that the distinguisher runs the learner $\mc{L}^{\star}$, then,
    \begin{equation}
        \Ex_A\Pr[\mc{A}\lr{|A\rangle^{\otimes t+1}} = 1] =  \Ex_A\Pr[\mc{L}^{\star}_{t+1}\lr{|A\rangle^{\otimes t+1}} = A]\, .
    \end{equation}
    By definition, the learner $\mc{L}^{\star}$ used by the distinguisher optimizes the average success probability over inputs, i.e.,
    \begin{equation}
        \Ex_A\Pr[\mc{L}^{\star}_{t+1}\lr{|A\rangle^{\otimes t+1}} = A] = \sup_{\mc{L}_{t+1}} \Ex_A\Pr[\mc{L}_{t+1}\lr{|A\rangle^{\otimes t+1}} = A]\, .
    \end{equation}
    Thus, we can write the advantage of the distinguisher as,
    \begin{equation}
        \begin{split}
            \Ex_{|\psi\rangle \sim \mathrm{BiGraph}_{3n}} \mathrm{Adv}(\mc{A}, \Lambda, t, |\psi\rangle) &= \Ex_{A \in \Z_2^{n \times 2n}} \mathrm{Adv}(\mc{A}, \Lambda, t, |A\rangle) \\
            &\geq \lra{\sup_{\mc{L}_{t+1}} \Ex_A\Pr[\mc{L}_{t+1}\lr{|A\rangle^{\otimes t+1}} = A] - \sup_{\mc{L}_{t+1}} \Ex_A\Pr[\mc{L}_{t+1}\lr{\Lambda\lr{|A\rangle^{\otimes t}}} = A]} \\
            &\stackrel{}{\geq} \sup_{\mc{L}_{t+1}} \Ex_A\Pr[\mc{L}_{t+1}\lr{|A\rangle^{\otimes t+1}} = A] - \sup_{\mc{L}_{t}} \Ex_A\Pr[\mc{L}_{t}\lr{|A\rangle^{\otimes t}} = A]\, ,
        \end{split}
    \end{equation}
    where the last step follows by making the observation
    \begin{equation}
        \sup_{\mc{L}_{t+1}} \Ex_A\Pr[\mc{L}_{t+1}\lr{\Lambda\lr{|A\rangle^{\otimes t}}} =A] \stackrel{}{\leq} \sup_{\mc{L}_{t}} \Ex_A\Pr[\mc{L}_{t}\lr{|A\rangle^{\otimes t}} = A] \stackrel{}{\leq} \sup_{\mc{L}_{t+1}} \Ex_A\Pr[\mc{L}_{t+1}\lr{|A\rangle^{\otimes t+1}} = A]\, ,
    \end{equation}
    which is proved in \Cref{prop:cl-con} below. Note that we used the notation $\mc{L}_t$ to denote a $t$-copy learner above. Thus, we get
    \begin{equation}\label{eq:cloning-stab-error-triangle-inequal}
        \epsilon_{\mathrm{Cl}}^{\star}\lr{\mathrm{Stab}_{3n}, t, t+1}\geq  \sup_{\mc{L}_{t+1}} \Ex_A\Pr[\mc{L}_{t+1}\lr{|A\rangle^{\otimes t+1}} = A] - \sup_{\mc{L}_{t}} \Ex_A\Pr[\mc{L}_{t}\lr{|A\rangle^{\otimes t}} = A]\, .
    \end{equation}
    To lower bound the above error, we need lower and upper bounds on optimal average success probability of learning $(n, 2n)$-bipartite graph states from $t+1$ copies and $t$ copies, respectively. 
    
    \paragraph{Upper Bounding the Optimal Average Success Probability of Learning From $t$ Copies:}
    We can write any $t$-copy learner as a POVM $\mc{L}_t = \{L^t_A\}_A$, which gives
    \begin{equation}
        \sup_{\mc{L}_{t}} \Ex_A\Pr[\mc{L}_{t}\lr{|A\rangle^{\otimes t}} = A] = \sup_{\{L^t_{A'}\}_{A'}} \frac{1}{|\Z_2^{n \times 2n}|} \sum_{A \in \Z_2^{n \times 2n}} \tr\lr{L^t_A |A\rangle\langle A|^{\otimes t}}\, ,
    \end{equation}
    where $|\Z_2^{n \times 2n}|$ denotes the total number of matrices in $\Z_2^{n \times 2n}$, or, equivalently, total number of $(n, 2n)$-bipartite graph states.
    Using \Cref{lemma:bigraph_states_t_copy_upper_bound}, we get that $\sup_{\{L^t_{A'}\}_{A'}} \frac{1}{|\Z_2^{n \times 2n}|} \sum_{A \in \Z_2^{n \times 2n}} \tr\lr{L^t_A |A\rangle\langle A|^{\otimes t}}\leq d/|\Z_2^{n \times 2n}|$ where $d$ is the dimension of the subspace $\mathrm{span}\{|A\rangle^{\otimes t}~|~ A \in \Z_2^{n \times 2n}\}$. Now, we upper bound this dimension. To begin, notice that 
    \begin{equation}
        \begin{split}
            |A\rangle^{\otimes t} &= \frac{1}{(\sqrt{2^{3n}})^t} \sum_{\substack{{x_1, \ldots, x_t \in \Z_2^n} \atop {y_1, \ldots, y_t \in \Z_2^{2n}}}} (-1)^{\sum_i x_i^T A y_i} \bigotimes_{i=1}^t |x_i y_i\rangle \\
            &= \frac{1}{(\sqrt{2^{3n}})^t}\sum_{\substack{{x_1, \ldots, x_t \in \Z_2^n} \atop {y_1, \ldots, y_t \in \Z_2^{2n}}}} (-1)^{\langle A, \sum_i x_i y_i^T \rangle} \bigotimes_{i=1}^t|x_i y_i\rangle\, ,
        \end{split}
    \end{equation}
    We can partition the sum over $\{x_i\}_i$ and $\{y_i\}_i$ into different choices of $S$ with $ \sum_{i}x_i y_i^T= S$ and define the states 
    \begin{equation}
        |S\rangle = \sum_{\substack{{x_1, \ldots, x_t \in \Z_2^n} \atop {y_1, \ldots, y_t \in \Z_2^{2n}}} \atop {\sum_i x_i y_i^T} = S}  \bigotimes_{i=1}^t|x_i y_i\rangle\, .
    \end{equation}
    For any $S = \sum_{i=1}^tx_i y_i^T$, we have
    \begin{equation}
        \mathrm{rank}\lr{S} = \mathrm{rank} \lr{\sum_{i=1}^tx_i y_i^T} \stackrel{(a)}{\leq} \sum_{i=1}^t\mathrm{rank}\lr{x_i y_i^T} \leq t\, ,
    \end{equation}
    where $(a)$ follows by sub-additivity of the rank.
    Thus, we can write
    \begin{equation}
        |A\rangle^{\otimes t} =\frac{1}{(\sqrt{2^{3n}})^t}\sum_{S \in M} (-1)^{\langle A, S\rangle} |S\rangle\, ,
    \end{equation}
    where $M$ is a subset of matrices of $\Z_2^{n \times 2n}$ of rank at most $t$.
        Notice that the vectors $|S\rangle$ are orthogonal as they are supported on disjoint subsets of computational basis states. From above, we get that for any $A$, $|A\rangle^{\otimes t} \in \mathrm{span}\{|S\rangle ~|~ S \in \Z_2^{n \times 2n}, \mathrm{rank}(S)\leq t\}$. Thus, the dimension of subspace spanned by $|A\rangle^{\otimes t}$ is at most the number of different possible $|S\rangle$ vectors. Mathematically,
        \begin{equation}
            \mathrm{dim}\lr{\mathrm{span}\{|A\rangle^{\otimes t}~|~ A \in \Z_2^{n \times 2n}\}} \leq |\{S~|~ S \in \Z_2^{n \times 2n}, \mathrm{rank}(S) \leq t\}|\, .
        \end{equation}
        For $t = n-1$, we have that 
        \begin{equation}
            |\{S~|~ S \in \Z_2^{n \times 2n}, \mathrm{rank}(S) \leq n-1\}| = |\{S~|~ S \in \Z_2^{n \times 2n}, \mathrm{rank}(S) \neq n\}|\, ,
        \end{equation}
        and the total number of matrices in $\Z_2^{n \times 2n}$ with full rank ($\mathrm{rank} = n$) is given as
        \begin{equation}
            |\{S~|~ S \in \Z_2^{n \times 2n}, \mathrm{rank}(S) = n\}| = \prod_{i=0}^{n-1}\lr{2^{2n} -2^i}\, ,
        \end{equation}
        which follows from \Cref{lemma:full-rank-binary-matrices}. Thus, 
        \begin{equation}
            |\{S~|~ S \in \Z_2^{n \times 2n}, \mathrm{rank}(S) \leq n-1\}| = 2^{2n^2} - \prod_{i=0}^{n-1}\lr{2^{2n} -2^i}\, ,
        \end{equation}
        since the total number of matrices in $\Z_2^{n\times 2n}$ is $2^{2n^2}$. We therefore have that 
        \begin{equation}
            \sup_{\{L^{n-1}_{A'}\}_{A'}} \frac{1}{|\Z_2^{n \times 2n}|} \sum_{A \in \Z_2^{n \times 2n}} \tr\lr{L^{n-1}_A |A\rangle\langle A|^{\otimes n-1}} \leq \frac{d}{|\Z_2^{n \times 2n}|} \leq \frac{2^{2n^2} - \prod_{i=0}^{n-1}\lr{2^{2n} -2^i}}{2^{2n^2}}\, = 1 - \prod_{i=0}^{n-1}\lr{1 -2^{i - 2n}}\, .
        \end{equation}
        Now, using \Cref{lemma:math_inequal}, $1 - \prod_i(1-a_i) \leq \sum_i a_i$ with $a_i = 2^{i-2n} \in [0, 1]$ for all $i \in [0, n-1]$, we get that
        \begin{equation}\label{eq:upper-bound-n-1-copy}
            \sup_{\mc{L}_{n-1}} \Ex_A\Pr[\mc{L}_{n-1}\lr{|A\rangle^{\otimes n-1}} = A] = \sup_{\{L^{n-1}_{A'}\}_{A'}} \frac{1}{|\Z_2^{n \times 2n}|} \sum_{A \in \Z_2^{n \times 2n}} \tr\lr{L^{n-1}_A |A\rangle\langle A|^{\otimes n-1}} \leq \sum_{i=0}^{n-1} 2^{i-2n} = \frac{2^n-1}{2^{2n}}\, .
        \end{equation}
        
        \paragraph{Lower Bounding the Optimal Average Success Probability of Learning From $t+1$-Copies:} We begin by applying the Hadamard transform to the second part of the bipartite graph state, i.e.,
        \begin{equation}
            (I \otimes H^{\otimes 2n}) |A\rangle = \frac{1}{\sqrt{2^{5n}}}\sum_{\substack{{x \in \Z_2^{n}}\\ {y, s \in \Z_2^{2n}}}}(-1)^{x^T A y} (-1)^{s \cdot y} |x\rangle |s\rangle = \frac{1}{\sqrt{2^{n}}} \sum_{x} |x \rangle |x^T A\rangle\,,
        \end{equation}
        where second equality follows by averaging over $y$. Now, we measure the resulting state in computational basis, which give samples $(x, x^T A)$ with $x$ drawn uniformly at random from $\Z_2^n$. Given $t+1$ copies, this way we obtain samples $\{(x_i, y_i)\}_{i=1}^{t+1}$, and we aim to solve the resulting system of linear equations for $A$. Notice that each sample provides $2n$ linear equations and, given $t+1 = n$ samples, we get $2n^2$ linear equations. Thus, we can successfully identify $A$ if we get $n$ linearly independent $x_i$. Equivalently, we want to solve the system of linear equations $Y = X^TA$, which has a unique solution if $X \in \Z_2^{n \times n}$ has full rank. Thus, the probability of learning the unknown bipartite graph state is lower bounded by the probability of obtaining $n$ linearly independent vectors in $n$ uniform draws from $\Z_2^{n}$. This probability can be calculated using \Cref{lemma:full-rank-binary-matrices} to be
        \begin{equation}
            \frac{\prod_{i=0}^{n-1}\lr{2^{n} -2^i}}{2^{n^2}} = \prod_{i=0}^{n-1}\lr{1 -2^{i-n}} = \prod_{i=1}^n (1-2^{-i})\, \geq 0.28\, .
        \end{equation}
        Thus, 
        \begin{equation}\label{eq:lower-bound-n-copy}
            \sup_{\mc{L}_{n}} \Ex_A\Pr[\mc{L}_{n}\lr{|A\rangle^{\otimes n}} = A] \geq 0.28,\, \hspace{3ex}\forall n. 
        \end{equation}
        Hence, using \Cref{eq:cloning-stab-error-triangle-inequal,eq:upper-bound-n-1-copy,eq:lower-bound-n-copy}, we get
        \begin{equation}
            \begin{split}
                \epsilon_{\mathrm{Cl}}^{\star}\lr{\mathrm{Stab}_{3n}, n-1, n}&\geq  \sup_{\mc{L}_{n}} \Ex_A\Pr[\mc{L}_{n}\lr{|A\rangle^{\otimes n}} = A] - \sup_{\mc{L}_{n-1}} \Ex_A\Pr[\mc{L}_{n-1}\lr{|A\rangle^{\otimes n-1}} = A]\, \\
                &\geq 0.28 - \frac{2^n-1}{2^{2n}}\\
                &= 0.28 - O\lr{\frac{1}{2^n}}\, .
            \end{split}
        \end{equation}
        To extend the above result to $t < n-1$, we note that any cloner from $t< n-1$ copies achieving a better error guarantee can also act as a cloner from $t=n-1$ copies and would thus contradict the above result. (We prove this in detail in \Cref{lemma:diff_copies} below.) Moreover, the proof can be extended to $(3n+1)$- and $(3n+2)$-qubit stabilizer states by embedding $(3n)$-qubit stabilizer states through padding with extra $|0\rangle$ tensor factors, as proved in \Cref{lemma:lifting-stab-states}. Thus, we conclude that
        \begin{equation}
            \epsilon_{\mathrm{Cl}}^{\star}\lr{\mathrm{Stab}_{n}, t, t+1}\geq 0.28 - O(2^{-n/3}), \hspace{3ex} \forall t < \lfloor n/3\rfloor\, ,
        \end{equation}
        as claimed.
    \end{proof}

\begin{proposition}\label{prop:cl-con} For any $t \geq 1$,
    \begin{equation}
        \sup_{\mc{L}_{t+1}} \Ex_A\Pr[\mc{L}_{t+1}\lr{\Lambda\lr{|A\rangle^{\otimes t}}} =A]\stackrel{}{\leq} \sup_{\mc{L}_{t}} \Ex_A\Pr[\mc{L}_{t}\lr{|A\rangle^{\otimes t}} = A] \stackrel{}{\leq} \sup_{\mc{L}_{t+1}} \Ex_A\Pr[\mc{L}_{t+1}\lr{|A\rangle^{\otimes t+1}} = A]\, .
    \end{equation}
\end{proposition}
\begin{proof}
    To prove the first inequality, notice that $\sup_{\mc{L}_{t+1}} \Ex_A\Pr[\mc{L}_{t+1}\lr{\Lambda\lr{|A\rangle^{\otimes t}}} =A]$ represents the average success probability optimized over a particular class of $t$-copy learners, which first use an approximate cloning channel to go from $t$ to  $t+1$ copies and then use an optimal $(t+1)$-copy learner. Thus, the supremum on the left hand side of first inequality is over a subset of the supremum over all $t$-copy learners.
    
    The second inequality can be argued by noting that any $t$-copy learner can be simulated by a $(t+1)$-copy learner which proceeds by discarding one copy and running the $t$-copy learner. Thus, the supremum on the right hand side on second inequality is over a larger set than that of all $t$-copy learners.
\end{proof}

\begin{lemma}\label{lemma:bigraph_states_t_copy_upper_bound} The optimal average success probability of learning bipartite graph states from $t$ copies is upper bounded by
    \begin{equation}
        \sup_{\{L_{A'}\}_{A'}} \frac{1}{|\Z_2^{n \times 2n}|} \sum_{A \in \Z_2^{n \times 2n}} \tr\lr{L^t_A |A\rangle\langle A|^{\otimes t}} \leq \frac{d}{|\Z_2^{n \times 2n}|}\, ,
    \end{equation}
    where $d$ is the dimension of the subspace spanned by $t$-copy $(n, 2n)$-bipartite graph states and $|\Z_2^{n \times 2n}|$ is the total number of such states.
\end{lemma}
\begin{proof}
    We define the projector $\Pi_{\mathrm{BiGraph_{3n}}}$ onto the subspace $\mathrm{span}\{|A\rangle^{\otimes t}~|~ A \in \Z_2^{n \times 2n}\}$. Then, we have that 
    \begin{equation}
        |A\rangle\langle A|^{\otimes t} \leq \Pi_{\mathrm{BiGraph_{3n}}}\, ,\forall A \in \Z_2^{n \times 2n}\, .
    \end{equation}
    Thus, we get that
    \begin{equation}
        \begin{split}
            \sup_{\{L_{A'}\}_{A'}} \frac{1}{|\Z_2^{n \times 2n}|} \sum_{A \in \Z_2^{n \times 2n}} \tr\lr{L^t_A |A\rangle\langle A|^{\otimes t}} &\leq \sup_{\{L_{A'}\}_{A'}} \frac{1}{|\Z_2^{n \times 2n}|} \sum_{A \in \Z_2^{n \times 2n}} \tr\lr{L^t_A \Pi_{\mathrm{BiGraph_{3n}}}}\\
            &= \sup_{\{L_{A'}\}_{A'}} \frac{1}{|\Z_2^{n \times 2n}|}\tr \lr{\sum_{A \in \Z_2^{n \times 2n}} L^t_A \Pi_{\mathrm{BiGraph_{3n}}}}\\
            &= \frac{1}{|\Z_2^{n \times 2n}|} \tr(\Pi_{\mathrm{BiGraph_{3n}}})
            = \frac{d}{|\Z_2^{n \times 2n}|}\, ,
        \end{split}
    \end{equation}
    where the second-to-last equality uses that POVM elements sum to the identity, and the last step used 
    \begin{equation}
        d = \tr(\Pi_{\mathrm{BiGraph_{3n}}} )
        =\mathrm{dim}\lr{\mathrm{span}\{|A\rangle^{\otimes t}~|~ A \in \Z_2^{n \times 2n}\}}\, .
    \end{equation}
    This concludes the proof.
\end{proof}

\begin{lemma}\label{lemma:full-rank-binary-matrices}
    The number of full rank matrices $S \in \Z_2^{n \times k}$ with $k \geq n$ is given by
    \begin{equation}
        \prod_{i=0}^{n-1}\lr{2^{k} -2^i}\, .
    \end{equation}
\end{lemma}
\begin{proof}
    Notice that since $k \geq n$, full rank matrices have rank $n$. Any full rank matrix in $\Z_2^{n \times k}$ has $n$ linearly independent vectors from $\Z_2^{k}$ in the rows of the matrix. Hence, the total number of such full-rank matrices is equal to the number of ways of choosing $n$ linearly independent vectors from $\Z_2^{k}$. 
    
    The first vector can be chosen from all $2^k-1$ non-zero vectors. Then, there are $2^k - 2$ ways to choose the second vector such that it is linearly independent the first. In general, for the $t+1$-th vector to be linearly independent from the first $t$, it can be any of the $2^k - 2^t$ vectors that lie outside the subspace spanned by first $t$ vectors. We thus see that the total number of ways to choose $n$ linearly independent vectors is
    \begin{equation}
        \prod_{i=0}^{n-1}\lr{2^{k} -2^i}\, .
    \end{equation}
    This completes the proof.
\end{proof}

\begin{lemma}\label{lemma:math_inequal} For any sequence $\{a_i\}_{i \in \mathbb{N}}$ with $0 \leq a_i \leq 1, \forall i$, we have that for any $t$
\begin{equation}
    1 - \prod_i^t (1-a_i) \leq \sum_i^t a_i\, .
\end{equation}
\end{lemma}
\begin{proof}
    We prove this by induction. Note that for $t=1$, the inequality holds even with equality. Assume that it is true for some $t\geq 1$, i.e., we have that
    \begin{equation}
        1 - \prod_i^{t}(1-a_i) \leq \sum_i^t a_i\, .
    \end{equation}
    Then, 
    \begin{equation}
        \begin{split}
            1 - \prod_i^{t+1}(1-a_i)= 1 -\prod_i^{t}(1-a_i) + a_{t+1}\prod_i^{t}(1-a_i) \leq \sum_i^t a_i + a_{t+1}\prod_i^{t}(1-a_i) \leq \sum_i^{t+1} a_{i}\, ,
        \end{split}
    \end{equation}
    where the first inequality follows by induction hypothesis and last statement follows because $a_{t+1} \geq 0$ and $\prod_i^{t}(1-a_i) \leq 1$.
\end{proof}

\begin{lemma}\label{lemma:diff_copies}
    If a family of states $\mc S$ does not admit a $(k-1, k, \epsilon)$-quantum cloning scheme for some $\epsilon$, then it also does not admit a $(t-1, t, \epsilon)$-quantum cloning scheme for any $t \leq k$.
\end{lemma}
\begin{proof}
    We prove this by contradiction. Assume that $\Lambda$ is a $(t-1, t, \epsilon)$ cloner for the family of states $\mc S$ with error $\epsilon$ for some $t \leq k$. This means that, for any $\rho \in \mc S$,
    \begin{equation}
        \td(\Lambda(\rho^{\otimes t-1}), \rho^{\otimes t}) \leq \epsilon.
    \end{equation}
    Define $\Lambda^{\prime}$ as $\Lambda^{\prime} = \Lambda \otimes I^{\otimes k-t}$.
    Then, 
    \begin{equation}
        \begin{split}
            \td\left(\Lambda^{\prime}(\rho^{\otimes k-1}), \rho^{\otimes k}\right) &= \td\left(\Lambda(\rho^{\otimes t-1}) \otimes \rho^{\otimes k-t}, \rho^{\otimes t} \otimes \rho^{\otimes k-t}\right)\\
            &\stackrel{}{=} \td(\Lambda(\rho^{\otimes t-1}), \rho^{\otimes t}) \leq \epsilon \, ,
        \end{split}
    \end{equation}
    Thus, $\Lambda^{\prime}$ is a $(k-1, k, \epsilon)$-cloner for the family of states $\mc S$, contradicting our assumption that no such cloner exists.
\end{proof}

\begin{lemma}\label{lemma:lifting-stab-states}
    Let $\mathrm{Stab}_n$ be the class of $n$-qubit stabilizer states.
    Suppose that $\mathrm{Stab}_n$ does not admit a $(t-1, t, \epsilon)$ quantum cloning scheme, then the class $\mathrm{Stab}_{n+1}$ of $(n+1)$-qubit stabilizer states does not admit a $(t-1, t, \epsilon)$ quantum cloning scheme.
\end{lemma}
\begin{proof}
    The proof is provided in \Cref{a:proofs} and follows by noting that for an $n$-qubit stabilizer state $|\psi\rangle$, $|\psi \rangle \otimes |0\rangle$ is an $(n+1)$-qubit stabilizer state and thus any cloning scheme for $(n+1)$-qubit stabilizer states would yield a cloning scheme for $n$-qubit stabilizer states.
\end{proof}
{
Note that in the above, we prove a lower bound on worst case cloning error for stabilizer states by lower bounding it with an average case cloning error for the subclass of bipartite graph states. Next, we provide an average-case cloning lower bound for stabilizer states directly. We relate the average case cloning error with the worst case cloning error by using ideas similar to the one used in \cite{Werner_1998, Zhou_2025}. We use the notation
\begin{equation}
    \epsilon_{\mathrm{Cl}, \mathrm{avg}}^{\star}(\mathrm{Stab}_n, t, t+m) := \min_{\Lambda} \Ex_{\rho \in \mathrm{Stab}_n}  \td \lr{\Lambda(\rho^{\otimes t}), \rho^{\otimes t+m}}
\end{equation}
to denote the optimal average case cloning error in $t \rightarrow t+m$ cloning task for $n$-qubit stabilizer states.
\begin{theorem}\label{thm:avg-case-cl-lower-bound}
    For any integers $n, t, m \geq 1$, we have 
    \begin{equation}
        \epsilon_{\mathrm{Cl}}^{\star}(\mathrm{Stab}_n, t, t+m) = \epsilon_{\mathrm{Cl}, \mathrm{avg}}^{\star}(\mathrm{Stab}_n, t, t+m)\, .
    \end{equation}
\end{theorem}
\begin{proof}
    We use the notations $\epsilon_{\mathrm{Cl}}(\Lambda,\mathrm{Stab}_n, t, t+m)$ and $\epsilon_{\mathrm{Cl}, \mathrm{avg}}(\Lambda,\mathrm{Stab}_n, t, t+m)$ to denote the worst-case and average-case cloning errors, respectively, incurred by the cloning map $\Lambda$ in the $t\rightarrow t+m$ cloning task for $n$-qubit stabilizer states. Now, we argue that for any cloner $\Lambda$, there exists a cloner $\Lambda^{\prime}$ such that
    \begin{equation}
        \epsilon_{\mathrm{Cl}}(\Lambda^{\prime},\mathrm{Stab}_n, t, t+m) \leq \epsilon_{\mathrm{Cl}, \mathrm{avg}}(\Lambda,\mathrm{Stab}_n, t, t+m)\, .
    \end{equation}
    To see this, for any $t \rightarrow t+m$ cloner $\Lambda$, we construct the cloner $\Lambda^{\prime}$ as follows:
    \begin{equation}
        \Lambda^{\prime}(X) = \Ex_{C}\lrq{C^{\dagger, \otimes t+m} \Lambda(C^{\otimes t} X C^{\dagger, \otimes t}) C^{\otimes t+m}}\, ,
    \end{equation}
    where the expectation is over uniformly random Clifford unitaries. In other words, the cloner $\Lambda^{\prime}$ takes a state on the $t$-copy Hilbert space, applies the $t$-fold tensor power of a random Clifford unitary $C$, uses the cloner $\Lambda$ to obtain a state on the $(t+m)$-copy Hilbert space, and then applies the inverse of the previously applied random Clifford unitary, $C^{\otimes t+m}$. Now, for any stabilizer state $\rho$, we can calculate the trace distance

    \begin{equation}
        \begin{split}
            \td\lr{\Lambda^{\prime}(\rho^{\otimes t}), \rho^{\otimes t+m}} &= \frac{1}{2}\lraa{\Lambda^{\prime}(\rho^{\otimes t}) - \rho^{\otimes t+m}}\\
            &= \frac{1}{2}\lraa{\Ex_{C}\lrq{C^{\dagger, \otimes t+m} \Lambda(C^{\otimes t} \rho^{\otimes t} C^{\dagger, \otimes t}) C^{\otimes t+m}} - \rho^{\otimes t+m}}\\
            &\stackrel{}{=} \frac{1}{2}\lraa{\Ex_{C}\lrq{C^{\dagger, \otimes t+m} \lr{\Lambda(C^{\otimes t}\rho^{\otimes t} C^{\dagger, \otimes t})- C^{\otimes t+m}\rho^{\otimes t+m}C^{\dagger, \otimes t+m}}  C^{\otimes t+m}}}\\
            &\stackrel{(a)}{\leq} \Ex_{C} \td\lr{\Lambda(C^{\otimes t}\rho^{\otimes t} C^{\dagger, \otimes t}), C^{\otimes t+m}\rho^{\otimes t+m}C^{\dagger, \otimes t+m}}\\
            &\stackrel{(b)}{=}\epsilon_{\mathrm{Cl, avg}}(\Lambda,\mathrm{Stab}_n, t, t+m)\, ,
        \end{split} 
    \end{equation}
    where $(a)$ follows by using triangle inequality and invariance of trace distance under unitary transformations, and (b) follows by noting that for $\rho \in \mathrm{Stab}_n$, and a uniformly random $n$-qubit Clifford unitary $C$, $C\rho C^{\dagger}$ is uniformly distributed on the set of $n$-qubit stabilizer states. As the above equation is true for any $\rho \in \mathrm{Stab}_n$, we may take the supremum over $n$-qubit stabilizer states to obtain
    \begin{equation}
        \epsilon_{\mathrm{Cl}}(\Lambda^{\prime},\mathrm{Stab}_n, t, t+m) = \sup_{\rho \in \mathrm{Stab_n}} \td\lr{\Lambda^{\prime}(\rho^{\otimes t}), \rho^{\otimes t+m}} \leq \epsilon_{\mathrm{Cl, avg}}(\Lambda,\mathrm{Stab}_n, t, t+m)\, .
    \end{equation}
    And, trivially, for any $\Lambda^{\prime}$, the worst-case cloning error can be lower bounded by the average-case error:
    \begin{equation}
        \epsilon_{\mathrm{Cl}} (\Lambda^{\prime},\mathrm{Stab}_n, t, t+m) \geq \epsilon_{\mathrm{Cl, avg}}(\Lambda^{\prime},\mathrm{Stab}_n, t, t+m)\, .
    \end{equation}
    Thus, we have that 
    \begin{equation}
        \min_{\Lambda} \epsilon_{\mathrm{Cl, avg}}(\Lambda,\mathrm{Stab}_n, t, t+m) \leq \min_{\Lambda} \epsilon_{\mathrm{Cl}}(\Lambda,\mathrm{Stab}_n, t, t+m) \leq \min_{\Lambda} \epsilon_{\mathrm{Cl, avg}}(\Lambda,\mathrm{Stab}_n, t, t+m)\, .
    \end{equation}
    Hence, 
    \begin{equation}
        \epsilon_{\mathrm{Cl}}^{\star}(\mathrm{Stab}_n, t, t+m) = \epsilon_{\mathrm{Cl}, \mathrm{avg}}^{\star}(\mathrm{Stab}_n, t, t+m)\, .
    \end{equation}
\end{proof}
Combining \Cref{thm:avg-case-cl-lower-bound,thm:stab-states-clone} yields:
\begin{corollary}
    Let $\mathrm{Stab}_n$ be the set of pure $n$-qubit stabilizer states. Then,
    \begin{equation}
        \epsilon^{\star}_{\mathrm{Cl, avg}}\lr{\mathrm{Stab}_n, t, t+1} \geq 0.28 - O(2^{-n/3}), \hspace{3ex} \forall t < \lfloor n/3\rfloor\, .
    \end{equation}
\end{corollary}}

\subsection{Cloning Lower Bound for Mixed Abelian StateHSP}

We consider cloning for classes of states that have symmetries. Let $G$ be an Abelian group, and let $\Psi_{H}^{\epsilon}$ be the set of states that $\epsilon$-hides the subgroup $H \leq G$. Then we consider the cloning of $\mc{S} _{\alpha} = \bigcup_{H, |H|= \alpha} \Psi_{H}^{\epsilon}$, i.e., states that are promised to hide an unknown hidden subgroup of fixed order. Here, we do not make any assumptions about the purity of the states. Thus, $\mc{S}_{\alpha}$ also contains mixed states, which is relevant for our lower bound proof below. Our cloning lower bound technique here is again based on the probability of uniquely identifying an unknown $H$ from state copies. Along the way, we identify the optimal measurement for the Abelian state hidden subgroup problem in the next section.

\subsubsection{Optimality of Character POVMs for a Class of Abelian StateHSPs}

We show that in the worst case sense, character POVMs are optimal up to classical post-processing for solving a class of Abelian StateHSP. We show this optimality for Abelian StateHSPs with non-isomorphic irreducible representations. To show this, we use the following result~\cite{wright2016learn}.

\begin{proposition}[\cite{wright2016learn}, Proposition 1.2.7]\label{thm:block_diagonal_optimal_measurement}
    Suppose $\rho$ is block diagonal with blocks corresponding to the (known) orthogonal projectors $\{\Pi_i\}_i$. Then,
    \begin{itemize}
        \item One can without loss of generality pre-process any quantum algorithm that acts on $\rho$ by measuring $\{\Pi_i\}_i$ on $\rho$.
        \item If $\rho$ is a multiple of the identity in each block, then $\{\Pi_i\}_i$ is an optimal measurement for extracting information from $\rho$.
    \end{itemize}
\end{proposition}

The proof of this result is simple. It is easy to see that, assuming a block diagonal structure, we can always perform the projective measurement onto the blocks without losing any information. 
Moreover, if the state is proportional to the maximally mixed state within each block then we cannot get any more information after identifying the block. 

Let us now sketch how we will use \Cref{thm:block_diagonal_optimal_measurement} before presenting the detailed proof. 
For an Abelian StateHSP $(G, \mu, \Psi_H^{\epsilon})$ and for any $\rho\in \Psi_H^{\epsilon}$, we consider the state
\begin{equation}
    \sigma = \frac{1}{|G|}\sum_{g \in G} \mu(g) \rho \mu(g)^{\dagger}\, .
\end{equation} 
We will argue that this state is block-diagonal with respect to the character POVM $\{\Pi_\lambda\}_\lambda$. Then, using the non-isomorphic property of irreducible representations, we will argue that each block is one-dimensional and hence the above state is also maximally mixed in each block. We can combine this with Proposition~\ref{thm:block_diagonal_optimal_measurement} to obtain the optimality of character POVMs on $\sigma$. 
However, this optimality result is limited, it only applies to algorithms that process single copies of $\sigma$ at a time, but it does not account for general multi-copy measurements. To show the optimality among this more general class of algorithms, we consider $\sigma^{\otimes t}$, which is block diagonal with respect to tensor powers of single-copy character POVMs. As each block is still one-dimensional, this will lead to the optimality of single-copy character POVMs even when multi-copy measurements are allowed. Intuitively, the optimality can be understood by noting that allowing for multiple i.i.d.~copies of $\sigma$ is equivalent to allowing non-i.i.d.~states from $\Psi_H^\epsilon$, and in this scenario it is easy to see that single-copy character POVMs are optimal.

The rest of this subsection makes the above reasoning precise to establish the following result:

\begin{theorem}[Optimality of Character POVMs for Abelian StateHSP]\label{thm:optimality-statehsp}
    Let $(G, \mu, \Psi_H^{\epsilon})$ be an Abelian StateHSP problem such that the unitary representation $\mu$ has non-isomorphic irreducible representations. Then, the character POVM $\{\Pi_{\lambda}\}_{\lambda}$ is the optimal procedure to solve the Abelian StateHSP in the worst-case sense.
\end{theorem}
\begin{proof}
     Let $t$ be the number of copies provided.
     We assume that we are allowed any arbitrary $t$-copy measurements on instances of $\Psi_H^{\epsilon}$, i.e., we are allowed to make measurements on $\rho^{\otimes t}$ with $\rho \in \Psi_H^{\epsilon}$. Now, for any $\rho \in \Psi_H^{\epsilon}$, define
    \begin{equation}\label{eqn:worst-case-state}
        \sigma = \frac{1}{|G|}\sum_{g \in G} \mu(g) \rho \mu(g)^{\dagger}.
    \end{equation}

    \begin{claim}\label{claim:t-twirl-state-closure}
        Let $\rho \in \Psi_H^{\epsilon}$, for an underlying Abelian Group $G$ with the unitary representation $\mu$ and let
        \begin{equation}
            \sigma = \frac{1}{|G|}\sum_{g \in G} \mu(g) \rho \mu(g)^{\dagger},
        \end{equation}
        then, 
        \begin{enumerate}
            \item $\sigma \in \Psi_H^{\epsilon}$,
            \item $[\sigma, \mu(g)] = 0, \forall g \in G$.
        \end{enumerate}
    \end{claim}
    \begin{proof}
        See Appendix \ref{a:proofs}.
    \end{proof}
    
    As a consequence of bullet point 1 in Claim~\ref{claim:t-twirl-state-closure}, $\sigma$ is also a valid instance of the Abelian StateHSP. So, for the purposes of a worst-case analysis, we can assume that we are given $\sigma^{\otimes t}$. 
    As a consequence of bullet point 2 of Claim~\ref{claim:t-twirl-state-closure}, $\sigma$ has the same eigenvectors as $\mu(g)$, i.e., $\sigma$ is diagonal in the basis $\{|\lambda_i,  v_{\lambda_i}\rangle = |\lambda_i \rangle\}_i$, where we can drop the $v_{\lambda_i}$ because the irreducible representations are by assumption non-isomorphic.
    Now, we define the $t$-copy POVM $\{\bigotimes_{i}^t \Pi_{\lambda_i}\}$, with character POVM elements $\Pi_{\lambda_i} = |\lambda_i \rangle\langle \lambda_i |$, which because of the assumption on non-isomorphic irreducible representations are rank-one projections. 
    Note that:
    \begin{enumerate}
        \item $\sigma^{\otimes t}$ is diagonal with respect to the basis projected on by the POVM elements $\{\bigotimes_{i= 1}^t \Pi_{\lambda_i}\}$.
        \item As $\sigma^{\otimes t}$ is block-diagonal with blocks of size $1$, it is proportional to the identity in each block. %
    \end{enumerate}
    Hence, by \Cref{thm:block_diagonal_optimal_measurement}, the POVM $\{\bigotimes_{i = 1}^t \Pi_{\lambda_i}\}$ is optimal to solve the Abelian StateHSP in the worst case sense.
\end{proof}

In the case of the phaseless stabilizer StateHSP, we have non-isomorphic irreducible representations. This yields the optimality of Bell sampling for phaseless stabilizer StateHSP.

\subsubsection{Cloning Lower Bounds for Mixed Abelian StateHSPs}
We are now well-equipped to prove cloning lower bounds for particular classes of mixed Abelian StateHSP. 
Our approach is similar to the one that led to the lower bound for structured sample amplification (Theorem~\ref{thm:general-lower-bound}). Namely, we will construct a distinguisher that distinguishes between true and cloned copies with noticeable advantage exactly when increasing the number of state copies leads to a noticeable increase in the success probability of correctly identifying the hidden subgroup.

We consider the Abelian StateHSP $(G, \mu, \Psi_H^{\epsilon})$ such that the underlying Hilbert space can be written as
\begin{equation}
    \mathcal{H} = \bigoplus_{\lambda \in \hat{G}} \mathbb{C}\lr{|\lambda\rangle}\, ,
\end{equation}
where $\hat{G}$ is the dual group to $G$, and $\mathbb{C}\lr{|\lambda\rangle}$ is the span of vector $|\lambda\rangle$ under the complex field with
\begin{equation}
    \mu(g) |\lambda\rangle  = \chi_{\lambda}(g)|\lambda \rangle\, \hspace{3ex} \forall g \in G, \lambda \in \hat{G}\, ,
\end{equation}
where $\chi_{\lambda}$ is the character corresponding to the irrep labelled by $\lambda$. Equivalently, the unitary representation $\mu(g)$ is given as
\begin{equation}\label{eq:new_rep_assumption}
    \mu(g) = \bigoplus_{\lambda \in \hat{G}} \chi_{\lambda}(g)\, .
\end{equation}

Now, we consider the set of states 
\begin{equation}
    \mc{S} _{\alpha} = \bigcup_{H\leq G: |H|= \alpha} \Psi_{H}^{\epsilon}\, ,
\end{equation}
for some fixed $\alpha$.
That is, we consider the (mixed and pure) state instances of an Abelian StateHSP with Abelian group $G$, representation $\mu$, and some hidden subgroup $H\leq G$ of order $\alpha$.

For a given state $\rho$ with hidden subgroup $H \leq G$, we refer to $n_{H^{\perp}}$ %
as the number of independent generators of $H$, and we let $p_{\rho}^n$ be the probability of uniquely specifying $H$ given $n$ copies of the state $\rho$.
Moreover, we take the states
\begin{equation}\label{eq:sigma-general-statehsp}
    \sigma_{H} = \frac{1}{|H^{\perp}|}\sum_{\lambda \in H^{\perp}} |\lambda\rangle \langle \lambda|\, .
\end{equation}

\begin{lemma}\label{lemma:sigma-general-statehsp}
    For $\sigma_H$ states defined above,
    \begin{enumerate}
        \item if $g \in H$, $\tr\lr{\mu(g) \sigma_H} = 1$, and
        \item if $g \in G\setminus H$, $\tr\lr{\mu(g) \sigma_H} = 0$.
    \end{enumerate}
\end{lemma}
\begin{proof}
    To show this, we use the property of characters that,
\begin{equation}
    \sum_{\lambda \in H^{\perp}} \chi_{\lambda}(g) = \begin{cases}
        |H^{\perp}| & \text{ if } g \in H, \\
        0 & \text{ otherwise}.
    \end{cases}
\end{equation}
Now,
\begin{equation}
    \begin{split}
        \tr\lr{\mu(g) \sigma_H} &= \tr\lr{\frac{1}{|H^{\perp}|} \sum_{\lambda \in H^{\perp}} \mu(g) |\lambda \rangle \langle \lambda|} \\
        &= \frac{1}{|H^{\perp}|} \sum_{\lambda \in H^{\perp}} \chi_{\lambda}(g) 
    \end{split}
\end{equation}
Thus, using the above equations,
\begin{equation}
    \tr\lr{\mu(g) \sigma_H} = \begin{cases}
        1 & \text{if } g \in H,\\
        0 & \text{otherwise}\, .
    \end{cases}
\end{equation}
\end{proof}

Thus, $\sigma_H$ hides the subgroup $H$. Now, we define
\begin{equation}
    \mc{S}_{\alpha}^{\sigma} =  \bigcup_{H \leq G, |H| = \alpha} \{\sigma_H\}.
\end{equation}
It is easy to see that $\mc{S}_{\alpha}^{\sigma} \subseteq \mc{S}_{\alpha}$. For this class, we first prove an optimal learning strategy, after which we establish a quantum cloning analogue of \Cref{thm:general-lower-bound}.
{
\begin{lemma}\label{lemma:avg-case-optimality-character-diagonal-states}
    Let $G$ be an Abelian group with unitary representation $\mu$ as in \Cref{eq:new_rep_assumption}. Then, the character POVM $\{\Pi_{\lambda}\}_{\lambda}$ is the optimal procedure to learn the states in the set $\mc{S}_{\alpha}^{\sigma}$ in an average-case sense.
\end{lemma}
\begin{proof}
    By definition, all states in $\mc{S}_{\alpha}^{\sigma}$ are diagonal in the character basis. Moreover the underlying unitary representation has non-isomorphic irreps. Thus, using \Cref{thm:block_diagonal_optimal_measurement} and arguments similar to \Cref{thm:optimality-statehsp}, we can straightforwardly prove the optimality of character POVM up to classical post-processing for any state $\sigma_H \in \mc{S}_{\alpha}^{\sigma}$. Since, character POVM is optimal for any state in the set, the statement holds in an average case-sense.
\end{proof}
}
\begin{theorem}[Formal Statement of Theorem~\ref{thm:quantum-results-informal}, Point 1: Error Lower Bounds for Cloning of Abelian StateHSP]\label{thm:structured-cloning} Let $d \in \mathbb{N}$.
       For any Abelian group $G$ with representation $\mu$ as in \Cref{eq:new_rep_assumption}.
       Then, for any $t \leq d$, the optimal error in cloning from $t-1$ copies to $t$ copies for the class of states $\mc{S}_{\alpha}$ is lower bounded as
       \begin{equation}
           \epsilon^{\star}_{\mathrm{Cl}}(\mc{S}_{\alpha}, t-1, t) 
           \geq \epsilon^{\star}_{\mathrm{Cl}}(\mc{S}_{\alpha}^{\sigma}, t-1, t) 
           \geq 
           \frac{1}{2}\Ex_{\sigma_H \sim \mc{S}_{\alpha}^{\sigma}}(p^{d}_{\sigma_H} - p^{d-1}_{\sigma_H})\, , \hspace{3ex} \forall t \leq d\, ,
       \end{equation}
       where the expectation is uniform over the states $\sigma_H$ defined in~\eqref{eq:sigma-general-statehsp} for $H\leq G$ with $|H|=\alpha$. %
\end{theorem}
\begin{proof}
    Let $\Lambda$ be any quantum cloning map for the class of states $\mc{S}_{\alpha}$. We first consider cloning from $t - 1 = d-1$ copies, which we later extend to any $t- 1\leq d-1$ copies. Notice that $\mc{S}_{\alpha}^{\sigma} \subseteq \mc{S}_{\alpha}$, thus
    \begin{equation}
        \epsilon^{\star}_{\mathrm{Cl}}(\mc{S}_{\alpha}, t-1, t) {\geq \epsilon^{\star}_{\mathrm{Cl}}(\mc{S}_{\alpha}^{\sigma}, t-1, t)}.
    \end{equation}
    Then, given a $\sigma_H\in \mc{S}_{\alpha}^{\sigma}$, we compare
    \begin{enumerate}
        \item true copies, $\sigma_H^{\otimes d}$, and
        \item cloned copies, $\Lambda(\sigma_H^{\otimes (d-1)})$.
    \end{enumerate}
    We define two probability distribution over $(\hat{G})^{\times d}$ as 
    \begin{equation}
        q_{\Lambda(\sigma_H^{\otimes (d - 1)})}(\lambda_1, \ldots , \lambda_{d}) = \Tr(\Lambda(\sigma_H^{\otimes (d-1) }) \bigotimes_{i=1}^{d}\Pi_{\lambda_i})\, ,
    \end{equation}
    and 
    \begin{equation}
            q_{\sigma_H^{\otimes d}}(\lambda_1, \ldots , \lambda_{ d}) = \Tr(\sigma_H^{\otimes d } \bigotimes_{i=1}^{d} \Pi_{\lambda_i})\, ,
    \end{equation}
    where $\{\Pi_{\lambda}\}_{\lambda}$ is the character POVM. We now again use a distinguisher based-approach, noting that the optimal cloning error can be rewritten as
    \begin{equation}\label{eqn:distinguisher-picture-cloning}
        \begin{split}
        \epsilon^{\star}_{\mathrm{Cl}}(\mc{S}_{\alpha}^{\sigma}, d-1, d) 
        &= \min_{\Lambda} \sup_{\rho \in \mc{S}_{\alpha}^{\sigma}} \max_{\mc A} \mathrm{Adv}(\mc A, H, \Lambda, d) \, \\
        &\geq \min_{\Lambda} \Ex_{\rho \sim \mc{S}_{\alpha}^{\sigma}} \mathrm{Adv}(\mc A, H, \Lambda, d)\, ,
        \end{split}
    \end{equation}
    where the second step follows by fixing a distinguisher, and lower bounding the supremum by average over the class of states. We consider the following distinguisher:
    \begin{algorithm}
            \caption{Distinguisher $\mc A$}\label{alg:general-quantum-distinguisher}
            \algorithmicrequire{ $tn$-qubit state $\sigma$ for some $t \in \mathbb{N}$, classical description of $H$} \\
            \algorithmicensure{ Accept or Reject} \\
            1. Measure the character POVM $\{\Pi_{\lambda}\}_{\lambda}$ on qubits $i=kn+1,\ldots,(k+1)n$ of $\sigma$ for every $k=0,\ldots,t-1$.\\
            2. Run a consistent learning algorithm that outputs one of the consistent subgroups $\hat{H}$ uniformly at random using the measurement samples from the first step.\\
            3. Accept if $\hat{H} = H$, else Reject.
    \end{algorithm}
    
    For the distinguisher defined above, the expected distinguishing advantage is
    \begin{equation}
        \begin{split}
        \Ex_{\sigma_H \sim \mc{S}_{\alpha}^{\sigma}}\mathrm{Adv}(\mc A) &= \Ex_{\sigma_H \sim \mc{S}_{\alpha}^{\sigma}}\lra{\Pr_{q_{\sigma_H^{\otimes d}}}\lrq{\mc{A}\lr{\lambda_1, \ldots \lambda_{ d}} = 1} - \Pr_{q_{\Lambda(\sigma_H^{\otimes (d - 1)})}}\lrq{\mc{A}\lr{\lambda_1, \ldots \lambda_{ d}} = 1}}\, \\
        &\geq \lra{\Ex_{\sigma_H \sim \mc{S}_{\alpha}^{\sigma}}\Pr_{q_{\sigma_H^{\otimes d}}}\lrq{\mc{A}\lr{\lambda_1, \ldots \lambda_{ d}} = 1} - \Ex_{\sigma_H \sim \mc{S}_{\alpha}^{\sigma}}\Pr_{q_{\Lambda(\sigma_H^{\otimes (d - 1)})}}\lrq{\mc{A}\lr{\lambda_1, \ldots \lambda_{ d}} = 1}}\, ,
        \end{split}
    \end{equation}
    where the $\lambda_i$ are sampled from the respective distributions.
    Here, we slightly abuse notation when using $\mc{A}\lr{\lambda_1, \ldots \lambda_{ d}}$ also for the decision made by distinguisher $\mathcal{A}$ after obtaining measurement outcomes $\lambda_1, \ldots \lambda_{ d}$ in step 1 of \Cref{alg:general-quantum-distinguisher}.
    Since the states in $\mc{S}_{\alpha}^{\sigma}$ are diagonal in the character basis, the single copy character POVMs are optimal for learning $H$ using \Cref{lemma:avg-case-optimality-character-diagonal-states}. Moreover, since there is a uniform prior over $H$, the probability of the distinguisher accepting on cloned $d$ copies can be upper-bounded by the probability of the distinguisher accepting on $d - 1$ true copies: 
    \begin{align}
        \Ex_{\sigma_H \sim \mc{S}_{\alpha}^{\sigma}}\Pr\left[\mathcal{A}\left(\Lambda (\sigma_H^{\otimes (d - 1)})\right) = 1\right]
        &= \Ex_{\sigma_H \sim \mc{S}_{\alpha}^{\sigma}} \Pr_{q_{\Lambda(\sigma_H^{\otimes d - 1})}}\lrq{\mc{A}\lr{\lambda_1, \ldots \lambda_{ d}} = 1} \\
        &\leq \Ex_{\sigma_H \sim \mc{S}_{\alpha}^{\sigma}}\Pr_{q_{\sigma_H^{\otimes (d-1)}}}\lrq{\mc{A}\lr{\lambda_1, \ldots \lambda_{ d- 1}} = 1}\\
        &=  \Ex_{\sigma_H \sim \mc{S}_{\alpha}^{\sigma}}\Pr\left[\mathcal{A}\left(\sigma_H^{\otimes (d-1) }\right)= 1\right]\, ,
    \end{align}
    where we again slightly abuse notation by using the same $\mathcal{A}$ to denote both the distinguisher and its action on the observed measurement samples, as well as on different numbers of copies/samples.
    This is a quantum analogue of Proposition~\ref{prop:sa-con}.
    Moreover, the probability of learning the subgroup from $d$ copies is at least the probability of doing so from $d - 1$ copies. Thus, the advantage of the distinguisher in distinguishing the true and cloned copies is lower bounded as
    \begin{equation}
        \Ex_{\sigma_H \sim \mc{S}_{\alpha}^{\sigma}}\mathrm{Adv}(\mc A) \geq \Ex_{\sigma_H \sim \mc{S}_{\alpha}^{\sigma}}\Pr_{q_{\sigma_H^{\otimes d}}}\lrq{\mc{A}\lr{\lambda_1, \ldots \lambda_{ d}} = 1} - \Ex_{\sigma_H \sim \mc{S}_{\alpha}^{\sigma}}\Pr_{q_{\sigma_H^{\otimes (d-1)}}}\lrq{\mc{A}\lr{\lambda_1, \ldots \lambda_{ d-1}} = 1}.
    \end{equation}
    Notice that the above difference on the right hand side can alternatively be thought of in terms of a difference between probabilities of $d$ vs $d-1$ samples specifying a subgroup $H$. Thus, we can use similar calculation as used in the proof of Theorem~\ref{thm:general-lower-bound} to lower bound this by $\frac{1}{2}\Ex_{\sigma_H \sim \mc{S}_{\alpha}^{\sigma}}(p^{d}_{\sigma_H} - p^{d-1}_{\sigma_H})$, yielding 
    \begin{equation}
        \epsilon^{\star}_{\mathrm{Cl}}(\mc{S}_{\alpha}^{\sigma}, d-1, d) \geq \frac{1}{2}\Ex_{\sigma_H \sim \mc{S}_{\alpha}^{\sigma}}(p^{d}_{\sigma_H} - p^{d-1}_{\sigma_H}) \, ,
    \end{equation}
    where $p^{d}_{\sigma_H}$ and $p^{d - 1}_{\sigma_H}$ are the probabilities of the measurement outcomes that uniquely specify the hidden subgroup $H$ from $d$ and $d - 1$ copies of the state, respectively. This gives
    \begin{equation}
         \epsilon^{\star}_{\mathrm{Cl}}(\mc{S}_{\alpha}, d-1, d) \geq \epsilon^{\star}_{\mathrm{Cl}}(\mc{S}_{\alpha}^{\sigma}, d-1, d) \geq \frac{1}{2}\Ex_{\sigma_H \sim \mc{S}_{\alpha}^{\sigma}}(p^{d}_{\sigma_H} - p^{d-1}_{\sigma_H}).
    \end{equation}
    We can then extend this to a lower bound for cloning from $t-1 \leq d-1$ to $t$ copies via Lemma~\ref{lemma:diff_copies}.
\end{proof}

\section*{Acknowledgements}

We thank Yanlin Chen, Daniel Grier, and Natalie McHugh for insightful discussions. Furthermore, we are grateful to Kenny Chen for pointing out small mistakes in the sample amplification part of an earlier version of the paper. We acknowledge that ChatGPT-5.6 Sol was useful in finding the mistake in the claimed structured random purification results in the first version of this paper, in identifying the bipartite graph state class now used in the stabilizer state cloning lower bound proof as a promising candidate, and in developing the proof strategy for the no-go results on structured random purifications given in \Cref{appendix:purification-no-go} as well as for formalising the proof of \Cref{prop:sa-con} and \Cref{thm:avg-case-cl-lower-bound}.
We refined and verified these ideas provided by ChatGPT and (in our opinion) improved the presentation.
The final paper reflects our own understanding of the results, and we take full responsibility for every claim, proof, and citation contained in the paper.

\sloppy
\printbibliography
\newpage
\appendix

\crefalias{section}{appendix}

\section{Deferred Proofs}\label{a:proofs}

\begin{proposition}\label{prop:sa-con}
     The probability of acceptance by the distinguisher defined in Algorithm~\ref{alg:distinguisher-general-lower-bound} for $Z_{f, t+1}^{SA} = T_{SA}(Z_{f, t})$ is upper bounded as
    \begin{equation}
       \Ex_f \Pr\lrq{\mc{A}^f(Z_{f, t+1}^{SA}) = 1} \leq \Ex_f\Pr\lrq{\mc{A}^f(Z_{f, t})= 1}\, 
    \end{equation}
    where the expectation is over uniformly random functions from $\mc{F}$.
\end{proposition}

To prove \Cref{prop:sa-con}, we require some additional lemmas.

\begin{lemma}\label{lemma:optimal_strategy}
    Given a sample labeled according to a function drawn at random according to a prior $\pi$ over some class $\mc{F}$, the learner that outputs a random hypothesis with maximum weight under the posterior maximizes the success probability of identifying the unknown function.
\end{lemma}
\begin{proof}
    Denote the random variable representing the function by $F$, and denote the drawn sample by $S=(X,F(X))=(X,Y)$. Then, the success probability of a learner $L$ correctly identifying $F$ from $S$ is
    \begin{equation}
        \Pr[\mathrm{success}] = \Ex\Pr[L(S) = F~|~ S=s, F=f] =  \Ex \Pr[L(S) = F~|~ S=s] .
    \end{equation}
    Now, we can write
    \begin{equation}
        \begin{split}
            \Pr[L(S) = F |S=s] 
            &= \sum_{f} \Pr\lrq{L(S) = F \text{ and } F = f | S= s} \\
            &= \sum_{f} \Pr\lrq{L(S) = f | S= s} \Pr\lrq{F = f | S= s} \\
            &\leq \max_f\Pr\lrq{F = f | S= s} \sum_{f'} \Pr\lrq{L(S) = f' | S= s}\\
            &= \max_f\Pr\lrq{F = f | S= s}\, ,
        \end{split}
    \end{equation}
    where the second equality holds because the events $L(S) = f$ and $F = f$ are independent given some data set $s$. 

    Thus, for a given sample set $s$, the probability of correctly identifying the unknown function is upper bounded by $\max_f\Pr\lrq{F = f | S= s}$, which is exactly the maximum probability weight of a hypothesis under the posterior. Consequently, 
    \begin{equation}
        \Pr[\mathrm{success}] \leq \Ex_S \max_f\Pr\lrq{F = f | S= s}\, .
    \end{equation}
    
    Now, consider the learner $L_{\mathrm{MAP}}$ defined via
    \begin{equation}
        \Pr[L_{\mathrm{MAP}}(S) = f| S=s] = \begin{cases}
            1/a & \text{ if } f \in \argmax_{f}\Pr\lrq{F = f | S= s} \\
            0 & \text{ otherwise}
        \end{cases},
    \end{equation}
    where $a$ is the cardinality of subset of functions that maximize the posterior probability weight. In simpler words, the learner outputs a hypothesis that has maximum posterior probability weight uniformly at random. This learner is similar to the Maximum A Posteriori (MAP) hypothesis in machine learning. The success probability of this learner for a data set $s$ is given by
    \begin{equation}
        \Pr[L_{\mathrm{MAP}}(S) =F |S=s] = \frac{1}{a} \sum_{f \in \argmax_{f}\Pr\lrq{F = f | S= s}} \Pr\lrq{F = f | S= s} =  \max_f\Pr\lrq{F = f | S= s}\, .
    \end{equation}
    Thus, the learner $L_{\mathrm{MAP}}$ achieves the optimal success probability for any given data set $s$. Therefore, its expected success probability of identification is
    \begin{equation}
        \Pr[\text{success for } L_{\mathrm{MAP}}] = \Ex_S \max_f\Pr\lrq{F = f | S= s}\, .
    \end{equation}
    Thus, the optimal average success probability of correctly identifying the unknown function is $\Ex_S \max_f\Pr\lrq{F = f | S= s}$, and it is achieved by $L_{\mathrm{MAP}}$.
\end{proof}

\begin{lemma}\label{lemma:random_consistent_hypothesis}
    If the prior over the function class $\mathcal{F}$ is uniform, then randomly outputting a consistent hypothesis is optimal for identifying the unknown function.
\end{lemma}
\begin{proof}
    From \Cref{lemma:optimal_strategy}, we know that the optimal strategy is given by randomly outputting a hypothesis that maximises the posterior. Now, using Bayes' rule,
    \begin{equation}
        \Pr\lrq{F = f | S= s} = \frac{\Pr\lrq{S= s~|~ F = f}\Pr[F= f]}{\Pr[S = s]}
    \end{equation}
    Define the set of hypotheses consistent with sample $s$ as
    \begin{equation}
        V_s = \lrb{f \in \mc{F}~|~ f \text{ is  consistent with } s}\, .
    \end{equation}
    Then, we can write
    \begin{equation}
        \Pr[S = s] = \sum_{h \in V_s} \Pr[S = s~|~ F = h] \Pr[F = h]\, .
    \end{equation}
    The term $\Pr\lrq{S= s~|~ F = f}\Pr[F= f]$ in the numerator of our expression of interest is non-zero only if $f \in V_s$, so we can write
    \begin{equation}
        \Pr\lrq{F = f | S= s} = \frac{\mathbf{1}_{f \in V_s}\Pr\lrq{S= s~|~ F = f}\Pr[F= f]}{\sum_{h \in V_s} \Pr[S = s~|~ F = h] \Pr[F = h]}\, .
    \end{equation}
    Now, notice that for any $h_1, h_2 \in V_s$,
    \begin{equation}
        \Pr[S = s~|~ F = h_1] = \Pr[S = s~|~ F = h_2]\, .
    \end{equation}
    That is, the likelihood of a sample $s$ given a hypothesis $h$ is constant over all consistent hypotheses $h\in V_s$. Indeed, if we write $S=(X,F(X))$, then we have $\Pr[S = s~|~F=h]
        = \mathbf{1}_{h \in V_s} \Pr[X=x]$.
    Therefore,
    \begin{equation}
         \Pr\lrq{F = f ~|~ S= s} = \frac{\mathbf{1}_{f \in V_s}\Pr[F= f]}{\sum_{h \in V_s} \Pr[F = h]}\, .
    \end{equation}
    For a uniform random prior over the class, i.e., $\Pr[F= f] = 1/|\mc{F}|$, this becomes
    \begin{equation}
        \Pr\lrq{F = f ~|~ S= s} = \begin{cases}
            \frac{1}{|V_s|} & \text{if } f \in V_s, \\
            0 & \text{otherwise}
        \end{cases}\, .
    \end{equation}
    Thus, each consistent hypothesis has same posterior weight. Thus, the optimal success probability is $\Ex_S 1/|V_S|$, and it is achieved by the learner that outputs random consistent hypothesis.
\end{proof}

Now, we are ready to prove the statement of \Cref{prop:sa-con}.
\begin{proof}[Proof of \Cref{prop:sa-con}] 
      Let $F$ denote a variable distributed uniformly over $\mathcal{F}$. Let $Z_t = \left((X^{(i)}, F(X^{(i)}))\right)_{i=1}^t$ denote a sample of size $t$, with the $X^{(i)}$ drawn i.i.d.~according to the input distribution and with labels produced according to $F$. Consider the $t$-sample learner $L_{SA}$ that first applies $T_{SA}$ to the given data set and then outputs a hypothesis chosen uniformly at random from the space of consistent hypotheses. Then, the distinguisher in \Cref{alg:distinguisher-general-lower-bound} accepts on an amplified sample precisely when $L_{SA}$ outputs the true function. Therefore, 
      \begin{align}
          \Ex_f \Pr\lrq{\mc{A}^f(Z_{f, t+1}^{SA}) = 1}
          &= \Pr[L_{SA}(Z_t)=F]\\
          &= \Ex_{Z_t}\Pr[L_{SA}(Z_t)=F~|~Z_t]\\
          &\leq \Ex_{Z_t}\max_{f\in\mathcal{F}}\Pr[F=f~|~Z_t] \, ,
      \end{align}
      where in the second equality we conditioned on the sample, and in the final inequality we appealed to MAP-optimality as proved in \Cref{lemma:optimal_strategy}.
      As we are considering a uniform prior over $\mathcal{F}$, we can apply \Cref{lemma:random_consistent_hypothesis} to see that the learner $L_{\mathrm{MAP}}$ that outputs a random consistent hypothesis is Bayes-optimal with 
      \begin{equation}
          \Ex_{Z_t}\max_{f\in\mathcal{F}}\Pr[F=f~|~Z_t]
          = \Pr[L_{\mathrm{MAP}}(Z_t)=F] \, .
      \end{equation}
      For a true sample, the space of consistent hypotheses is always non-empty, so by construction of \Cref{alg:distinguisher-general-lower-bound} we get
      \begin{equation}
          \Pr[L_{\mathrm{MAP}}(Z_t)=F]
          = \Ex_f \Pr\lrq{\mc{A}^f(Z_{f, t}) = 1}\, .
      \end{equation}
      Combining the inequalities above proves the proposition.
\end{proof}

{
\begin{lemma}\label{lemma:prop_bigraph}
    For states in $\mathrm{Bigraph}_{3n}$, the following properties hold:
    \begin{enumerate}
        \item The state $|A\rangle$ labelled by matrix $A \in \Z_{2}^{n \times 2n}$ is a $(3n)$-qubit stabilizer state with stabilizer generators
        \begin{equation}
            \left\{
                X_i \prod_{j=1}^{2n} Z_{n+j}^{A_{i,j}}
            \right\}_{i=1}^{n}
            \cup
            \left\{
                \left(\prod_{j=1}^{n} Z_j^{A_{j,k}}\right)X_{n+k}
            \right\}_{k=1}^{2n},
        \end{equation}\, 
        where the subscripts indicate the qubits on which the Pauli operators act.
        \item For two different $A, B \in \Z_{2}^{n \times 2n}$, the states $|A\rangle$ and $|B\rangle$ are different stabilizer states.
    \end{enumerate}
\end{lemma}
\begin{proof}
    \begin{enumerate}
        \item We note that
        \begin{equation}
            |A\rangle = \frac{1}{\sqrt{2^{3n}}}\sum_{\substack{{x \in \Z_2^{n}}\\ {y \in \Z_2^{2n}}}}(-1)^{x^T A y} |x\rangle |y\rangle \, .
        \end{equation}
        Now, we act on this state by a Pauli operator $X_i \prod_{j=1}^{2n} Z_{n+j}^{A_{i,j}}$ for some $i \in [n]$:
        \begin{equation}
            \begin{split}
                X_i \prod_{j=1}^{2n} Z_{n+j}^{A_{i,j}} |A\rangle &= \frac{1}{\sqrt{2^{3n}}} \sum_{\substack{{x \in \Z_2^{n}}\\ {y \in \Z_2^{2n}}}}(-1)^{x^T A y}\lr{X_i \prod_{j=1}^{2n} Z_{n+j}^{A_{i,j}}} |x\rangle |y\rangle\\
                &= \frac{1}{\sqrt{2^{3n}}} \sum_{\substack{{x \in \Z_2^{n}}\\ {y \in \Z_2^{2n}}}}(-1)^{x^T A y + \sum_{j}A_{i, j} y_j} |x_1 \ldots \bar{x}_i \ldots x_n \rangle |y\rangle\, \\
                &= \frac{1}{\sqrt{2^{3n}}}\sum_{\substack{{x \in \Z_2^{n}}\\ {y \in \Z_2^{2n}}}} (-1)^{x^T A y} |x\rangle |y\rangle\, ,
                \end{split}
        \end{equation}
        where the last equality follows by changing variable $\bar{x}_i \rightarrow x_i$ and combining terms in the phase. Since the choice of $i$ is arbitrary, $\lrb{ X_i \prod_{j=1}^{2n} Z_{n+j}^{A_{i,j}}}_{i = 1}^n$ are $n$ stabilizers of the state $|A\rangle$. A similar calculation can be used to show that $\lrb{\left(\prod_{j=1}^{n} Z_j^{A_{j,k}}\right)X_{n+k}}_{k = 1}^{2n}$ all stabilizes the state $|A\rangle$. Furthermore, they are independent and thus form a set of $(3n)$-independent stabilizer generators of the $(3n)$-qubit state $|A\rangle$. 
        \item We prove this by contradiction. Suppose that for two different matrices $A, B \in \Z_2^{n \times 2n}$, we have that $|A\rangle = |B\rangle$. Then, applying Hadamard transform on last $(2n)$-qubits gives the states
        \begin{equation}
            I \otimes H^{\otimes 2n}|A\rangle = \frac{1}{\sqrt{2^n}} \sum_x |x\rangle |x^T A\rangle \,, \hspace{3ex} I \otimes H^{\otimes 2n}|B\rangle = \frac{1}{\sqrt{2^n}} \sum_x |x\rangle |x^T B\rangle\, .
        \end{equation}
        By assumption, these states are also equal and hence, $x^T A = x^T B, \forall x \in \Z_2^n$. This implies that $A = B$ leading to a contradiction. 
    \end{enumerate}
\end{proof}
}

\begin{proof}[Proof of \Cref{lemma:lifting-stab-states}]
    We show this through a proof by contradiction. Assume that there exist a $(t-1, t, \epsilon)$ quantum cloner $\Lambda$ for $\mathrm{Stab}_{n+1}$. Then, we simply define a new cloner $\Lambda^{\prime}$ for $n$-qubit stabilizer states as
    \begin{equation}
        \Lambda^{\prime} = \tr_{1}\Lambda \circ V^{\otimes t-1},
    \end{equation}
    where $V$ is the isometry from $2^n$ dimensions to $2^{n+1}$ dimensions that acts as $V|\psi\rangle = |\psi\rangle \otimes |0\rangle$, and where $\tr_1$ traces out the $(n+1)$st register in every copy. Now, notice that for any $n$-qubit stabilizer state $|\psi\rangle$, $V|\psi\rangle = |\psi\rangle \otimes |0\rangle$ is an $(n+1)$-qubit stabilizer state by Claim~\ref{claim_lifting_stab_states}, and
    \begin{equation}
        \td\left(\Lambda^{\prime}(|\psi\rangle^{\otimes t-1}), |\psi\rangle^{\otimes t}\right) \leq \td\left(\Lambda \circ V^{\otimes t-1}(|\psi\rangle^{\otimes t-1}), |\psi 0\rangle^{\otimes t}\right) ,
    \end{equation}
    which follows again by using the contraction property of trace distance under CPTP maps. Now,
    \begin{equation}
        \begin{split}
            \td\left(\Lambda \circ V^{\otimes t-1}(|\psi\rangle^{\otimes t-1}), |\psi 0\rangle^{\otimes t}\right)  &= \td\left(\Lambda(|\psi 0\rangle^{\otimes t-1}), |\psi0\rangle^{\otimes t}\right) \leq \epsilon.
        \end{split}
    \end{equation}
    This gives a $( t-1, t, \epsilon)$ cloner for $n$-qubit stabilizer states by using the $(t-1, t, \epsilon)$ cloner for $n+1$ qubit stabilizer states, but this contradicts the assumption. 
\end{proof}

\begin{claim}\label{claim_lifting_stab_states}
    For any $n$-qubit stabilizer state $|\psi\rangle$, $|\psi 0\rangle$ is a $(n+1)$-qubit stabilizer state.
\end{claim}
\begin{proof}
    To show this, we create stabilizer generators of $|\psi0\rangle$. Let $g_1, \ldots g_n$ be a set of stabilizer generators for $|\psi\rangle$. Take the set
    \begin{equation}
        A = \{g_1 \otimes I, \ldots g_n \otimes I, I^{\otimes n} \otimes Z\}.
    \end{equation}
    Notice that the elements of $A$ stabilize $|\psi 0\rangle$. Moreover, they can be shown to be independent by leveraging the independence of $g_i$, and they also commute. Hence, the above $n+1$ elements are independent stabilizer generators for $|\psi 0\rangle$, showing that $|\psi 0\rangle$ is an $(n+1)$-qubit stabilizer state.
\end{proof}

\begin{proof}[Proof of \Cref{claim:t-twirl-state-closure}]
    Notice that, for any $ g^{\prime} \in G$,
    \begin{align*}
        \tr\left(\mu(g^{\prime})\sigma\right) &= \tr\left(\mu(g^{\prime})\frac{1}{|G|}\sum_{g \in G} \mu(g) \rho \mu(g)^{\dagger}\right) \\
        &= \frac{1}{|G|}\sum_{g \in G}\tr\left(\mu(g^{\prime}) \mu(g) \rho \mu(g)^{\dagger}\right)\tag{\stepcounter{equation}\theequation} \\
        &\stackrel{(a)}{=} \frac{1}{|G|}\sum_{g \in G}\tr\left(\mu(g) \mu(g^{\prime}) \rho \mu(g)^{\dagger}\right) \\
        &\stackrel{(b)}{=} \frac{1}{|G|}\sum_{g \in G}\tr\left(\mu(g^{\prime}) \rho\right) = \tr\left(\mu(g^{\prime}) \rho\right)\, ,
    \end{align*}
    where $(a)$ follows because the group is Abelian, and $(b)$ follows from cyclic property of trace. Hence, $\rho \in \Psi_H^{\epsilon}$ if and only if $\sigma \in \Psi_H^{\epsilon}$.

    Now, to show the second part, take any $g^{\prime} \in G$, then,
    \begin{equation}
        \begin{split}
            \mu(g^{\prime})\sigma &= \frac{1}{|G|}\sum_{g \in G}\mu(g^{\prime}) \mu(g) \rho \mu(g)^{\dagger} \\
            &=  \frac{1}{|G|}\sum_{g \in G} \mu(g^{\prime}g) \rho \mu(g)^{\dagger} \\
            &\stackrel{(c)}{=} \frac{1}{|G|}\sum_{g^{\prime \prime} \in G} \mu(g^{\prime \prime}) \rho \mu(g^{\prime \prime}(g^{\prime })^{-1})^{\dagger} \\
            &\stackrel{(d)}{=} \frac{1}{|G|}\sum_{g^{\prime \prime} \in G} \mu(g^{\prime \prime}) \rho \mu(g^{\prime \prime})^{\dagger} \mu(g^{\prime}) = \sigma \mu(g^{\prime})\, ,
        \end{split}
    \end{equation}
    where $(c)$ follows by redefining $g^{\prime \prime} = g^{\prime}g = g g^{\prime}$, and $(d)$ follows by Abelianity and noticing that $\mu((g^{\prime })^{-1})^{\dagger} = \mu(g^{\prime})$. 
\end{proof}

\section{No-Go Theorem for Structured Random Purification}\label{appendix:purification-no-go}

In this section, we state and prove two results about the non-existence of certain structured forms of random purification channels. As explained in \Cref{footnote:mistake}, this is in sharp contrast to a previous version of this paper, which mistakenly claimed to construct such channels.

First, we prove a non-existence result about structured random purification channels for specific Abelian StateHSPs.

\begin{theorem}
    Let $G$ be an Abelian group such that $|G| > 1$ and with the unitary representation $\mu(g)$ as in \Cref{eq:new_rep_assumption}. There does not exist a channel $\Lambda$ such that
    \begin{equation}\label{eq:structured-random-purification-stateHSP}
        \Lambda(\sigma^{\otimes t})=  \Ex_g |\sigma_g\rangle \langle \sigma_g|^{\otimes t}
    \end{equation}
    holds for all mixed (i.e., not necessarily pure) instances $\sigma$ that hides some subgroup of $G$ under the unitary representation $\mu$ for any $t > 1$, where $|\sigma_g\rangle = (I \otimes \mu(g)) |\sigma\rangle$ with $|\sigma\rangle =(\sqrt{\sigma} \otimes I) \sum_{\lambda} |\lambda \lambda\rangle$ the standard purification.
\end{theorem}

This theorem and its proof in particular demonstrate that Theorem 9 from the first arXiv version of our paper was false.

\begin{proof}
    Consider the mixed Abelian StateHSP instances $\sigma_H = \frac{1}{|H^{\perp}|}\sum_{\lambda \in H^{\perp}} |\lambda\rangle \langle \lambda|$ for any subgroup $H\leq G$. Assume for contradiction that for some $t>1$, there exists a structured random purification CPTP map $\Lambda$ that satisfies \Cref{eq:structured-random-purification-stateHSP} for any $\sigma\in\{\sigma_H\}_H$. Then, by linearity,
    \begin{equation}
        \Lambda(\sigma_H^{\otimes t}) = \frac{1}{|H^{\perp}|^t}\sum_{\bar{\lambda} \in \lr{H^{\perp}}^t} \Lambda\lr{|\bar{\lambda}\rangle \langle \bar{\lambda}|} \, ,
    \end{equation}
    where $\bar{\lambda} = \lambda_1 \cdots \lambda_t$. Now, we claim that the image (which equals the support for density matrices) of $\Lambda(|\bar{\lambda}\rangle \langle \bar{\lambda}|)$ is contained in the image of $\Lambda(\sigma_H^{\otimes t})$. To show this, we note that for any $|v\rangle \in \mathrm{ker}(\Lambda(\sigma_H^{\otimes t}))$, we get
    \begin{equation}
        0=\langle v | \Lambda(\sigma_H^{\otimes t}) |v\rangle = \frac{1}{|H^{\perp}|^t}\sum_{\bar{\lambda} \in \lr{H^{\perp}}^t} \langle v|\Lambda\lr{|\bar{\lambda}\rangle \langle \bar{\lambda}|}|v\rangle\, .
    \end{equation}
    From positivity of $\Lambda\lr{|\bar{\lambda}\rangle \langle \bar{\lambda}|}$, we get that $|v\rangle\in \mathrm{ker}(\Lambda(\sigma_H^{\otimes t}))$ implies $|v\rangle\in \mathrm{ker}\lr{\Lambda\lr{|\bar{\lambda}\rangle \langle \bar{\lambda}|}}$ for all $\bar{\lambda} \in \lr{H^{\perp}}^t$. Thus, $\mathrm{ker}\lr{\Lambda(\sigma_H^{\otimes t})} \subseteq \mathrm{ker}\lr{\Lambda\lr{|\bar{\lambda}\rangle \langle \bar{\lambda}|}}$ and hence $\mathrm{Im}\lr{\Lambda\lr{|\bar{\lambda}\rangle \langle \bar{\lambda}|}} \subseteq \mathrm{Im}\lr{\Lambda(\sigma_H^{\otimes t})}$ for all $\bar{\lambda} \in \lr{H^{\perp}}^t$, since the support (which, again, equals the image for density matrices) is the orthogonal complement of kernel. Also, notice that $\bar{\lambda} \in \lr{H^{\perp}}^t$ iff $\lambda_1, \ldots, \lambda_t \in H^{\perp}$. 
    Thus, we have
    \begin{equation}
    \mathrm{Im}\lr{\Lambda\lr{|\bar{\lambda}\rangle \langle \bar{\lambda}|}} \subseteq \bigcap_{H: \bar{\lambda} \in \lr{H^{\perp}}^t} \mathrm{Im}\lr{\Lambda(\sigma_H^{\otimes t})}\, .
    \end{equation}
    Now, note that, by assumption, 
    \begin{equation}
        \begin{split}
        \Lambda(\sigma_H^{\otimes t}) = \Ex_g |\sigma_g\rangle \langle \sigma_g|^{\otimes t} &= \Ex_g \frac{1}{|H^{\perp}|^t} \sum_{\bar{\lambda}, \bar{\eta} \in \lr{H^{\perp}}^t} \chi_{\bar{\lambda}}(g) \overline{\chi_{\bar{\eta}}(g) }|\bar{\lambda}\rangle \langle \bar{\eta}| \otimes |\bar{\lambda}\rangle \langle \bar{\eta}|\, \\
        &\stackrel{(a)}{=} \frac{1}{|H^{\perp}|^t}\sum_{\bar{\lambda} \cong \bar{\eta}} |\bar{\lambda} \bar{\lambda}\rangle \langle \bar{\eta} \bar{\eta}|\, \\
        &= \frac{1}{|H^{\perp}|^t}\sum_{\{\bar{\lambda}\}}\sum_{\bar{\eta}_1, \bar{\eta}_2 \cong \bar{\lambda}} |\bar{\eta}_1 \bar{\eta}_1\rangle \langle \bar{\eta}_2 \bar{\eta}_2|\, ,
        \end{split}
    \end{equation}
    where $(a)$ follows from the orthogonality conditions of characters.
    We denote the number of isomorphic irreps of $\bar{\lambda}$ by $f(\{\bar{\lambda}\})$. Thus, we can write the above as
    \begin{equation}\label{eq:cloned-state-multiplicity-decomposed}
        \Lambda(\sigma_H^{\otimes t}) =   \sum_{\{\bar{\lambda}\}}\frac{f(\{\bar{\lambda}\})}{|H^{\perp}|^t} |\psi_{\{\bar{\lambda}\}}\rangle \langle \psi_{\{\bar{\lambda}\}}|\, ,
    \end{equation}
    where $|\psi_{\{\bar{\lambda}\}}\rangle = \frac{1}{\sqrt{f(\{\bar{\lambda}\})}} \sum_{\bar{\eta} \cong \bar{\lambda}} |\bar{\eta} \bar{\eta}\rangle$. Notice that $\{|\psi_{\{\bar{\lambda}\}}\rangle\}_{\{\bar{\lambda}\}}$ is a set of orthonormal vectors because they are supported on disjoint set of elements. In particular, we see that $\mathrm{Im}\lr{\Lambda(\sigma_H^{\otimes t})} = \mathrm{span}\lr{\{|\psi_{\{\bar{\lambda}\}}\rangle\}_{\{\bar{\lambda}\}}}$ with $\{\bar{\lambda}\}\in \lr{H^{\perp}}^t$. 
    
    We aim to find the intersection of images of $\Lambda(\sigma_H^{\otimes t})$ for different choices of $H$ containing a fixed $t$-tuple $\bar{\eta}$. We consider specifically $\bar{\eta} = \eta \cdots \eta$ for $\eta = 0$. 
    Notice that
    \begin{equation}
        \bigcap_{H: 0 \in H^{\perp}} H^{\perp} = \{0\}\, ,
    \end{equation}
    simply because $G^{\perp}$ is just $\{0\}$. This implies $\bigcap_{H: 0 \in H^{\perp}} \lr{H^{\perp}}^t = \{\bar{0}\}$. Now, define the set $C_{H}^t =\mathrm{span}\lr{\{|\bar{\lambda} \bar{\lambda}\rangle\}_{\bar{\lambda} \in \lr{H^{\perp}}^t}}$. It is easy to see from \Cref{eq:cloned-state-multiplicity-decomposed} that
    \begin{equation}
        \mathrm{Im}\lr{\Lambda(\sigma_H^{\otimes t})} \subseteq C_H^t\, ,
    \end{equation}
    because each $|\psi_{\{\bar{\lambda}\}}\rangle$ lies in $\mathrm{span}\lr{\{|\bar{\lambda} \bar{\lambda}\rangle\}_{\bar{\lambda} \in \lr{H^{\perp}}^t}}$. Thus, $\bigcap_{H: \bar{0} \in \lr{H^{\perp}}^t} \mathrm{Im}\lr{\Lambda(\sigma_H^{\otimes t})} \subseteq \bigcap_{H: \bar{0} \in \lr{H^{\perp}}^t} C_H^t$. But since $\bigcap_{H: 0\in H^{\perp}} \lr{H^{\perp}}^t = \{\bar{0}\}$, we get that 
    \begin{equation}
        \bigcap_{H, \bar{0} \in H^{\perp: t}} C_H^t = \mathrm{span}\lr{|\bar{0}\bar{0}\rangle}\, .
    \end{equation}
    Combined with the above, we have shown
    \begin{equation}\label{eq:cloner-image-inclusion}
        \mathrm{Im}\lr{\Lambda\lr{|\bar{0}\rangle \langle \bar{0}|}} \subseteq \bigcap_{H: \bar{0} \in \lr{H^{\perp}}^t} \mathrm{Im}\lr{\Lambda(\sigma_H^{\otimes t})} \subseteq \mathrm{span}\lr{|\bar{0}\bar{0}\rangle}.
    \end{equation} 
    
    Next, we aim to show that $\bigcap_{H: \bar{0} \in \lr{H^{\perp}}^t} \mathrm{Im}\lr{\Lambda(\sigma_H^{\otimes t})}$ does not have any non-trivial elements. For any non-trivial subgroup $H^{\perp}$, we find the orthogonal projection of $|\bar{0}\bar{0}\rangle$ onto the image of ${\Lambda(\sigma_H^{\otimes t})}$. Define $\Pi_{\lr{H^{\perp}}^t} = \sum_{\{\bar{\lambda}\}}|\psi_{\{\bar{\lambda}\}}\rangle \langle \psi_{\{\bar{\lambda}\}}|$, where the sum is over all those $\bar{\lambda}$ that define the cosets inside $\lr{H^{\perp}}^t$. Then we have
    \begin{equation}
           \langle \bar{0}\bar{0}|\Pi_{\lr{H^{\perp}}^t}|\bar{0}\bar{0}\rangle = \langle \bar{0}\bar{0}|\sum_{\{\bar{\lambda}\}}|\psi_{\{\bar{\lambda}\}}\rangle \langle \psi_{\{\bar{\lambda}\}}|\bar{0}\bar{0}\rangle = \langle \bar{0}\bar{0} |\psi_{\{\bar{0}\}}\rangle \langle \psi_{\{\bar{0}\}}| \bar{0}\bar{0}\rangle 
           = \frac{1}{f(\{\bar{0}\})} < 1\, ,
    \end{equation}
    where the second equality follows by noting that the tuple $\bar{0}$ is contained in the coset $\{\bar{0}\}$ and thus the state $|\bar{0}\bar{0}\rangle$ has non-zero overlap only with $|\psi_{\{\bar{0}\}}\rangle$, and the last inequality follows from $f(\{\bar{0}\}) > 1$, which is true for non-trivial $H^{\perp}$. To see this, suppose that $H^{\perp}$ contains some non-trivial element $a$. Then, the irrep $\bar{0}$ is isomorphic to any irrep of the form $(a,a^{-1},a,a^{-1},\ldots,a,a^{-1},0,\ldots,0)$ because the characters satisfy $\chi_a(g) \chi_{a^{-1}}(g) = \chi_{a a^{-1}}(g) = 1$. Hence, for non-trivial $H^{\perp}$, $\mathrm{Im}\lr{\Lambda\lr{\sigma_H^{\otimes t}}}$ does not contain $|\bar{0}\bar{0}\rangle$, because an orthogonal projection of $|\bar{0}\bar{0}\rangle$ onto the image of $\sigma_H^{\otimes t}$ does not give back the same vector. 

    Thus, $\bigcap_{H: \bar{0} \in \lr{H^{\perp}}^t} \mathrm{Im}\lr{\Lambda(\sigma_H^{\otimes t})}$ does not contain $|\bar{0}\bar{0}\rangle$. Together with \Cref{eq:cloner-image-inclusion}, this implies that $\bigcap_{H: \bar{0} \in \lr{H^{\perp}}^t} \mathrm{Im}\lr{\Lambda(\sigma_H^{\otimes t})}$ is a strict subset of $\mathrm{span}\lr{|\bar{0}\bar{0}\rangle}$ and thus the trivial subspace. Therefore,
    \begin{equation}
        \mathrm{Im}\lr{\Lambda\lr{|\bar{0}\rangle\langle \bar{0}| }} \subseteq \bigcap_{H: \bar{0} \in \lr{H^{\perp}}^t} \mathrm{Im}\lr{\Lambda(\sigma_H^{\otimes t})} = \mathrm{span}\lr{0}\, .
    \end{equation}
    This means that $\Lambda\lr{|\bar{0} \rangle \langle \bar{0}| } = 0$ which contradicts the trace preserving property of CPTP maps.
    \end{proof}

In fact, we prove a stronger statement for Bell-diagonal states: We show that, for any $t \geq 4$, there does not exist a structured random purification channel which takes any $t$ i.i.d.~copies of a Bell-diagonal state to $t$ i.i.d.~copies of a Clifford-random purification.

\begin{theorem}\label{thm:no-go-srp-bell-diagonal-states}
    Let $\mc{S}$ be the set of states that are diagonal in Bell basis and supported on some linear subspace $L^{\perp}$, i.e.,
    \begin{equation}
        \mc{S} = \lrb{\sigma_L~\big|~ \sigma_L = \frac{1}{{|L^{\perp}|}}\sum_{\lambda \in L^{\perp}}|\Phi_\lambda\rangle\langle \Phi_\lambda|}\, .
    \end{equation}
    Then there does not exist a structured random purification CPTP map $\Lambda$ such that
    \begin{equation}
        \Lambda(\sigma_L^{\otimes t}) = \Ex_C (I \otimes C^{\otimes t}) \Psi_{L}^{\otimes t} (I \otimes C^{\dagger, \otimes t})\, ,
    \end{equation}
    holds simultaneously for all $L$ for some $t \geq 4$, where the expectation is over Clifford unitaries and where $\Psi_L = |\Psi_L\rangle \langle \Psi_L |$ with $\ket{\Psi_L} = (\sqrt{\sigma_L} \otimes I) \sum_{\lambda} |\Phi_{\lambda} \Phi_{\lambda}\rangle$ being the standard purification.
\end{theorem}
\begin{proof}
    We prove this by contradiction. Suppose that a channel $\Lambda$ exists so that for some $t \geq 4$,
    \begin{equation}\label{eq:clifford-no-random-purify-assumption}
        \Lambda\lr{\sigma_L^{\otimes t}} = \Ex_C (I \otimes C^{\otimes t}) \Psi_{L}^{\otimes t} (I \otimes C^{\dagger, \otimes t})
    \end{equation}
    holds for all linear subspaces $L\leq \Z_2^{2n}$. That is, the CPTP map $\Lambda$ takes any mixed state of the form $\sigma_L$ to a Clifford random purification. Then, from linearity, we have that
    \begin{equation}\label{eq:linearity-thm-no-go-rp}
        \Lambda\lr{\sigma_L^{\otimes t}} = \frac{1}{{|L^{\perp}|^t}}\sum_{\bar{\lambda}\in \lr{L^{\perp}}^t} \Lambda\lr{|\Phi_{\bar{\lambda}}\rangle \langle \Phi_{\bar{\lambda}}|}\, .
    \end{equation}
    Thus, $\mathrm{Im}\lr{\Lambda\lr{|\Phi_{\bar{\lambda}}\rangle \langle \Phi_{\bar{\lambda}}|}} \subseteq \mathrm{Im}\lr{\Lambda\lr{\sigma_{L}^{\otimes t}}}$ and for $\bar{\lambda} = \bar{0}$, we have that
    \begin{equation}
        \mathrm{Im}\lr{\Lambda(|\Phi_{\bar{0}}\rangle \langle \Phi_{\bar{0}})} \subseteq \bigcap_{L: 0 \in L^{\perp}}\mathrm{Im}\lr{\Lambda\lr{\sigma_{L}^{\otimes t}}}\, .
    \end{equation}
    Moreover, taking partial trace in \Cref{eq:linearity-thm-no-go-rp}, we get
    \begin{equation}
        \tr_B\lr{\Lambda\lr{\sigma_L^{\otimes t}}} = \frac{1}{{|L^{\perp}|^t}}\sum_{\bar{\lambda}\in \lr{L^{\perp}}^t} \tr_B\lr{\Lambda\lr{|\Phi_{\bar{\lambda}}\rangle \langle \Phi_{\bar{\lambda}}|}}\, .
    \end{equation}
    Using positivity of the partial trace, we have that
    \begin{equation}
        \mathrm{Im}\lr{\tr_B\Lambda(|\Phi_{\bar{0}}\rangle \langle \Phi_{\bar{0}}|)} \subseteq \bigcap_{L: 0 \in L^{\perp}}\mathrm{Im}\lr{\tr_B\Lambda\lr{\sigma_{L}^{\otimes t}}} = \bigcap_{L: 0 \in L^{\perp}}\mathrm{Im}\lr{{\sigma_{L}^{\otimes t}}}\, ,
    \end{equation}
    where the partial trace is over purifying register and the last inequality follows by the assumption that $\Lambda$ is a random purification channel. Also, notice that $\bigcap_{L: 0 \in L^{\perp}} L^{\perp} = \{0\}$, thus 
    \begin{equation}
        \mathrm{Im}\lr{\tr_B\Lambda(|\Phi_{\bar{0}}\rangle \langle \Phi_{\bar{0}}|)} \subseteq\bigcap_{L: 0 \in L^{\perp}}\mathrm{Im}\lr{{\sigma_{L}^{\otimes t}}} = \mathrm{span}\lr{|\Phi_{\bar{0}}\rangle}\, .
    \end{equation}
    Hence, from trace preservation and since the reduced state after tracing out the purifying register is pure, we get that
    \begin{equation}
            \Lambda(|\Phi_{\bar{0}}\rangle \langle \Phi_{\bar{0}}|) = |\Phi_{\bar{0}}\rangle \langle \Phi_{\bar{0}}| \otimes \rho\, ,
    \end{equation}
    for some state $\rho$ on the purifying register. 
    Now, define the projector $\Pi_{{\bar{0}}} = |\Phi_{\bar{0}}\rangle \langle \Phi_{\bar{0}}| \otimes I$. Then, from above, we have
    \begin{equation}
        \tr\lr{\Lambda(|\Phi_{\bar{0}}\rangle \langle\Phi_{\bar{0}}|) \Pi_{\bar{0}}} = 1\, .
    \end{equation}
    
    Next, we evaluate the projector overlap with $\Lambda\lr{\sigma^{\otimes t}}$ for any $\sigma_L$ with $|\Phi_0\rangle$ in its image:
    \begin{equation}
        \tr\lr{\Lambda(\sigma_L^{\otimes t})\Pi_{\bar{0}}} = \tr \lr{\sigma_L^{\otimes t}|\Phi_{\bar{0}}\rangle\langle\Phi_{\bar{0}}|} = \frac{1}{|L^{\perp}|^t}\, ,
    \end{equation}
    where we used that, by assumption, $\Lambda(\sigma_L^{\otimes t})$ is a random purification of $\sigma_L$.
    At the same time, 
    \begin{equation}
        \tr\lr{\Lambda(\sigma_L^{\otimes t})\Pi_{\bar{0}}} =
        \tr\lr{\frac{1}{|L^{\perp}|^t} \sum_{\bar{\lambda}} \Lambda(|\Phi_{\bar{\lambda}}\rangle \langle \Phi_{\bar{\lambda}}|) \Pi_{\bar{0}}} = \frac{1}{|L^{\perp}|^t}\lr{1 + \sum_{\bar{\lambda}\neq \bar{0}} \tr\lr{\Lambda(|\Phi_{\bar{\lambda}}\rangle \langle \Phi_{\bar{\lambda}}|) \Pi_{\bar{0}}}}\, .
    \end{equation}
    Thus, we conclude that 
    \begin{equation}
        \tr\lr{\Lambda\lr{|\Phi_{\bar{\lambda}}\rangle \langle \Phi_{\bar{\lambda}}|}\Pi_{\bar{0}}} = 0, \hspace{3ex} \forall \bar{\lambda} \in \lr{L^{\perp}}^t, \bar{\lambda} \neq \bar{0}\, .
    \end{equation}

    This implies that
    \begin{equation}\label{eq:clifford-random-purify-condition}
        \Pi_{\bar{0}}\Lambda(\sigma_L^{\otimes t})(I - \Pi_{\bar{0}}) = 0\, ,
    \end{equation}
   holds for all linear subspaces $L$ containing $0$. We fix a choice of $L$ so that $L^{\perp}$ also contains a non-trivial element $a$. 
    We aim to show that $\Pi_{\bar{0}}\Ex_C (I \otimes C^{\otimes t}) \Psi_{L}^{\otimes t} (I \otimes C^{\dagger, t}) (I - \Pi_{\bar{0}})\neq 0$, which will contradict the above equation and thus prove non-existence of such a random purification channel. To show this, we find the non-diagonal element
    \begin{equation}
        \begin{split}
            \lr{\langle \Phi_{\bar{0}}| \otimes I} \Ex_C (I \otimes C^{\otimes t}) \Psi_{L}^{\otimes t} (I \otimes C^{\dagger, \otimes t}) \lr{|\Phi_{\bar{\gamma}}\rangle \otimes I} &= \frac{1}{{|L^{\perp}|^t}} \sum_{\bar{\lambda}, \bar{\eta}} \langle \Phi_{\bar{0}}| \Phi_{\bar{\lambda}}\rangle \langle \Phi_{\bar{\eta}}|\Phi_{\bar{\gamma}}\rangle \otimes  \Ex_C C^{\otimes t} |\Phi_{\bar{\lambda}}\rangle \langle \Phi_{\bar{\eta}}| C^{\dagger, \otimes t}\\
            &=\frac{1}{{|L^{\perp}|^t}}  \Ex_C C^{\otimes t} |\Phi_{\bar{0}}\rangle \langle \Phi_{\bar{\gamma}}| C^{\dagger, \otimes t}\, .
        \end{split}
    \end{equation}
    Now, we want to show that for some $\bar{\gamma} \neq \bar{0}$, the right hand side in above equation is non-zero, i.e., 
    \begin{equation}
        \Ex_C C^{\otimes t} |\Phi_{\bar{0}}\rangle \langle \Phi_{\bar{\gamma}}| C^{\dagger, \otimes t} \neq 0, 
    \end{equation}
    for some $\bar{\gamma}$. As $\Ex_C C^{\otimes t} |\Phi_{\bar{0}}\rangle \langle \Phi_{\bar{\gamma}}| C^{\dagger, \otimes t}$ lies in the Clifford commutant, it suffices to show that it has non-zero overlap with one of the generators of the Clifford commutant. Assuming $t \geq 4$, one of the generators of Clifford commutant \cite{bittel2025completetheorycliffordcommutant} is given by
    \begin{equation}
        \Omega_4 = \frac{1}{2^{2n}}\sum_{P \in \mathbb{P}_{2n}} P^{\otimes 4} \otimes I^{\otimes t-4}\, ,
    \end{equation}
    where $\mathbb{P}_{2n} = \{I, X, Y, Z\}^{\otimes 2n}$. Now, we take $\bar{\gamma} = (a, a, a, a, x_1, \ldots, x_{t-4})$ for $x_1 = x_2 = \ldots = x_{t-4} = 0$ and find
    \begin{equation}
        \begin{split}
            \tr \lr{\Omega_4 \Ex_C C^{\otimes t} |\Phi_{\bar{0}}\rangle \langle \Phi_{\bar{\gamma}}| C^{\dagger, \otimes t}} &\stackrel{(a)}{=} \tr \lr{\Omega_4|\Phi_{\bar{0}}\rangle \langle \Phi_{\bar{\gamma}}| }\\
            &= \langle \Phi_{\bar{\gamma}}| \Omega_4 |\Phi_{\bar{0}}\rangle\, .
        \end{split}
    \end{equation}
    where $(a)$ follows by cyclicity of the trace and noting that $\Omega_4$ is invariant under the action of $t$-copy Cliffords. Now, notice that we can write this matrix element of $\Omega_4$ as,
    \begin{equation}
        \begin{split}
            \langle \Phi_{\bar{\gamma}}| \Omega_4 |\Phi_{\bar{0}}\rangle &= \frac{1}{2^{2n}}\langle \Phi_{a, a, a, a}|\sum_{P \in \mathbb{P}_{2n}} P^{\otimes 4} | \Phi_{0, 0, 0, 0}\rangle \\
            &= \frac{1}{2^{2n}} \sum_{P \in \mathbb{P}_{2n}} (\langle \Phi_{a} |P |\Phi_0\rangle)^4 \\
            &= \frac{1}{2^{2n}} \sum_{P \in \mathbb{P}_{2n}} (\langle {a} |P |0\rangle)^4 
        \end{split}
    \end{equation}
    where the first equality follows by tracing out all those tensor products where the $\Omega_4$ equals the identity, and where the third equality follows by noting that $|\Phi_{a}\rangle$ is related to $|a\rangle$ state by action of some $(2n)$-qubit Clifford unitary for any $a$. Now, we can write
    \begin{equation}
        \begin{split}
            \langle \Phi_{\bar{\gamma}}| \Omega_4 |\Phi_{\bar{0}}\rangle &= \frac{1}{2^{2n}}\sum_{P \in \mathbb{P}_{2n}} \lr{\langle a | P | 0\rangle}^4 \\
            &= \frac{1}{2^{2n}}\sum_{P \in \mathbb{P}_{2n}} \langle 0|P|0\rangle \\
            &= \frac{1}{2^{2n}} \prod_{i=1}^{2n} \sum_{P \in \{I, X, Y, Z\}} \langle 0 |P |0\rangle\, \\
            &= 1\, ,
        \end{split}
    \end{equation}
    where second equality follows by noting that $|a\rangle  = P|0\rangle$ for some Pauli $P$ and that the sum remains invariant if we multiply every element by a Pauli; the third equality follows by factorising the $(2n)$-qubit Pauli into single-qubit Paulis; and the last equality follows by noting that $\langle 0| P_i |0\rangle$ is non-zero (namely, equal to $1$) only for $P_1 \in \{I, Z\}$. Thus, we get that
    \begin{equation}
        \langle \Phi_{\bar{\gamma}}| \Omega_4 |\Phi_{\bar{0}}\rangle = 1\, ,
    \end{equation}
    and hence the overlap $\tr \lr{\Omega_4 \Ex_C C^{\otimes t} |\Phi_{\bar{0}}\rangle \langle \Phi_{\bar{\gamma}}| C^{\dagger, \otimes t}} = 1$ and thus 
    \begin{equation}\label{eq:clifford-random-purify-contradiction}
         \Ex_C C^{\otimes t} |\Phi_{\bar{0}}\rangle \langle \Phi_{\bar{\gamma}}| C^{\dagger, \otimes t} \neq 0\, .
    \end{equation}
    Thus, \Cref{eq:clifford-no-random-purify-assumption,eq:clifford-random-purify-condition,eq:clifford-random-purify-contradiction} are inconsistent, showing a contradiction. We conclude that there cannot exist a channel $\Lambda$ such that \Cref{eq:clifford-no-random-purify-assumption} holds for all linear subspaces $L$ for any $t \geq 4$.
\end{proof}

\begin{remark}
    For $t\leq 3$, there exists a structured random purification channel which takes $t$ i.i.d. copies of an arbitrary mixed stabilizer state and maps it to a Clifford-random purification. The channel is simply the random purification channel of \cite{tang2025conjugatequerieshelp} but with random Cliffords instead of Haar-random unitaries, and it works because Cliffords form a $t$-design for $t \leq 3$.
\end{remark}

\begin{remark}\label{remark:tagged-bell-states}
    Define the tagged-Bell basis as standard Bell basis states with an ancilla in the $|0\rangle$ state, $\{|\Phi_{\lambda}\rangle \otimes |0\rangle\}_{\lambda}$. Then, for any subspace $L \leq \Z_2^{2n^{\prime}}$, we take states diagonal in the tagged-Bell basis, i.e., 
    \begin{equation}
        \mc{S} = \lrb{\sigma_L~\big|~ \sigma_L = \frac{1}{{|L^{\perp}|}}\sum_{\lambda \in L^{\perp}}|\Phi_\lambda\rangle\langle \Phi_\lambda|\otimes |0\rangle \langle 0|}\, .
    \end{equation}
    For this set of states, we can also show that there does not exist a Clifford-random purification channel identical to that in the proof of \Cref{thm:no-go-srp-bell-diagonal-states}, with Bell basis states replaced by tagged-Bell basis states. This allows us to go from an even number of qubits as in \Cref{thm:no-go-srp-bell-diagonal-states} to an odd number of qubits.
\end{remark}

Finally, we show how to extend this impossibility result for Clifford-random purifications to also exclude more general purification channels. We define mixed stabilizer states as normalized projections on the stabilizer codes. That is, a mixed stabilizer state $\rho$ is specified by $k\leq n$ commuting stabilizer generators $\{P_1,\ldots,P_k\}$ and takes the form $\rho = 2^{-n}\prod_{j=1}^k (I+P_j)$.

\begin{theorem}
    Let $\mc{S}_{\mathrm{Mix-Stab}_n}$ be the set of mixed $n$-qubit stabilizer states for $n\geq 2$. 
    Let $\eta$ be some probability distribution supported on a set $\mathcal{U}$ of $n$-qubit unitaries. 
    Then, there does not exist a structured random purification CPTP map $\Lambda_{\eta, \mc{U}}$ that simultaneously satisfies the following two properties:
    \begin{enumerate}
        \item For some $t \geq 4$,
        \begin{equation}
            \Lambda_{\eta, \mc{U}}\lr{\sigma^{\otimes t}} = \Ex_{U \stackrel{\eta}{\leftarrow } \mc{U}} (I \otimes U^{\otimes t}) \Psi_{\sigma}^{\otimes t} (I \otimes U^{\dagger, \otimes t})
        \end{equation}
        holds for all $\sigma \in \mc{S}_{\mathrm{Mix-Stab}_n}$, where the expectation is over unitaries drawn from $\mathcal{U}$ according to $\eta$, and where $\Psi_{\sigma}$ denotes the standard random purification of the state $\sigma$.
        \item The states $(I \otimes U) \Psi_{\sigma} (I \otimes U^{\dagger})$ are pure stabilizer states for all $\sigma \in \mc{S}_{\mathrm{Mix-Stab}_n}$ and for all $U \in \mc{U}$.
    \end{enumerate}
\end{theorem}
\begin{proof}
   We first assume that $n$ is even. At the end of the proof, we argue that the argument can be extended to odd $n$. We argue by first showing that for the second property to hold, $\mc{U}$ can be taken to be the a subset of $n$-qubit Clifford unitaries. 
    Then, we argue that, in fact, $\mc{U}$ can be taken to be equal to the set of $n$-qubit Cliffords, and $\eta$ can be taken to be the uniform distribution over that set.
    This allows us to use \Cref{thm:no-go-srp-bell-diagonal-states} to argue that no such random purification channel exists.

    We start by taking the maximally mixed state $\sigma_{\mathrm{mm}} = \mathds{1}_2^{\otimes n}/2^n$, which is clearly a mixed $n$-qubit stabilizer state. Assume that for some $t \geq 4$ there exists a channel $\Lambda_{\eta, \mc U}$ that produces a pure stabilizer-random purification of the maximally mixed state in the sense of properties 1. and 2. That is, $\Lambda_{\eta, \mc U}(\sigma_{\mathrm{mm}})=\Ex_{U \stackrel{\eta}{\leftarrow } \mc{U}}(I \otimes U^{\otimes t}) \Psi_{\sigma_{\mathrm{mm}}}^{\otimes t} (I \otimes U^{\dagger, \otimes t})$, and $\Psi_{\sigma_{\mathrm{mm}}}^U = (I \otimes U)\Psi_{\sigma_{\mathrm{mm}}} (I \otimes U^{\dagger})$ is a stabilizer state for any $U \in \mc U$. The latter means that, for any $U \in \mc U$, the state $\Psi_{\sigma_{\mathrm{mm}}}^U$ has a stabilizer group $S_U = \{P_{i, 1}^U \otimes P_{i, 2}^U\}_{i=1}^{4^n}$ for some $n$-qubit Paulis $P_{i, b}^U$.
    The standard purification of the maximally mixed state $\sigma_{\mathrm{mm}}$ is simply the Bell state $|\Phi_0\rangle = \frac{1}{\sqrt{2^n}}\sum_{x} |x\rangle |x\rangle$ with stabilizer group $\{P^\top \otimes P~|~ P \in \{I, X, Y, Z\}^{\otimes n}\}$. 
    
    Notice that
    \begin{equation}
        \begin{split}
            (P_{i, 1}^U \otimes P_{i, 2}^U ) ( I \otimes U) |\Phi_0\rangle &= (I \otimes P_{i, 2}^U U (P_{i, 1}^U)^\top) |\Phi_0\rangle
        \end{split}
    \end{equation}
    by the transpose trick. As $P_{i, 1}^U \otimes P_{i, 2}^U$ is a stabilizer of $\Psi_{\sigma_{\mathrm{mm}}}^U$, we get from the definition of $\Psi_{\sigma_{\mathrm{mm}}}^U$ that
    \begin{equation}
        \begin{split}
            \lr{I \otimes \lr{P_{i, 2}^U U (P_{i, 1}^U)^{\top} - U}} |\Phi_0\rangle = \frac{1}{\sqrt{2^n}} \sum_{x} |x\rangle \otimes \lr{ P_{i, 2}^U U (P_{i, 1}^U)^{\top} - U}|x\rangle = 0\, .
        \end{split}
    \end{equation}
    Thus, we conclude that $ \lr{ P_{i, 2}^U U(P_{i, 1}^U)^{\top} - U}|x\rangle = 0$ holds for all $x \in \Z_2^n$, and hence $P_{i, 2}^U = U ((P_{i, 1}^U)^{\top})^{-1} U^{\dagger}$. Since the transpose of a Pauli equals the original Pauli up to a sign, and the inverse of a Pauli equals the original Pauli up to a phase, we get that the unitary $U$ maps the Pauli $P_{i, 1}^U$ to another Pauli $P_{i, 2}^U$ (up to a sign) for every $i \in [4^n]$.  
    {We claim that the $P_{i, 1}^U$, $i\in [4^n]$, are pairwise distinct modulo phase factors. To see this, suppose that $P_{j, 1}^U = c P_{i, 1}^U$ for some $|c|=1$. Using $P_{k, 2}^U = U ((P_{k, 1}^U)^{\top})^{-1} U^{\dagger}$ for $k=i,j$, we get $P_{j, 2}^U = \overline{c} U((P_{i, 1}^U)^{\top})^{-1} U^\dagger = \overline{c} P_{i, 2}^U$. Hence, we conclude that $P_{j, 1}^U\otimes P_{j, 2}^U = P_{i, 1}^U\otimes P_{i, 2}^U$ and thus $i=j$ since the $4^n$ elements of the stabilizer group of $\Psi_{\sigma_{\mathrm{mm}}}^U$ are distinct. Thus, $\{P_{i, 1}^U\}_{i=1}^{4^n}$ runs over all $4^n$ Paulis modulo phase and, modulo phase, $U$ maps every $n$-qubit Pauli to another Pauli $P_{i, 2}^U$, making $U$ a Clifford unitary.}

    We can therefore assume that $\mc U\subseteq \mathrm{Cliff}_n$. Suppose now that $\eta$ is not uniform over that set. Then, consider the channel $\tilde{\Lambda}$ obtained by conjugating the output of $\Lambda_{\eta, \mc U}$ by a random Clifford over the purifying register. That is, 
    \begin{equation}
        \tilde{\Lambda}(\sigma^{\otimes t}) = \Ex_C \lrq{(I \otimes C)^{\otimes t} \Lambda_{\eta, \mc U}(\sigma^{\otimes t})(I \otimes C^{\dagger})^{\otimes t}}\, ,
    \end{equation}
    where the expectation is over uniformly random Clifford unitaries. It is easy to see that $\tilde{\Lambda} = \Lambda_{\tilde{\eta}, \tilde{\mathcal{U}}}$, where $\tilde{\mathcal{U}}=\mathrm{Cliff}_n$ and $\tilde{\eta}$ is the uniform distribution over $\mathrm{Cliff}_n$:
    \begin{equation}
        \begin{split}
            \tilde{\Lambda}(\sigma^{\otimes t}) &= \Ex_C \lrq{(I \otimes C)^{\otimes t} \Lambda_{\eta, \mc U}(\rho^{\otimes t})(I \otimes C^{\dagger})^{\otimes t}}\\
            &= \Ex_C \Ex_{U \xleftarrow[]{\eta} \mc U} (I \otimes CU)^{\otimes t} \Psi_{\sigma}^{\otimes t} (I \otimes (CU)^{\dagger})^{\otimes t}\, \\
            &= \Ex_C (I \otimes C)^{\otimes t} \Psi_{\sigma}^{\otimes t} (I \otimes C^{\dagger})^{\otimes t}\, ,
        \end{split}
    \end{equation}
    where the last step follows due to left and right invariance of Haar measure on the Clifford group.
    Thus, if $\Lambda_{\eta, \mc U}$ is a structured random purification CPTP map satisfying 1. and 2., then so is $\tilde{\Lambda}$.
    
    So, it suffices to prove the theorem for the case of $\mc U = \mathrm{Cliff}_n$ and for $\eta$ the uniform distribution over that set. Since $n$ is even by assumption, for this case, \Cref{thm:no-go-srp-bell-diagonal-states} yields a subset of mixed $n$-qubit stabilizer states (diagonal in Bell basis) such that $\Lambda_{\eta, \mc U}$ does not simultaneously randomly purify all of them for any $t\geq 4$. 

    The argument can be generalised beyond even $n$. 
    {To show this, we can simply take the states as defined in \Cref{remark:tagged-bell-states}, which are mixed stabilizer states on an odd number of qubits that do not admit a Clifford-random purification channel, as in \Cref{thm:no-go-srp-bell-diagonal-states}. We can then show that there does not exist a random purification channel with the two properties from the theorem statement. The proof is effectively identical to the above proof for the case of even $n$, now starting with the random purification of $\sigma_{\text{mm}} \otimes |0\rangle \langle 0|$.}
\end{proof}

\end{document}